\documentclass[12pt]{article}

\usepackage{amsfonts}
\usepackage{amsmath,amsthm}
\usepackage{amssymb}
\usepackage{color}
\usepackage{graphicx}
\usepackage{longtable}
\usepackage{mathrsfs}
\usepackage{mathtools}
\usepackage{natbib}
\usepackage{stmaryrd}
\usepackage{verbatim}
\usepackage{authblk}

\DeclareMathOperator{\Prob}{\mathbb{P}}
\DeclareMathOperator{\E}{\mathbb{E}}
\newcommand{\del}{\mathrm{\Delta}}
\newcommand{\delhat}{\hat{\mathrm{\Delta}}}
\newcommand{\delsel}{\mathrm{\Delta}_{\mathrm{sel}}}
\newcommand{\delmut}{\mathrm{\Delta}_{\mathrm{mut}}}
\newcommand{\delhatsel}{\hat{\mathrm{\Delta}}_{\mathrm{sel}}}

\newcommand{\vx}{\mathbf{x}}
\newcommand{\vy}{\mathbf{y}}
\newcommand{\vX}{\mathbf{X}}
\newcommand{\va}{\mathbf{a}}
\newcommand{\vA}{\mathbf{A}}
\newcommand{\MSS}{\mathrm{MSS}}
\newcommand{\RMC}{\mathrm{RMC}}

\newcommand{\cM}{\mathcal{M}}
\newcommand{\M}{\mathrm{M}}
\newcommand{\F}{\mathrm{F}}

\newtheorem{theorem}{Theorem}
\newtheorem{corollary}{Corollary}
\newtheorem{lemma}{Lemma}
\newtheorem{proposition}{Proposition}

\theoremstyle{definition}
\newtheorem{definition}{Definition}
\newtheorem*{unityaxiom}{Fixation Axiom}
\newtheorem{assumption}{Assumption}

\date{}

\begin{document}

\title{A mathematical formalism for natural selection with arbitrary spatial and genetic structure}

\author[1,2]{Benjamin Allen}
\author[2]{Alex McAvoy}

\affil[1]{Department of Mathematics\\
			Emmanuel College\\
			Boston, MA 02115}
			
\affil[2]{Program for Evolutionary Dynamics\\
			Harvard University\\
			Cambridge, MA 02138}

\maketitle

\begin{abstract}
	We define a general class of models representing natural selection between two alleles. The population size and spatial structure are arbitrary, but fixed. Genetics can be haploid, diploid, or otherwise; reproduction can be asexual or sexual. Biological events (e.g.~births, deaths, mating, dispersal) depend in arbitrary fashion on the current population state. Our formalism is based on the idea of genetic sites. Each genetic site resides at a particular locus and houses a single allele. Each individual contains a number of sites equal to its ploidy (one for haploids, two for diploids, etc.). Selection occurs via replacement events, in which  alleles in some sites are replaced by copies of others. Replacement events depend stochastically on the population state, leading to a Markov chain representation of natural selection. Within this formalism, we define reproductive value, fitness, neutral drift, and fixation probability, and prove relationships among them. We identify four criteria for evaluating which allele is selected and show that these become equivalent in the limit of low mutation. We then formalize the method of weak selection. The power of our formalism is illustrated with applications to evolutionary games on graphs and to selection in a haplodiploid population.
\end{abstract}

\section{Introduction}
Ever since the Modern Evolutionary Synthesis, mathematics has played an indispensable role in the theory of evolution. Typically, the contribution of mathematics comes in the development and analysis of mathematical models. By representing evolutionary scenarios in a precise way, mathematical modeling can clarify conceptual issues, elucidate underlying mechanisms, and generate new hypotheses.

However,  conclusions from mathematical models must always be interrogated with respect to their robustness. Often, this interrogation takes place in \emph{ad hoc} fashion: Assumptions are relaxed one at a time (sometimes within the original work, sometimes in later works by the same or other authors), until a consensus emerges as to which conclusions are robust and which are merely artifacts.

An alternative approach is to take advantage of the generality made possible by mathematical abstraction. If one can identify a minimal set of assumptions that apply to a broad class of models, any theorem proven from these assumptions will apply to the entire class. Such theorems eliminate the duplicate work of deriving special cases one model at a time. More importantly, the greater the generality in which a theorem is proven, the more likely it is to represent a robust scientific principle. This mathematically general approach has been applied to a number of fields within evolutionary biology, including demographically-structured populations \citep{MetzdeRoos,DiekmannNonlinear,DiekmannLinear,DiekmannPhysiological,Lessard2018}, group- and deme-structured populations \citep{simon2013towards,lehmann2016invasion}, evolutionary game theory \citep{Corina, CorinaMultiple,wu2013dynamic,mcavoy2016structure}, quantitative trait evolution \citep{ChampagnatUnifying,DurinxPhysiological,allen2013adaptive,van2015social}, population extinction and persistence \citep{schreiber:JMB:2010,roth:JMB:2013,roth:JBD:2014,benaim:arxiv:2018}, and many aspects of population genetics \citep{tavare1984line,burger2000mathematical,ewens2004mathematical}.

Currently, there is great theoretical and empirical interest in understanding how the spatial and/or genetic structure of a population influences its evolution. Here, spatial structure refers to the physical layout of the habitat as well as patterns of interaction and dispersal; genetic structure refers to factors such as ploidy, sex ratio, and mating patterns. These factors can affect the rate of genetic change \citep{allen2015molecular,mcavoy2018stationary}, the balance of selection versus drift \citep{ErezGraphs,broom2010evolutionary,adlam2015amplifiers,pavlogiannis2018construction}, and the evolution of cooperation and other social behavior \citep{NowakMay,TaylorHow,RoussetBilliard,Ohtsuki,Taylor,NowakStructured,debarre2014social,pena2016evolutionary,allen2017evolutionary,fotouhi2018conjoining}.

To study the effects of spatial structure in a mathematically general way, \cite{allen2014measures} introduced a class of models with fixed population size and spatial structure. Each model in this class represents competition between two alleles on a single locus in a haploid, asexually-reproducing population. Replacement depends stochastically on the current population state, subject to general assumptions that are compatible with many established models in the literature. For this class, \cite{allen2014measures} defined three criteria for success under natural selection and proved that they coincide when mutation is rare.

Here, we generalize the class of models studied by \cite{allen2014measures} and significantly extend the results. As in 
\cite{allen2014measures}, selection occurs on a single biallelic locus, in a population of fixed size and structure. However, whereas \cite{allen2014measures} assumed haploid genetics, the class introduced here allows for arbitrary genetic structure, including diploid (monoecious or dioecious), haplodiploid, and polyploid genetics. Arbitrary mating patterns are allowed, including self-fertilization. This level of generality is achieved using the notion of \emph{genetic sites}. Each genetic site houses a single allele copy, and each individual contains a number of genetic sites equal to its ploidy. Spatial structure is also arbitrary, in that the patterns of interaction and replacement among individuals are subject only to a minimal assumption ensuring the unity of the population. We also allow for arbitrary mutational bias.

In this class of models, which we present in Section \ref{sec:formalism}, natural selection proceeds by \emph{replacement events}. Replacement events distill all interaction, mating, reproduction, dispersal, and death events into what ultimately matters for selection---namely, which alleles are replaced by copies of which others. Replacement events occur with probability depending on the current population state, according to a given \emph{replacement rule}. The replacement rule implicitly encodes all relevant aspects of the spatial and genetic structure.

The replacement rule, together with the mutation rate and mutational bias, define an \emph{evolutionary Markov chain} representing natural selection. Basic results on the asymptotic behavior of the evolutionary Markov chain are established in Section \ref{sec:StationarityFixation}.

In Section \ref{sec:selection}, we turn to the question of identifying which of two competing alleles is favored by selection. We compare four criteria: one based on fixation probabilities, one based on time-averaged frequency, and two based on expected frequency change due to selection. We prove (Theorem \ref{thm:rhodelsel}) that these coincide in the limit of low mutation, thereby generalizing the main result of \cite{allen2014measures}.

Sections \ref{sec:RVfit} and \ref{sec:neutral} explore the closely-related concepts of reproductive value, fitness, and neutral drift. We define these notions in the context of our formalism and prove connections among them. Interestingly, to define reproductive value requires an additional assumption that does not necessarily hold for all models; thus, the concept of reproductive value may not be as general as is sometimes thought. We also provide a new proof for the recently-observed principle \citep{maciejewski2014reproductive,allen2015molecular} that the reproductive value of a genetic site is proportional to the fixation probability, under neutral drift, of a mutation arising at that site.

We next turn to weak selection (Section \ref{sec:weaksel}), meaning that the alleles in question have only a small effect on reproductive success. Mathematically, weak selection can be considered a perturbation of neutral drift. Using this perturbative approach, one can obtain closed-form conditions for success under weak selection for models that would be otherwise intractable. This approach has fruitfully been applied in a great many contexts \citep{TaylorHow,RoussetBilliard,leturque2002dispersal,NowakFinite,Ohtsuki,LessardFixation,Taylor,AntalPhenotype,Corina,chen2013sharp,debarre2014social,durrett2014spatial,tarnita2014measures,van2015social,allen2017evolutionary}. Our second main result (Theorem \ref{thm:delhatselweak}) formalizes this weak-selection approach for our class of models. It asserts that, to determine whether an allele is favored under weak selection, one can take the expectation of a quantity describing selection over a probability distribution that pertains to neutral drift. The usefulness of this result stems from the fact that many evolutionary models become much simpler in the case of neutral drift.

The bulk of this work adopts a ``gene's-eye view," in that the analysis is conducted at the level of genetic sites. In Section \ref{sec:ind}, we reframe our results using quantities that apply at the level of the individual. This reframing again requires additional assumptions, such as fair meiosis. Without these additional assumptions, natural selection cannot be characterized solely in terms of individual-level quantities.

We illustrate the power of our formalism with two examples (Section \ref{sec:examples}). The first is a model of evolutionary games on an arbitrary weighted graph. For this model, we recover recent results of \cite{allen2017evolutionary}, using only results proven in this work. The second is a haplodiploid population model in which a mutation may have different selective effects in males and females. We obtain a simple condition to determine whether such a mutation is favored under weak selection.

Although our formalism is quite general in some respects, it still makes a number of simplifying assumptions. For example, we assume a population of fixed size in a constant environment, but real-world populations are subject to demographic fluctuations and ecological feedbacks, which may have significant consequences for their evolution \citep{Dieckmann,Metz,Geritz,pelletier2007evolutionary,wakano2009spatial,schoener2011newest,constable2016demographic,chotibut2017population}. Other limitations arise from our assumptions of fixed spatial structure, single-locus genetics, and trivial demography. Section \ref{sec:discussion} discusses these limitations and the prospects for extending beyond them.

\section{Class of models for natural selection}\label{sec:formalism}
We consider a class of models representing selection, on a single biallelic locus, in a population with arbitrary---but fixed---spatial and genetic structure. Each model within this class is represented by a set of genetic sites (partitioned into individuals), a replacement rule, a mutation probability, and a mutational bias. In this section, we introduce each of these ingredients in detail and discuss how they combine to form a Markov chain representing natural selection. A glossary of our notation is provided in Table \ref{notationtable}.

\subsection{Sites and individuals}\label{sec:sites}
We represent arbitrary spatial and genetic structure by using the concept of \emph{genetic sites} (Figs.~1A, 2). Each genetic site corresponds to a particular locus, on a single chromosome, within an individual. Since we consider only single-locus traits, each individual has a number of sites equal to its ploidy (e.g.~one for haploids, two for diploids). 

The genetic sites in the population are represented by a finite set $G$. The individuals are represented by a finite set $I$. To each individual $i \in I$, there corresponds a set of genetic sites $G_i \subseteq G$ residing in $i$. The collection of these sets, $\{G_i \}_{i \in I}$, forms a partition of $G$. We use the equivalence relation $\sim$ to indicate that two sites reside in the same individual; thus, $g \sim h$ if and only if $g,h \in G_i$ for some $i \in I$.

The total number of sites is denoted $n \coloneqq \left| G\right|$, and the total number of individuals is denoted $N \coloneqq \left| I\right|$. The ploidy of individual $i \in I$ is denoted $n_i \coloneqq \left| G_i\right|$; for example, $n_i=2$ if $i$ is diploid. The total number of sites is equal to the total ploidy across all individuals: $\sum_{i \in I} n_i=n$.

For a particular model within the class defined here, each individual may be labeled with additional information. For example, each individual may be designated as male or female and/or could be understood as occupying a particular location. However, these details are not explicitly represented in our formalism. In particular, we do not specify any representation of spatial structure (lattice, graph, metapopulation, etc.), although our formalism is compatible with all of these. Instead, all relevant aspects of spatial and genetic structure are implicitly encoded in the replacement rule (see Section \ref{sec:replacement} below). The spatial and genetic structure are considered fixed, in the sense that the roles of individuals and genetic sites do not change over time.

\subsection{Alleles and states} \label{sec:states}
There are two competing alleles, $a$ and $A$. Each genetic site holds a single allele copy. The allele currently occupying site $g \in G$ is indicated by the variable $x_g \in \left\{0,1\right\}$, with 0 corresponding to $a$ and 1 corresponding to $A$. The overall population state is represented by the vector $\vx \coloneqq \left(x_g\right)_{g \in G}$, which specifies the allele ($a$ or $A$) occupying each genetic site. The set of all possible states is denoted $\left\{0,1\right\}^G$.

It will sometimes be convenient to label a state by the subset of sites that contain the $A$ allele. Thus, for any subset $S \subseteq G$, we let $\mathbf{1}_S \in \left\{0,1\right\}^G$ denote the state in which sites in $S$ have allele $A$, and sites not in $S$ have allele $a$. That is, the state $\mathbf{1}_S$ is defined by
\begin{align}
	\left( \mathbf{1}_S \right)_g= \begin{cases}
		1 & g \in S , \\ 0 & g \notin S. \end{cases}
\end{align}

Of particular interest are the \emph{monoallelic states} $\va\coloneqq\mathbf{1}_\emptyset$, in which only allele $a$ is present; and $\vA\coloneqq\mathbf{1}_G$, in which only allele $A$ is present.

\subsection{Replacement}\label{sec:replacement}

\begin{figure}
	\centering
	\includegraphics[width=\textwidth]{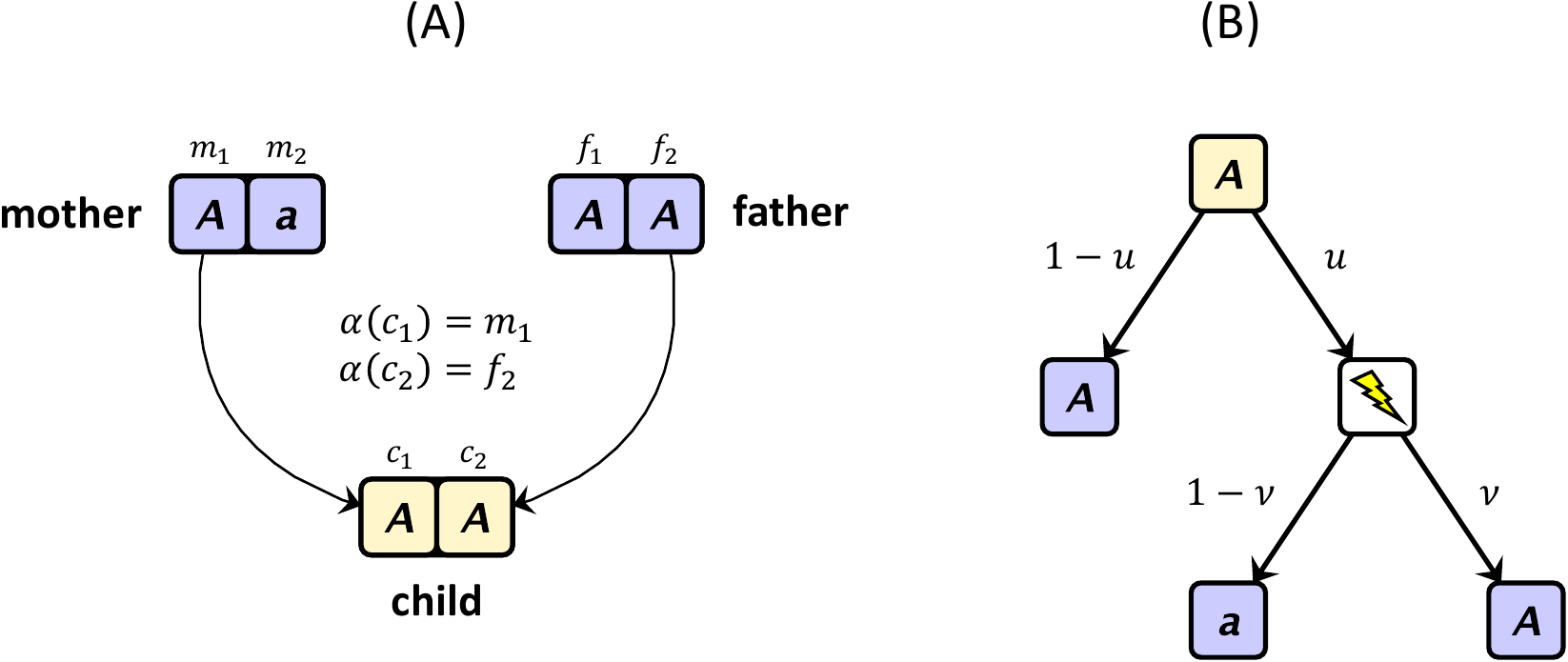}
	\caption{(A) The parentage mapping, $\alpha$, in the case of a diploid, sexually-reproducing population. In a diploid population, each individual contains two genetic sites. Here, site $c_1$ in the child inherits the allele from site $m_1$ in the mother ($\alpha(c_1)=m_1$), while site $c_2$ in the child inherits the allele from site $f_2$ in the father ($\alpha(c_2)=f_2$). Note that, although arrows are drawn from parent to child, the parentage map $\alpha$ is from child to parent. (B) Mutations are resolved as follows: With probability $1-u$, there is no mutation and the allele remains the parental type ($A$ in this case). With probability $u$, the allele mutates (lightning bolt) and becomes either $A$ (probability $\nu$) or $a$ (probability $1-\nu$). \label{fig:mutation}}
\end{figure}

\begin{figure}
	\centering
	\includegraphics[width=\textwidth]{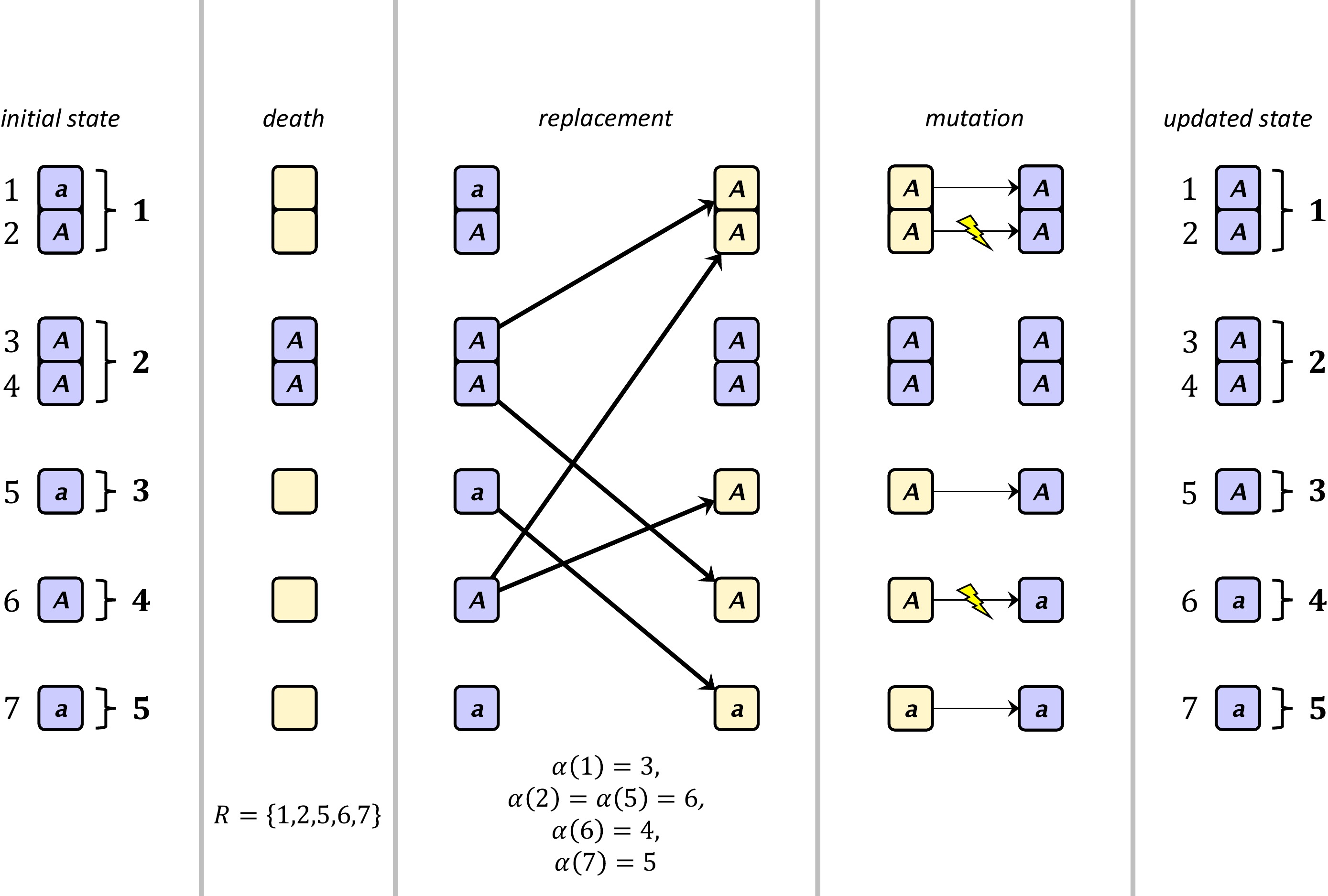}
	\caption{One complete update step in the evolutionary Markov chain. An example population is pictured with two diploid and three haploid individuals. Genetic sites are indicated by numerals to the left of the site, and the individuals in which these sites reside are labeled by bold numerals. First, a replacement event, $\left(R,\alpha\right)$, is chosen according to the distribution $\left\{p_{\left(R,\alpha\right)}\left(\vx\right)\right\}_{\left(R,\alpha\right)}$. In this case, the replaced set, $R$, is shown in yellow. Note that arrows are drawn from parents to children, but the parentage map, $\alpha$, is from child to parent. For every genetic site that is replaced under this event (yellow), the replicated allele is then subjected to possible mutation, resulting in a new state. \label{fig:replacement}}
\end{figure}

Natural selection proceeds by \emph{replacement events}, wherein some individuals are replaced by the offspring of others (Figs.~\ref{fig:mutation}A, \ref{fig:replacement}). We let $R\subseteq G$ denote the set of genetic sites that are replaced in such an event. For example, if only a single individual $i \in I$ dies, then $R=G_i$. If the entire population is replaced, then $R=G$.

The alleles in the sites in $R$ are then replaced by alleles in new offspring. Each new offspring inherits (possibly mutated) copies of alleles from its parents. The parentage of new alleles is recorded in a set mapping $\alpha:R \to G$. For each replaced site $g \in R$, $\alpha\left(g\right)$ indicates the parental site from which $g$ inherits its new allele. In other words, the new allele in $g$ is derived from a parent copy that (in the previous time-step) occupied site $\alpha\left(g\right)$. In haploid asexual models, $\alpha\left(g\right)$ simply indicates the parent of the new offspring in $g$. In models with sexual reproduction, $\alpha$ identifies not only the parents of each new offspring but also which allele were inherited from each parent (Fig.~\ref{fig:mutation}A).

Overall, a replacement event is represented by the pair $(R, \alpha)$, where $R \subseteq G$ is the set of replaced positions and $\alpha:R \to G$ is the parentage mapping. Any pair $\left(R, \alpha\right)$ with $R \subseteq G$ and $\alpha:R \to G$ can be considered a potential replacement event. Whether or not a given replacement event is possible in a given state, and how likely it is to occur, depends on the model in question. The probability that a given replacement event $\left(R,\alpha\right)$ occurs in state $\vx \in \left\{0,1\right\}^G$ is denoted $p_{\left(R, \alpha\right)}\left(\vx\right)$; these satisfy $\sum_{\left(R,\alpha\right)} p_{\left(R, \alpha\right)} \left(\vx\right) =1$ for each fixed $\vx$. The probabilities $\left\{p_{\left(R, \alpha\right)} \left(\vx\right) \right\}_{\left(R,\alpha\right)}$, as functions of $\vx$, are collectively called the \emph{replacement rule}.

All biological events such as births, deaths, mating, dispersal and interaction, and all aspects of spatial and genetic structure, are represented implicitly in the replacement rule. For example, in a model of a diploid population with nonrandom mating, the replacement rule encodes mating probabilities as well as the laws of Mendelian inheritance. In a model of a spatially-structured population with social interactions (see Section \ref{sec:games}), the replacement rule encodes interaction patterns, as well as the effects of interactions on births and deaths. From these biological details, the replacement rule distills what ultimately matters for selection: the transmission and inheritance of alleles.

\subsection{Mutation}\label{sec:mutation}
Each replacement of an allele provides an opportunity for mutation (Fig.~\ref{fig:mutation}B). Mutation is described by two parameters: (i) the \emph{mutation probability}, $0 \leqslant u \leqslant 1$, which is the probability that a given allele copy in a new offspring is mutated from its parent; and (ii) the \emph{mutational bias}, $0 < \nu < 1$, which is the probability that such a mutation results in $A$ rather than $a$.

In each time-step, after the replacement event $\left(R,\alpha\right)$ has been chosen, mutations are resolved and the new state, $\vx'$, is determined as follows. For each replaced site $g \in R$, one of three outcomes occurs:
\begin{itemize}
	\item With probability $1-u$, there is no mutation, and site $g$ inherits the allele of its parent: $x_g' = x_{\alpha(g)}$,
	\item With probability $u\nu$, a mutation to $A$ occurs, and $x_g'=1$,
	\item With probability $u(1-\nu)$, a mutation to $a$ occurs, and $x_g'=0$.
\end{itemize}
Mutation events are assumed to be independent across replaced sites and across time. Each site that is not replaced retains its current allele: $x_g'=x_g$ for all $g \notin R$. In this way, the updated state, $\vx'$, is determined.

\subsection{The evolutionary Markov chain} \label{sec:Markov}
Overall, from a given state $\vx$, first a replacement event is chosen according to the probabilities $\left\{p_{\left(R, \alpha\right)} \left(\vx\right)\right\}_{\left(R,\alpha\right)}$, and then mutations are resolved as described in Section \ref{sec:mutation}. This update leads to a new state $\vx'$, and the process then repeats (Fig.~\ref{fig:replacement}). This process defines a Markov chain $\cM$ on $\left\{0,1\right\}^G$, which we call the \emph{evolutionary Markov chain}. The evolutionary Markov chain is completely determined by the replacement rule $\left\{p_{\left(R, \alpha\right)} \left(\vx\right) \right\}_{\left(R,\alpha\right)}$, the mutation rate $u$, and the mutational bias $\nu$. We denote the transition probability from state $\vx$ to state $\vy$ in $\cM$ by $P_{\vx \to \vy}$.

\subsection{Fixation Axiom}
In order for the population to function as a single evolving unit, it should be possible for an allele to sweep to fixation. To state this principle formally, we introduce some new notation. For a given replacement event, $\left(R, \alpha\right)$, let $\tilde{\alpha}: G \to G$ be the mapping that coincides with $\alpha$ on elements of $R$ and coincides with the identity otherwise:
\begin{align}
	\tilde{\alpha}\left(g\right) = \begin{cases} \alpha\left(g\right) & g \in R , \\ g & g \notin R. \end{cases}
\end{align}
In words, $\tilde{\alpha}$ maps to the parent of each replaced site, and to the site itself for those not replaced.

We now formalize the notion of population unity as an axiom:
\begin{unityaxiom}
	There exists a genetic site $g \in G$, a positive integer $m$, and a finite sequence $\{ (R_k, \alpha_k) \}_{k=1}^m$ of replacement events, such that
	\renewcommand{\labelenumi}{(\alph{enumi})}
	\begin{enumerate}
		\item $p_{(R_k, \alpha_k)}(\vx)>0$ for all $k \in \{1, \ldots, m\}$ and all $\vx \in \{0,1\}^G$, 
		\item $g \in R_k$ for some $k \in \{1, \ldots, m\}$,
		\item For each $h \in G$, $\tilde{\alpha}_1 \circ \tilde{\alpha}_2\circ \cdots \circ \tilde{\alpha}_m (h) = g.$
	\end{enumerate}
\end{unityaxiom}

In words, there should be at least one genetic site $g \in G$ that can eventually spread its contents throughout the population, such that all sites ultimately trace their ancestry back to $g$. Part (b) is included to guarantee that no site is eternal (otherwise no evolution would occur). The Fixation Axiom ensures that the population evolves as a single unit, rather than (for example) being comprised of isolated subpopulations with no gene flow among them. We regard this axiom as a defining property of our class of models. 

\subsection{Relation to \cite{allen2014measures}}
Our formalism extends the class of models introduced by \cite{allen2014measures}, which considered only haploid populations with asexual reproduction, to populations with arbitrary genetic structure. Despite the differences in genetics, the two classes are very similar in their formal structure. Indeed, one can ``forget" the partition of genetic sites into individuals and instead consider the population as consisting of haploid asexual replicators. With this perspective, the results of \cite{allen2014measures} can be applied at the level of genetic sites rather than individuals. 

Beyond genetic structure, our current formalism generalizes that of \cite{allen2014measures} in three ways. First, whereas \cite{allen2014measures} assumed unbiased mutation, we consider here arbitrary mutational bias, $0<\nu<1$. Second, \cite{allen2014measures} assumed that the total birth rate is constant over states; here this assumption is deferred until Section \ref{sec:RV}, by which point we have already established a number of fundamental results. Third, our Fixation Axiom generalizes its analogue in \cite{allen2014measures} (there labeled Assumption 2), which required that fixation be possible from every site. Here, we only require fixation to be possible from at least one site. The current formulation allows for ``dead end" sites, such as those in sterile worker insects, which not allowed in the formalism of \cite{allen2014measures}.

Despite the increase in generality, some proofs from \cite{allen2014measures} carry over to the current formalism with little or no modification. We will not repeat proofs from \cite{allen2014measures} here unless they need to be modified significantly.

\section{Stationarity and Fixation}\label{sec:StationarityFixation}
In this section, we establish fundamental results regarding the asymptotic behavior of the evolutionary Markov chain. We also define fixation probability and introduce probability distributions that characterize the frequency with which states arise under natural selection.

\subsection{Demographic variables} \label{sec:demographic}
We first introduce the following variables as functions of the state $\vx \in \left\{0,1\right\}^G$. The frequency of the allele $A$ is denoted by $x$:
\begin{equation}
	\label{eq:freqdef}
	x \coloneqq \frac{1}{n} \sum_{g \in G} x_g.
\end{equation}
The (marginal) probability that the allele in site $g \in G$ transmits a copy of itself to site $h \in G$ over the next transition is denoted $e_{gh}\left(\vx\right)$:
\begin{equation}
	e_{gh}\left(\vx\right) \coloneqq \sum_{\substack{\left(R, \alpha\right) \\ \alpha\left(h\right) =g}} p_{\left(R, \alpha\right)} \left(\vx\right).
\end{equation}
The expected number of copies that the allele in $g$ transmits, which we call the \emph{birth rate} of site $g$ in state $\vx$, can be calculated as:
\begin{equation}
	b_g \left(\vx\right) \coloneqq \sum_{h\in G} e_{g h} \left(\vx\right) = \sum_{\left(R, \alpha\right)} p_{\left(R, \alpha\right)} \left(\vx\right) \, \left|\alpha^{-1}\left(g\right)\right| .
\end{equation}
The probability that the allele in $g$ is replaced, which we call the \emph{death probability} of site $g$ in state $\vx$, can be calculated as:
\begin{equation}
	\label{eq:ddef}
	d_g \left(\vx\right) \coloneqq \sum_{h \in G} e_{hg} \left(\vx\right) = \sum_{\substack{\left(R, \alpha\right) \\ g \in R}} p_{\left(R, \alpha\right)} \left(\vx\right) .
\end{equation}
The Fixation Axiom guarantees that $d_g \left(\vx\right) >0$ for all $g \in G$ and $\vx \in \left\{0,1\right\}^G$.

The total birth rate in state $\vx$ is denoted $b\left(\vx\right)$. Since the population size is fixed, $b\left(\vx\right)$ also gives the expected number of deaths:
\begin{equation}
	b\left(\vx\right) \coloneqq \sum_{g \in G} b_{g} \left(\vx\right) = \sum_{g \in G} d_{g} \left(\vx\right) = \sum_{g,h \in G} e_{g h} \left(\vx\right).
\end{equation}

\subsection{The mutation-selection stationary distribution} \label{sec:MSS}
When mutation is present ($u > 0$), the evolutionary Markov chain is ergodic (aperiodic and positive recurrent; Theorem 1 of \citealp{allen2014measures}). In this case, the evolutionary Markov chain has a unique stationary distribution called the \emph{mutation-selection stationary distribution}, or \emph{MSS distribution} for short. For any state function $f\left(\vx\right)$, its time-averaged value converges almost surely, as time goes to infinity, to its expectation under this distribution:
\begin{align}
	\lim_{T \to \infty} \frac{1}{T} \sum_{t=0}^{T-1} f\left(\vX \left(t\right)\right) = \E_\MSS\left[f\right] \quad \text{almost surely}.
\end{align}

We denote the probability of state $\vx$ in the MSS distribution by $\pi_\MSS\left(\vx\right) \coloneqq \Prob_\MSS\left[\vX=\vx\right]$. The MSS distribution is uniquely determined by the system of equations
\begin{subequations}
	\label{eq:MSSsystem}
	\begin{align}
		\label{eq:MSSrecur}
		\pi_\MSS\left(\vx\right) & = \sum_{\vy \in \{0,1\}^G} \pi_\MSS\left(\vy\right) P_{\vy \to \vx} , \\
		\label{eq:MSSsum}
		\sum_{\vx \in \left\{0,1\right\}^G} \pi_\MSS\left(\vx\right) & = 1.
	\end{align}
\end{subequations}

\subsection{Fixation probability} \label{sec:fixprob}
When there is no mutation ($u=0$), the monoallelic states $\va$ and $\vA$ are absorbing, and all other states are transient (Theorem 2 of \citealp{allen2014measures}). Thus, from any initial state, the evolutionary Markov chain converges, almost surely as $t \to \infty$, to one of the two monoallelic states. We say that the population has become \emph{fixed for allele $a$} if the state converges to $\va$, and \emph{fixed for allele $A$} if the state converges to $\vA$.

The \emph{fixation probability} of an allele is informally defined as the probability that it becomes fixed when starting from a single copy. A precise definition, however, must take into account that the fate of a mutant allele can depend on the site in which it arises \citep{allen2014measures,maciejewski2014reproductive,adlam2015amplifiers,allen2015molecular,chen2016fixation}. Since each replacement provides an independent opportunity for mutation, new mutations arise in proportion to the rate at which a site is replaced \citep{allen2014measures}. Thus, in state $\va$, $A$ mutations arise in site $g$ at a rate proportional to $d_g\left(\va\right)$, while in state $\vA$, $a$ mutations arise in site $g$ at a rate proportional to $d_g\left(\vA\right)$. The probability of multiple $A$ mutations arising in state $\va$, or multiple $a$ mutations arising in state $\vA$, is of order $u^2$ as $u \to 0$. We formalize these observations as a lemma:

\begin{lemma}
	\label{lem:mutantarise}
	\begin{subequations}\label{eq:mutantarise}
		\begin{align}
			P_{\va \to \vx} & = 
			\begin{cases}
				1 - u\nu \, b\left(\va\right) + \mathcal{O}\left(u^2\right) & \text{if $\vx = \va$}, \\
				u\nu \, d_g\left(\va\right) + \mathcal{O}\left(u^2\right) & \text{if $\vx = \mathbf{1}_{\left\{g\right\}}$ for some $g \in G$}, \\
				\mathcal{O}\left(u^2\right) & \text{otherwise};
			\end{cases} \\
			P_{\vA \to \vx} & = 
			\begin{cases}
				1 - u \left(1-\nu\right) \, b\left(\vA\right) + \mathcal{O}\left(u^2\right) & \text{if $\vx = \vA$}, \\
				u \left(1-\nu\right) \, d_g\left(\vA\right) + \mathcal{O}\left(u^2\right) & \text{if $\vx = \mathbf{1}_{G \setminus \left\{g\right\}}$ for some $g \in G$}, \\
				\mathcal{O}\left(u^2\right) & \text{otherwise}.
			\end{cases}
		\end{align}
	\end{subequations}
\end{lemma}

The proof is a minor variation on the proof of Lemma 3 in \cite{allen2014measures} and is therefore omitted.

Lemma \ref{lem:mutantarise} motivates the following definitions (from \citealp{allen2014measures}), describing the relative likelihoods of initial states when a mutant first arises under rare mutation:

\begin{definition}\label{def:mutapper}
	The \emph{mutant appearance distribution for allele $A$} is a probability distribution on $\left\{0,1\right\}^G$ defined by
	\begin{align}
		\label{eq:mutA}
		\mu_A \left(\vx\right) \coloneqq \begin{cases} 
			\frac{d_g\left(\va\right)}{b\left(\va\right)} & \text{if $\vx = \mathbf{1}_{\left\{g\right\}}$ for some $g \in G$} , \\
			0 & \text{otherwise.}
		\end{cases}
	\end{align}
	Similarly, the \emph{mutant appearance distribution for allele $a$} is a probability distribution on $\left\{0,1\right\}^G$ defined by
	\begin{align}
		\label{eq:muta}
		\mu_a \left(\vx\right) \coloneqq \begin{cases} 
			\frac{d_g\left(\vA\right)}{b\left(\vA\right)} & \text{if $\vx = \mathbf{1}_{G \setminus \left\{g\right\}}$ for some $g \in G$} , \\
			0 & \text{otherwise.}
		\end{cases}
	\end{align}
\end{definition}

Taking these mutant appearance distributions into account, \cite{allen2014measures} defined the overall fixation probabilities of $A$ and $a$ as follows:

\begin{definition}\label{def:meanFixationProbabilities}
	The \emph{fixation probability of $A$}, denoted $\rho_A$, is defined as
	\begin{equation}
		\rho_A \coloneqq  \sum_{\vx \in \left\{0,1\right\}^G} \mu_A\left(\vx\right) \left( \lim_{t \to \infty} P^{\left(t\right)}_{ \vx \to \vA} \right).
	\end{equation}
	Similarly, the \emph{fixation probability of $a$}, denoted $\rho_a$, is defined as
	\begin{equation}
		\rho_a \coloneqq \sum_{\vx \in \left\{0,1\right\}^G} \mu_a\left(\vx\right) \left( \lim_{t \to \infty} P^{\left(t\right)}_{ \vx \to \va} \right).
	\end{equation}
\end{definition}

Above, $P^{\left(t\right)}_{ \vx \to \vy}$ denotes the probability of transition from state $\vx$ to state $\vy$ in $t$ steps. The Fixation Axiom guarantees that there is at least one site $g$ for which $\lim_{t \to \infty} P^{\left(t\right)}_{ \mathbf{1}_{\left\{g\right\}} \to \vA}$,  $\lim_{t \to \infty} P^{\left(t\right)}_{  \mathbf{1}_{G \setminus \left\{g\right\}} \to \va}$, $d_g\left(\va\right)$, and $d_g\left(\vA\right)$ are all positive. It follows that $\rho_A$ and $\rho_a$ are both positive.

\subsection{The limit of rare mutation}\label{sec:lowmutation}
We now consider the limit of low mutation for a fixed replacement rule, $\left\{p_{\left(R, \alpha\right)} \left(\vx\right) \right\}_{\left(R,\alpha\right)}$, and mutational bias $\nu$. There is an elegant relationship between the fixation probabilities and the limiting MSS distribution:

\begin{theorem}\label{thm:lowu}
	Fix a replacement rule $\left\{p_{\left(R, \alpha\right)} \left(\vx\right) \right\}_{\left(R,\alpha\right)}$ and a mutational bias $\nu$. Then for each state $\vx \in \left\{0,1\right\}^G$, $\lim_{u \to 0} \pi_\MSS\left(\vx\right)$ exists and is given by 
	\begin{equation}
		\label{eq:pAndRho}
		\lim_{u \to 0} \pi_\MSS\left(\vx\right) =
		\begin{cases} 
			\displaystyle \frac{\nu b \left(\va\right)\rho_{A}}{ \nu b\left(\va\right) \rho_{A}+\left(1-\nu\right) b \left(\vA\right)\rho_{a}} 
			& \text{for $\vx = \vA$} , \\[5mm]
			\displaystyle \frac{\left(1-\nu\right) b\left(\vA\right)\rho_{a}}{\nu b \left(\va\right)\rho_{A}+\left(1-\nu\right) b (\vA)\rho_{a}} 
			& \text{for $\vx = \va$} , \\[5mm]
			0 & \text{for $\vx \notin \left\{\va, \vA\right\}$}.
		\end{cases}
	\end{equation}
	Above, $\rho_A$ and $\rho_a$ are the fixation probabilities for this replacement rule when $u=0$.
\end{theorem}

Intuitively, Theorem \ref{thm:lowu} states that as $u\to 0$, the MSS distribution becomes concentrated on the monoallelic states $\vA$ and $\va$, with probabilities determined by the relative rates of transit, $\nu b \left(\va\right)\rho_{A}$ and $\left(1-\nu\right) b \left(\vA\right)\rho_{a}$. Theorem \ref{thm:lowu} result generalizes Theorem 6 of \cite{allen2014measures} and a result of \cite{van2015social}, both of which apply to the special case $\nu=1/2$ and $b\left(\va\right) =b\left(\vA\right)$.

We will prove Theorem \ref{thm:lowu} using the principle of state space reduction. Let $\mathcal{A}$ be a finite  Markov chain and let $S$ be a nonempty subset of the states of $\mathcal{A}$. (In proving Theorem \ref{thm:lowu} we will use $\mathcal{A}=\cM$ and $S=\{\va, \vA\}$.) For any states $s,s' \in S$, let $Q_{s \to s'}$ be the probability that, from initial state $s$, the next visit to $S$ occurs in state $s'$. We define a reduced Markov chain $\mathcal{A}_{|S}$ with set of states $S$ and transition probabilities $Q_{s \to s'}$. The following standard result (e.g.~Theorem 6.1.1~of \citealp{kemeny1960finite}) shows that stationary distributions for the original and reduced Markov chains are compatible in a simple way:

\begin{theorem}
	\label{thm:statereduction}
	Let $\mathcal{A}$ be finite Markov chain with a unique stationary distribution, $\pi_\mathcal{A}$, and let $S$ be a nonempty subset of states of $\mathcal{A}$. Then, the reduced Markov chain, $\mathcal{A}_{|S}$, has a unique stationary distribution, $\pi_{\mathcal{A}_{|S}}$, which is given by conditioning the stationary distribution $\pi_\mathcal{A}$ on the event $S$:
	\begin{equation}
		\pi_{\mathcal{A}_{|S}}\left(s\right) \coloneqq \frac{\pi_\mathcal{A}\left(s\right)}{\sum_{s'\in S} \pi_\mathcal{A}\left(s'\right)} .
	\end{equation}
\end{theorem}

\begin{proof}[Proof of Theorem \ref{thm:lowu}] The limits $\lim_{u \to 0} \pi_\MSS\left(\vx\right)$ exist for each $\vx \in \left\{0,1\right\}^G$ since each $\pi_\MSS\left(\vx\right)$ is a bounded, rational function of $u$ (see Lemma 1 of \citealp{allen2014measures}). Since $\pi_\MSS$ satisfies Eq.~\eqref{eq:MSSsystem} for each $u>0$, it also satisfies Eq.~\eqref{eq:MSSsystem} in the limit $u\rightarrow 0$. Therefore, $\lim_{u \to 0} \pi_\MSS$ is a stationary distribution for the mutation-free ($u=0$) evolutionary Markov chain, $\cM$. Since all states other than $\va$ and $\vA$ are transient when $u=0$, they must have zero probability in any stationary distribution; therefore, $\lim_{u \to 0} \pi_\MSS\left(\vx\right)=0$ for $\vx \notin \left\{\va, \vA\right\}$.
	
	To determine the limiting values of $\pi_\MSS\left(\va\right)$ and $\pi_\MSS\left(\vA\right)$, we temporarily fix some $u>0$ and consider the reduction of $\cM$ to the set of states $\left\{ \va, \vA \right\}$. By Theorem \ref{thm:statereduction}, the reduced Markov chain $\cM_{|\left\{ \va, \vA \right\}}$ has a unique stationary distribution, $\pi_{\cM_{|\left\{ \va, \vA \right\}}}$, satisfying
	\begin{equation}
		\label{eq:piratio1}
		\frac{\pi_{\cM_{| \left\{ \va, \vA \right\}}}\left(\vA\right)}{\pi_{\cM_{| \left\{ \va, \vA \right\}}}\left(\va\right)} =\frac{\pi_\MSS\left(\vA\right)}{\pi_\MSS\left(\va\right)}.
	\end{equation}
	Let $Q_{\va \to \vA}$ and $Q_{\vA \to \va}$ denote the transition probabilities in $\cM_{| \left\{ \va, \vA \right\}}$. Eq.~\eqref{eq:piratio1} and the stationarity of $\pi_{\cM_{|\left\{ \va, \vA \right\}}}$ imply that 
	\begin{equation}
		\label{eq:piratio}
		\frac{\pi_\MSS\left(\vA\right)}{\pi_\MSS\left(\va\right)} = \frac{Q_{\va \to \vA}}{Q_{\vA \to \va}} .
	\end{equation}
	
	We note that $Q_{\va \to \vA}$ equals the probability, in $\cM$ with initial state $\va$, of (i) leaving $\va$ in the initial step, and (ii) subsequently visiting $\vA$ before revisiting $\va$. Step (i) occurs with probability $u \nu b\left(\va\right) + \mathcal{O}\left(u^2\right)$ as $u \to 0$, while step (ii) occurs with probability $\rho_A + \mathcal{O}\left(u\right)$. Thus, overall, we have the expansion
	\begin{equation}
		Q_{\va \to \vA} = u \nu b\left(\va\right) \rho_A + \mathcal{O}\left(u^2\right) \qquad \left(u \to 0\right) .
	\end{equation}
	Similarly, we have
	\begin{equation}
		Q_{\vA \to \va} = u \left(1-\nu\right) b\left(\vA\right) \rho_a + \mathcal{O}\left(u^2\right) \qquad \left(u \to 0\right) .
	\end{equation}
	Substituting these expansions in Eq.~\eqref{eq:piratio} and taking the limit as $u \to 0$ yields
	\begin{equation}\label{eq:Pratio}
		\lim_{u \to 0} \frac{\pi_\MSS\left(\vA\right)}{\pi_\MSS\left(\va\right)} = \frac{\nu b\left(\va\right) \rho_A}{\left(1-\nu\right) b\left(\vA\right) \rho_a} .
	\end{equation}
	The desired result now follows from the fact that, since $\lim_{u \to 0} \pi_\MSS\left(\vx\right) =0$ for $\vx \notin \left\{\va, \vA\right\}$, we must have $\lim_{u \to 0} \pi_\MSS\left(\va\right) + \lim_{u \to 0} \pi_\MSS\left(\vA\right) = 1$.
\end{proof}

Theorem \ref{thm:lowu} implies that the stationary probabilities $\pi_{\MSS}\left(\vx\right)$ extend to smooth functions of the mutation rate $u$ on the interval $0 \leqslant u \leqslant 1$, with the values at $u=0$ defined according to Eq.~\eqref{eq:pAndRho}. As mentioned in the proof, the limiting probabilities in Eq.~\eqref{eq:pAndRho} comprise a stationary distribution for the evolutionary Markov chain with $u=0$. However, this stationary distribution is not unique---indeed, any probability distribution concentrated entirely on states $\vA$ and $\va$ is stationary for $u=0$. We can achieve uniqueness at $u=0$ by augmenting Eq.~\eqref{eq:MSSsystem} by an additional equation: 
\begin{equation}
	\label{eq:augment}
	\pi_\MSS\left(\vA\right) = 
	\begin{cases}
		\displaystyle \frac{Q_{\va \to \vA}}{Q_{\vA \to \va}} \pi_\MSS\left(\va\right) & 0 < u \leqslant 1 , \\[5mm]
		\displaystyle \frac{\nu b\left(\va\right) \rho_A}{\left(1-\nu\right) b\left(\vA\right) \rho_a} \pi_\MSS\left(\va\right) & u=0.
	\end{cases}
\end{equation}
The system of linear equations \eqref{eq:MSSsystem} and \eqref{eq:augment} has a unique solution that varies smoothly with $u$ for $0 \leqslant u \leqslant 1$, coincides with $\pi_{\MSS}\left(\vx\right)$ for $0<u\leqslant 1$, and coincides with the right-hand side of Eq.~\eqref{eq:pAndRho} for $u=0$. We will make use of these observations in Section \ref{sec:weaksel}.

Alternatively, Theorem \ref{thm:lowu} can be proven using Theorem 2 of \citet{fudenberg:JET:2006}, which implies that as $u\rightarrow 0$, the vector $\left(\pi_{\MSS}\left(\vA\right) ,\pi_{\MSS}\left(\va\right)\right)$ converges to the stationary distribution of the embedded Markov chain on the absorbing states, $\vA$ and $\va$. The transition matrix of this embedded Markov chain is
\begin{align}
	\begin{pmatrix}
		1-\gamma \nu b\left(\mathbf{a}\right)\rho_{A} & \gamma \nu b\left(\mathbf{a}\right)\rho_{A} \\
		\gamma\left(1- \nu\right) b\left(\mathbf{A}\right)\rho_{a} & 1-\gamma\left(1- \nu\right) b\left(\mathbf{A}\right)\rho_{a}
	\end{pmatrix} ,
\end{align}
where $\gamma$ is an arbitrary constant chosen small enough to ensure that this matrix has non-negative entries. The stationary distribution of this embedded Markov chain is independent of $\gamma$ and consists of the limiting probabilities for $\pi_{\MSS}\left(\vA\right)$ and $\pi_{\MSS}\left(\va\right)$ specified in Eq.~\eqref{eq:pAndRho}.

\subsection{The rare-mutation conditional distribution} \label{sec:RMC}
According to Theorem \ref{thm:lowu}, as $u\to 0$, the mutation-selection stationary distribution becomes concentrated on the monoallelic states, $\va$ or $\vA$. However, since no selection occurs in the monoallelic states, it is important to  quantify the frequencies with which other states are visited in transit between them. For this purpose, \cite{allen2014measures} introduced the \emph{rare-mutation dimorphic (RMD) distribution} for haploid models with two alleles. Here, we introduce a natural generalization of this distribution, which we call the \emph{rare-mutation conditional (RMC) distribution}. We avoid the term ``dimorphic" because it can be misleading with non-haploid genetics; for example, the genotypes $AA$, $Aa$, and $aa$ could correspond to three different morphologies.

\begin{definition}
	The \emph{rare-mutation conditional (RMC) distribution} is the probability distribution on $\left\{0,1\right\}^G \setminus \left\{\va, \vA\right\}$ obtained by conditioning the MSS distribution on being in states other than $\va$ and $\vA$, and then taking the limit $u \to 0$:
	\begin{equation}
		\pi_\RMC\left(\vx\right) \coloneqq \lim_{u \to 0} \Prob_\MSS\left[\vX=\vx\ |\ \vX \notin\left\{\va, \vA\right\}\right] 
		= \lim_{u \to 0} \frac{\pi_\MSS\left(\vx\right)}{1-\pi_\MSS\left(\vA\right) -\pi_\MSS\left(\va\right)}.
	\end{equation}
\end{definition}
The existence of the above limit was shown by Allen and Tarnita \citeyearpar[Lemma 2]{allen2014measures}.

Allen and Tarnita \citeyearpar[Theorem 3]{allen2014measures} derived a recurrence relation from which the RMC distribution can be computed, in the case of unbiased mutation ($\nu = 1/2$). Here, we show that this recurrence relation---and hence the RMC distribution itself---is in fact independent of the mutational bias $\nu$. Informally speaking, as $u \to 0$, the mutational bias only affects the amount of time spent in the monoalleleic states $\va$ and $\vA$, which are by definition excluded from the RMC distribution. The RMC distribution depends only on transition probabilities in the absence of mutation and is therefore independent of $\nu$.

\begin{theorem}
	\label{thm:RMCrecur}
	For any given replacement rule $\left\{p_{\left(R, \alpha\right)} \left(\vx\right) \right\}_{\left(R,\alpha\right)}$, the RMC distribution is independent of the mutational bias $\nu$ and is uniquely determined by the recurrence relations
	\begin{equation}
		\label{eq:RMCrecur}
		\pi_\RMC \left(\vx\right) = \sum_{\vy \notin \left\{\va, \vA\right\}} \pi_\RMC\left(\vy\right) \left(  P_{\vy \to \vx} + P_{\vy \to \va} \mu_A\left(\vx\right) + P_{\vy \to \vA} \mu_a\left(\vx\right) \right) ,
	\end{equation}
	where the transition probabilities $P_{\vy \to \mathbf{z}}$ above are evaluated at $u=0$.
\end{theorem}

\begin{proof} 
	Let us temporarily fix a positive mutation rate, $u>0$, and a mutational bias, $\nu$. We apply Theorem \ref{thm:statereduction} to reduce $\cM$ to the set of states $S \coloneqq \left\{0,1\right\}^G \setminus \left\{\va, \vA\right\}$, i.e.~those states for which both alleles are present. The reduced Markov chain $\cM_{|S}$ has for a stationary distribution $\left\{ \pi_{\cM_{|S}} \right\}$, which is determined by the recurrence relations
	\begin{equation}
		\label{eq:conditionalrecur}
		\pi_{\cM_{|S}} \left(\vx\right) = \sum_{\vy \notin \left\{\va, \vA\right\}}\pi_{\cM_{|S}}\left(\vy\right) \left(  P_{\vy \to \vx} + P_{\vy \to \va} Q_{\va \to \vx} + P_{\vy \to \vA} Q_{\vA \to \vx} \right).
	\end{equation}
	Here, $Q_{\va \to \vx}$ is the probability that, from state $\va$, the first visit to the set $\left\{0,1\right\}^G \setminus \left\{\va, \vA\right\}$ occurs in state $\vx$; $Q_{\vA \to \vx}$ is defined similarly. By Lemma \ref{lem:mutantarise}, these probabilities have the low-mutation expansion
	\begin{subequations}
		\begin{align}
			Q_{\va \to \vx} & = \mu_A\left(\vx\right)  + \mathcal{O}\left(u\right) ; \\
			Q_{\vA \to \vx} & = \mu_a\left(\vx\right) + \mathcal{O}\left(u\right) .
		\end{align}
	\end{subequations}
	Therefore, taking the $u \to 0$ limit of Eq.~\eqref{eq:conditionalrecur} yields Eq.~\eqref{eq:RMCrecur}.
	
	To show that Eq.~\eqref{eq:RMCrecur} uniquely defines the RMC distribution, we note that any solution (in $\left\{\pi_\RMC\left(\vx\right)\right\}$) to Eq.~\eqref{eq:RMCrecur} is a stationary distribution for a new Markov chain $\cM_\RMC$, with states $\left\{0,1\right\}^G \setminus \left\{\va, \vA\right\}$ and transition probabilities
	\begin{align}
		P_{\vx \to \vy}^\RMC \coloneqq P_{\vx \to \vy} + P_{\vx \to \va} \mu_A\left(\vy\right) + P_{\vx \to \vA} \mu_a\left(\vy\right) .
	\end{align}
	Let $\vx, \vy \in \left\{0,1\right\}^G \setminus \left\{\va, \vA\right\}$ be any pair of states with $\mu_A\left(\vy\right)>0$. Using the Fixation Axiom one can show that it is possible to reach state $\vA$ from $\vx$, in the original Markov chain $\cM$, by a finite sequence of transitions with nonzero probability. Therefore, it is also possible to reach $\vy$ from $\vx$ in $\cM_\RMC$ by a finite sequence of transitions with nonzero probability, which shows that $\cM_\RMC$ has only a single closed communicating class and therefore possesses a unique stationary distribution, determined by Eq.~\eqref{eq:RMCrecur}.
	
	Finally, we note that none of the quantities in Eq.~\eqref{eq:RMCrecur} depend on the mutational bias $\nu$ since they are evaluated at $u=0$. Thus, the RMC distribution is independent of $\nu$.
\end{proof}

The following lemma, which relates the RMC distribution to the $u$-derivative of the MSS distribution at $u=0$, is very useful for both proofs and computations:
\begin{lemma}
	\label{lem:lowu}
	For any given replacement rule $\left\{p_{\left(R, \alpha\right)} \left(\vx\right) \right\}_{\left(R,\alpha\right)}$ and mutational bias $\nu$, the limit
	\begin{equation}
		\label{eq:Kdef}
		K \coloneqq \lim_{u \to 0} \frac{u}{\Prob_\MSS\left[\vX \notin \left\{\va, \vA\right\}\right]}
	\end{equation}
	exists and is finite and positive. Furthermore, if $\phi\left(\vx\right)$ is any state function with $\phi\left(\va\right) =\phi\left(\vA\right) =0$, then
	\begin{equation}
		\E_{\RMC}\left[\phi\right] = K \frac{d \E_{\MSS}\left[\phi\right]}{du} \Big|_{u=0} .
	\end{equation}
\end{lemma}

Lemma \ref{lem:lowu} allows expectations under RMC distribution to be computed, up to the proportionality constant $K$, from the MSS distribution (which is often easier to analyze). For many purposes, it is not necessary to know the value of $K$, only that it exists and is positive.

\begin{proof}
	Summing Eq.~\eqref{eq:MSSrecur} over the states $\vx \notin \left\{\va,\vA\right\}$, we have
	\begin{align}
		\Prob_\MSS\left[\vX \notin \left\{\va, \vA\right\}\right] &= \pi_\MSS\left(\va\right) \sum_{\vx \notin \left\{\va,\vA\right\}} P_{\va \to \vx}  + \pi_\MSS\left(\vA\right) \sum_{\vx \notin \left\{\va,\vA\right\}} P_{\vA \to \vx} \nonumber \\
		&\quad\quad + \sum_{\vx,\vy \notin \left\{\va,\vA\right\}}\pi_\MSS\left(\vy\right) P_{\vy \to \vx} .
	\end{align}
	Applying Theorem \ref{thm:lowu} and Lemma \ref{lem:mutantarise} to the first two terms on the right-hand side, we obtain the following expansion as $u \to 0$:
	\begin{align}
		\Prob_\MSS\left[\vX \notin \left\{\va, \vA\right\}\right] &= u  \frac{ \nu \left(1-\nu\right) b\left(\va\right) b\left(\vA\right)\left(\rho_a + \rho_A\right) }
		{ \nu b\left(\va\right) \rho_{A}+\left(1-\nu\right) b\left(\vA\right) \rho_{a}} \nonumber \\
		&\quad\quad + \sum_{\vx,\vy \notin \left\{\va,\vA\right\}}\pi_\MSS\left(\vy\right) P_{\vy \to \vx} + \mathcal{O}\left(u^2\right) .
	\end{align}
	Dividing by $u$ and taking $u \to 0$, we have
	\begin{align}
		\label{eq:lowuexpand}
		\lim_{u \to 0} \frac{1}{u} \Prob_\MSS\left[\vX \notin \left\{\va, \vA\right\}\right] &= \frac{ \nu \left(1-\nu\right) b\left(\va\right) b\left(\vA\right)\left(\rho_a + \rho_A\right) }
		{ \nu b\left(\va\right) \rho_{A}+\left(1-\nu\right) b\left(\vA\right)\rho_{a}} \nonumber \\
		&\quad\quad + \lim_{u \to 0} \left( \frac{1}{u}  \sum_{\vx,\vy \notin \left\{\va,\vA\right\}} \pi_\MSS\left(\vy\right) P_{\vy \to \vx} \right) .
	\end{align}
	Since $\nu$, $1-\nu$, $b\left(\va\right)$, $b\left(\vA\right)$, $\rho_a$, and $\rho_A$ are all positive, the first term on the right-hand side of Eq.~\eqref{eq:lowuexpand} is positive. The limit in the second term exists and is finite since $\pi_\MSS\left(\vy\right)$ and $P_{\vy \to \vx}$ are both rational functions of $u$ and $\lim_{u \to 0} \pi_\MSS\left(\vy\right)=0$ for $\vy \notin \left\{\va,\vA\right\}$. The limit in the second term is nonnegative since both $\pi_\MSS\left(\vy\right)$ and $P_{\vy \to \vx}$ are. Therefore, the limit on the left-hand side of Eq.~\eqref{eq:lowuexpand} exists and is positive; consequently, the limit in Eq.~\eqref{eq:Kdef} exists and is positive as well.
	
	For the second claim, we have
	\begin{align}
		\E_\RMC\left[\phi\right] & = \lim_{u \to 0} \E_\MSS\left[\phi\left(\vX\right)\ |\ \vX \notin\left\{\va, \vA\right\}\right] \nonumber \\
		& = \lim_{u \to 0} \frac{\E_\MSS\left[\phi\right]}{\Prob_\MSS\left[\vX \notin \left\{\va, \vA\right\}\right]} \qquad \text{(since $\phi\left(\vA\right) =\phi\left(\va\right) =0$)} \nonumber \\
		& =  \left(\lim_{u \to 0} \frac{u}{\Prob_\MSS\left[\vX \notin \left\{\va, \vA\right\}\right]} \right) 
		\left(\lim_{u \to 0} \frac{\E_\MSS\left[\phi\right]}{u} \right) \nonumber \\
		& = K \frac{d \E_\MSS\left[\phi\right]}{du} \Big|_{u=0} ,
	\end{align}
	which completes the proof.
\end{proof}

\section{Selection}
\label{sec:selection}
We turn now to the question of how selection acts on the two competing alleles, $a$ and $A$. We can ask this question on two different time-scales. In the short term, we can look at how natural selection acts to change allele frequencies from a given state. In the longer term, we can look at the fixation probabilities of each allele, or at their stationary frequencies under mutation-selection balance. These notions lead to different criteria for evaluating the success of an allele under natural selection. In this section, we define these criteria and prove (Theorem \ref{thm:rhodelsel}) that they become equivalent in the limit of low mutation when averaged over the RMC distribution.

\subsection{Change due to selection}\label{sec:changesel}
To address questions of short-term selection, we consider an evolutionary process in a given state $\vx \in \left\{0,1\right\}^G$. We let $\del\left(\vx\right)$ denote the expected change in the \emph{absolute} frequency of $A$ (i.e. the number of $A$ alleles) from state $\vx$, over a single transition:
\begin{equation}
	\label{eq:Deltax}
	\del\left(\vx\right) \coloneqq \left(1-u\right) \sum_{g \in G} x_g b_g\left(\vx\right) - \sum_{g \in G} x_g d_g\left(\vx\right) + u \nu b\left(\vx\right).
\end{equation}
We use absolute (rather than relative) frequency in the definition of $\del\left(\vx\right)$ to avoid tedious factors of $1/n$. The three terms on the right-hand side represent the respective contributions of faithful reproduction, death, and mutation. Collecting the terms involving $u$, we have
\begin{equation}
	\label{eq:Deltax2}
	\del\left(\vx\right) = \sum_{g \in G} x_g \left(b_g\left(\vx\right) - d_g\left(\vx\right) \right) + u \sum_{g \in G} \left(\nu - x_g\right) b_g\left(\vx\right) .
\end{equation}
Eq.~\eqref{eq:Deltax2} can be understood as a version of the \cite{price1970selection} equation (but see \citealp{van2005use}). The two terms on the right-hand side of Eq.~\eqref{eq:Deltax2} represent the effects of selection and mutation, respectively, which motivates the following definitions \citep{Eusociality}:

\begin{definition}
	\label{def:del}
	The \emph{expected change due to selection} from state $\vx$ is defined as
	\begin{equation}
		\label{eq:delseldef}
		\delsel\left(\vx\right) \coloneqq \sum_{g \in G}x_g \left( b_g\left(\vx\right) - d_g\left(\vx\right) \right) .
	\end{equation}
	The \emph{expected change due to mutation} from state $\vx$ is defined as
	\begin{equation}
		\label{eq:delmutdef}
		\delmut\left(\vx\right) \coloneqq u \sum_{g \in G} \left(\nu - x_g\right) b_g\left(\vx\right) .
	\end{equation}
\end{definition}

With the above definitions, Eq.~\eqref{eq:Deltax2} can be restated as $\del\left(\vx\right) = \delsel\left(\vx\right) + \delmut\left(\vx\right)$.

\subsection{Equivalence of success criteria}
How does one judge which of the two competing alleles, $a$ and $A$, is favored by selection? There are a number of reasonable criteria to use:
\begin{itemize}
	\item In a given state $\vx$, we could say that $A$ is favored if $\delsel\left(\vx\right) >0$.
	\item For an evolutionary process with no mutation, we could say that $A$ is favored if it has larger fixation probability; that is, if $\rho_A > \rho_a$.
	\item For an evolutionary process with mutation, we could say that that $A$ is favored if its stationary frequency is greater than one would expect by mutation alone; the latter quantity can be obtained by setting $\rho_A=\rho_a$ in Eq.~\eqref{eq:pAndRho}. This leads to the success criterion
	\begin{align}
		\lim_{u \to 0} \E_\MSS\left[x\right] > \frac{\nu b\left(\va\right)}{\nu b\left(\va\right) + \left(1-\nu\right) b\left(\vA\right)}.
	\end{align}
	In the case that the overall birth rates coincide in the two monoallelic states, $b\left(\va\right) =b\left(\vA\right)$, this criterion reduces to $\lim_{u \to 0} \E_\MSS\left[x\right] > \nu$.
\end{itemize}

Our first main result shows that these success criteria become equivalent in the limit $u \to 0$, when $\delsel$ is averaged over the RMC distribution. Alternatively, one may average $\delsel$ over the MSS distribution and take the $u$-derivative at $u=0$.

\begin{theorem}
	\label{thm:rhodelsel}
	For any replacement rule, $\left\{p_{\left(R, \alpha\right)} \left(\vx\right) \right\}_{\left(R,\alpha\right)}$, and any mutational bias, $\nu$, the following success criteria are equivalent:
	\renewcommand{\labelenumi}{(\alph{enumi})}
	\begin{enumerate}
		\item $\rho_A > \rho_a$;
		\item $\displaystyle \lim_{u \to 0} \E_\MSS\left[x\right] > \frac{\nu b\left(\va\right)}{\nu b\left(\va\right) + \left(1-\nu\right) b\left(\vA\right)}$;
		\item $\E_\RMC \left[ \delsel \right]>0$;
		\item $\frac{d}{du}\E_\MSS\left[\delsel\right]  \big|_{u=0}  > 0$.
	\end{enumerate}
\end{theorem}

The equivalence of (b) and (d), in the case that $b(\va)=b(\vA)$, was previously shown by \cite{tarnita2014measures}. Under the further assumption that $\nu =1/2$, \citeauthor{Eusociality} \citeyearpar[Corollary 1 of Appendix A]{Eusociality} showed the equivalence of (b) and (d); \citeauthor{allen2014measures} \citeyearpar[Theorem 6 and Corollary 2]{allen2014measures} showed the equivalence of (a), (b), and (c); and \cite{van2015social} showed the equivalence of (a), (b), and a variant of (c). Special cases of this result for particular models were also proven by \cite{RoussetBilliard} and \cite{TaylorFixProb}.

\begin{proof}
	We begin by assuming a fixed $u>0$ and rewriting Eq.~\eqref{eq:Deltax2} as 
	\begin{equation}
		\label{eq:mutseldecomp}
		\del\left(\vx\right) = \delsel\left(\vx\right) - u \sum_{g \in G} \left(x_g-\nu\right) b_g\left(\vx\right) .
	\end{equation}
	We now take the expectation of both sides under the MSS distribution. The left-hand side vanishes since the expected change in any quantity is zero when averaged over a stationary distribution. We therefore have
	\begin{equation}
		\E_\MSS\left[ \delsel \right] = u \E_\MSS \left[ \sum_{g \in G} \left(x_g-\nu\right) b_g \right] .
	\end{equation}
	We also observe from Eq.~\eqref{eq:delseldef} that $\delsel\left(\va\right) =\delsel\left(\vA\right) =0$. Applying Lemma \ref{lem:lowu}, we have
	\begin{equation}
		\E_\RMC\left[ \delsel \right] = K \frac{d \E_\MSS\left[\delsel\right]}{du} \Big|_{u=0} = K \lim_{u \to 0}\E_\MSS \left [ \sum_{g \in G} \left(x_g-\nu\right) b_g \right],
	\end{equation}
	with $K>0$. Theorem \ref{thm:lowu} now gives
	\begin{align}
		\lim_{u \to 0}\E_\MSS \left [ \sum_{g \in G} \left(x_g-\nu\right) b_g \right]
		& = \left(1-\nu\right) b\left(\vA\right) \lim_{u \to 0} \pi_\MSS\left(\vA\right)  - \nu b\left(\va\right) \lim_{u \to 0} \pi_\MSS\left(\va\right) \nonumber \\
		& = \frac{\nu\left(1-\nu\right) b\left(\va\right) b\left(\vA\right) }{ \nu b\left(\va\right) \rho_{A}+\left(1-\nu\right) b \left(\vA\right) \rho_{a}} \left(\rho_A - \rho_a\right) .
	\end{align}
	The coefficient of $\rho_A - \rho_a$ above is positive; thus $\E_\RMC\left[ \delsel \right]$, $\frac{d}{du}\E_\MSS\left[\delsel\right]  \big|_{u=0}$, and $\rho_A - \rho_a$ have the same sign. This proves $\text{(a)} \Leftrightarrow \text{(c)} \Leftrightarrow \text{(d)}$. For (b), we write
	\begin{align}
		\lim_{u \to 0} &\E_\MSS\left[x\right] - \frac{\nu b\left(\va\right)}{\nu b\left(\va\right) + \left(1-\nu\right) b\left(\vA\right)} \nonumber \\
		&\quad = \lim_{u \to 0} \pi_\MSS\left(\vA\right) -\frac{\nu b\left(\va\right)}{\nu b\left(\va\right) + \left(1-\nu\right) b\left(\vA\right)} \nonumber \\
		&\quad =  \frac{\nu\left(1-\nu\right) b\left(\va\right) b\left(\vA\right) }
		{ \left( \nu b\left(\va\right) +\left(1-\nu\right) b\left(\vA\right)\right) \left( \nu b\left(\va\right) \rho_{A}+\left(1-\nu\right) b\left(\vA\right)\rho_{a}\right)} 
		\left(\rho_{A} - \rho_{a} \right),
	\end{align}
	by Theorem \ref{thm:lowu}. The last line above has the sign of $\rho_A - \rho_a$, thus $\text{(a)}\Leftrightarrow\text{(b)}$.
\end{proof}

If the conditions of Theorem \ref{thm:rhodelsel} are satisfied, we say that allele $A$ is \emph{favored by selection}. An interesting consequence of Theorem \ref{thm:rhodelsel} is that Condition (b) is independent of the mutational bias, $\nu$; either it holds for all values of $\nu$ or else for none of them. We can understand this result to say that, when mutation is vanishingly rare, mutational bias does not affect the direction of selection.

Theorem \ref{thm:rhodelsel} is our most general equivalence result. It shows that four reasonable measures of selection coincide with each other when mutation is rare. However, for many models of interest, none of the four conditions are analytically or computationally tractable \citep{ibsen2015computational}. In what follows, we will begin to introduce additional assumptions that allow us to define important notions such as reproductive value and fitness and obtain conditions that more tractable than those in Theorem \ref{thm:rhodelsel}.

\section{Reproductive value and fitness}\label{sec:RVfit}
Reproductive value and fitness are ubiquitous concepts in evolutionary theory. Both quantify the expected reproductive success of an individual or a genetic site. Fitness, which is used to quantify selection, takes into account the alleles present in the population. In contrast, reproductive value quantifies reproductive success in the \emph{absence} of selection and is therefore independent of the alleles in the population. Both fitness and reproductive value may depend on other factors such as age, sex, caste, and spatial location.

In this section, we define both notions for our class of models. First, however, we must introduce an additional assumption regarding the consistency of replacement in the monoallelic states $\va$ and $\vA$.

\subsection{Consistency of monoallelic states}\label{sec:consistency}
Since selection does not occur in the monoallelic states $\va$ and $\vA$, the capacity of a site to reproduce (i.e.~its reproductive value) is ascribable to the site itself, not to the alleles in the population. It is therefore reasonable to define reproductive value with respect to states $\va$ and $\vA$. To obtain a consistent definition, we require that the probabilities of replacement events coincide in these states:

\begin{assumption}
	\label{ass:neutral}
	(Consistency of monoallelic states) Each replacement event $\left(R,\alpha\right)$ has the same probability in state $\vA$ as in state $\va$: $p_{\left(R,\alpha\right)}\left(\vA\right) = p_{\left(R,\alpha\right)}\left(\va\right)$.
\end{assumption}

Assumption \ref{ass:neutral} does not necessarily hold for every plausible model of natural selection. For example, if allele $A$ has a positive fitness effect in some sites and a negative fitness effect in others (relative to allele $a$), then the patterns of replacement in state $\vA$ are likely to differ from those in state $\va$. Importantly, if Assumption \ref{ass:neutral} is not satisfied, there may not be a consistent way to define reproductive values. That is, one may obtain two different sets of reproductive values, one for state $\va$ and one for state $\vA$, with no obvious way of reconciling them.

With Assumption \ref{ass:neutral} in force, we denote the probability of a replacement event $(R,\alpha)$ in state $\va$ or $\vA$ by $p^\circ_{\left(R,\alpha\right)} \coloneqq p_{\left(R,\alpha\right)}\left(\vA\right) = p_{\left(R,\alpha\right)}\left(\va\right)$. We will also use the superscript ${}^\circ$ to denote the following other quantities in states $\va$ or $\vA$:
\begin{subequations}
	\label{eq:ebdcirc}
	\begin{align}
		e_{gh}^\circ & \coloneqq e_{gh}\left(\va\right) = e_{gh}\left(\vA\right) = \sum_{\substack{\left(R, \alpha\right) \\ \alpha\left(h\right) =g}} p^\circ_{\left(R, \alpha\right)} ; \\
		b_g^\circ & \coloneqq b_g\left(\va\right) = b_g\left(\vA\right) = \sum_{h\in G} e_{gh}^\circ ; \\
		d_g^\circ & \coloneqq d_g\left(\va\right) = d_g\left(\vA\right) = \sum_{h \in G} e_{hg}^\circ .
	\end{align}
\end{subequations}

\subsection{Reproductive value} \label{sec:RV}
Reproductive value \citep{fisher1930genetical,taylor1990allele,maciejewski2014reproductive} quantifies the expected contribution of an individual or a genetic site to the future gene pool of the population in the absence of selection, depending on factors such as age, sex, location, and caste. Reproductive values indicate the relative importance of different individuals to the process of natural selection. For example, a sterile worker in an insect colony has zero reproductive value, since it has no opportunity to transmit its genetic material.

We will first define reproductive value on the level of genetic sites, and later (in Section \ref{sec:ind}) extend to individuals. The reproductive value $v_g$ of a genetic site $g \in G$ is defined as follows:
\begin{definition}
	For a replacement rule satisfying Assumption \ref{ass:neutral}, the \emph{reproductive values} $\left\{v_g\right\}_{g \in G}$ are defined as the unique solution to the system of equations
	\begin{subequations}
		\label{eq:vsystem}
		\begin{align}
			\label{eq:vrecur}
			d_g^\circ v_g = \sum_{h \in G} e_{gh}^\circ v_h& \qquad \text{for all $g \in G$} ;\\
			\label{eq:vsum}
			\sum_{g \in G} v_g = n&.
		\end{align}
	\end{subequations}
\end{definition}
Eq.~\eqref{eq:vrecur} can be understood as saying that, for an allele occupying site $g \in G$, under one transition in either of the monoallelic states, the expected loss of reproductive value due to death, $d_g^\circ v_g$, is balanced by the expected reproductive value of new copies produced, $\sum_{h \in G} e_{gh}^\circ v_h$. The normalization in Eq.~\eqref{eq:vsum} is arbitrary, chosen so that the average reproductive value of each site is one. 

We prove in Proposition \ref{lem:RVunique} below that reproductive values are uniquely defined by Eq.~\eqref{eq:vsystem} and are nonnegative. We show in Section \ref{sec:neutral} below that, under neutral drift, the reproductive value of a site $g$ is proportional to the probability that a mutation arising at site $g$ becomes fixed.

Reproductive value (RV) provides a natural weighting for genetic sites when computing quantities related to natural selection. We indicate RV-weighted quantities with a hat; for example, the RV-weighted frequency $\hat{x}$ is defined as
\begin{equation}
	\hat{x} \coloneqq \frac{1}{n} \sum_{g \in G} v_g x_g .
\end{equation}
We also define the RV-weighted birth and death rates of each site $g \in G$:
\begin{subequations}
	\begin{align}
		\label{eq:bhat}
		\hat{b}_g \left(\vx\right) & \coloneqq \sum_{h\in G} e_{gh} \left(\vx\right) v_h ;  \\
		\label{eq:dhat}
		\hat{d}_g \left(\vx\right) & \coloneqq v_g d_g \left(\vx\right) .
	\end{align}
\end{subequations}

It follows from Eq.~\eqref{eq:vsystem} that in the monoallelic states, the RV-weighted birth and death rates are equal for each state:
\begin{equation}
	\label{eq:RVbdneutral}
	\hat{b}^\circ_g = \hat{d}^\circ_g.
\end{equation}

We let $\hat{b}\left(\vx\right)$ denote the total RV-weighted birth rate in state $\vx$, which is equal to the total RV-weighted death rate:
\begin{equation}
	\label{eq:sumbd}
	\hat{b}\left(\vx\right) = \sum_{g \in G} \hat{b}_g \left(\vx\right) = \sum_{g,h \in G} e_{gh}\left(\vx\right) v_h = \sum_{h \in G} \hat{d}_h \left(\vx\right) .
\end{equation}

\subsection{Fitness} \label{sec:fitness}
Fitness quantifies reproductive success under natural selection. Although the concept of fitness is fundamental in evolutionary biology \citep{haldane1924mathematical,fisher1930genetical}, it can be difficult to define for a general evolutionary process \citep{MetzFitness,RandESAs,doebeli2017towards}.

As with other quantities, we first define fitness on the level of genetic sites. Intuitively, the fitness of site $g$ in state $\vx$ should quantify the expected reproductive success of the allele occupying $g$ in this state. This success can be quantified in terms of its own reproductive value (if it survives) plus the expected reproductive value of copies it produces and transmits, which leads to the following definition:
\begin{definition}
	The \emph{fitness} of site $g \in G$ in state $\vx$ is defined as
	\begin{align}
		w_g \left(\vx\right) \coloneqq & \left( 1 - \sum_{h \in G} e_{hg}\left(\vx\right) \right)v_g + \sum_{h \in G}  e_{gh}\left(\vx\right) v_h \nonumber \\
		= & v_g - \hat{d}_g \left(\vx\right) + \hat{b}_g \left(\vx\right) . \label{eq:fitdef}
	\end{align}
\end{definition}

The definition of fitness used here differs from the definition in \cite{allen2014measures}, which does not take reproductive value into account.

We observe that, by Eq.~\eqref{eq:vsum}, the total fitness is $n$ in every state $\vx$, i.e.
\begin{equation}
	\label{eq:fitsum}
	\sum_{g \in G} w_g \left(\vx\right) = n.
\end{equation}
For the monoallelic states, it follows from Eq.~\eqref{eq:RVbdneutral} that each site has fitness equal to its reproductive value:
\begin{equation}
	\label{eq:fitneutral}
	w_g^\circ = v_g - \hat{d}_g^\circ + \hat{b}_g^\circ = v_g.
\end{equation}

\subsection{Selection with reproductive value}
To quantify selection using reproductive value, we use the expected change in the absolute RV-weighted frequency from a given state $\vx$, denoted $\delhat\left(\vx\right)$, which we define as follows:
\begin{align}\label{eq:delhatexpand}
	\delhat\left(\vx\right) \coloneqq  & - \sum_{g \in G} x_g v_g d_g \left(\vx\right) + \left(1-u\right) \sum_{g,h \in G} x_g e_{gh} \left(\vx\right) v_h 
	+ u \nu \sum_{g \in G} v_g d_g\left(\vx\right) \nonumber \\ 
	= & - \sum_{g \in G} x_g \hat{d}_g \left(\vx\right) + \left(1-u\right) \sum_{g \in G} x_g \hat{b}_g \left(\vx\right) + u \nu \sum_{g \in G} \hat{d}_g\left(\vx\right) .
\end{align}
(Note that $\delhat\left(\vx\right)$, like $\del\left(\vx\right)$, is defined using absolute rather than relative weights, in order to avoid tedious factors of $1/n$.) 

We rewrite Eq.~\eqref{eq:delhatexpand}, in analogy to Eq.~\eqref{eq:mutseldecomp}, as
\begin{equation}
	\label{eq:mutselRVdecomp}
	\delhat\left(\vx\right) = \delhatsel\left(\vx\right) - u \sum_{g \in G} \left(x_g - \nu\right) \hat{b}_g \left(\vx\right) .
\end{equation}
Above (following \citealp{tarnita2014measures}) we have introduced $\delhatsel\left(\vx\right)$, the \emph{expected change in RV-weighted frequency due to selection} from state $\vx$, which can be defined in a number of equivalent ways:
\begin{align}
	\label{eq:delhatsel}
	\delhatsel\left(\vx\right) &= \sum_{g \in G}x_g \left ( \hat{b}_g\left(\vx\right) -\hat{d}_g\left(\vx\right) \right) \nonumber \\
	& =  \sum_{g \in G} x_g \left( w_g\left(\vx\right) - v_g \right) \nonumber \\
	& = \sum_{g,h \in G}x_g \left( e_{gh}\left(\vx\right) v_h - e_{hg}\left(\vx\right) v_g \right) \nonumber \\
	& = \frac{1}{2} \sum_{g,h \in G} \left(x_g - x_h\right) \left( e_{gh}\left(\vx\right) v_h - e_{hg}\left(\vx\right) v_g \right) .
\end{align}

We now extend our equivalence of success criteria (Theorem \ref{thm:rhodelsel}) to include measures weighted by reproductive value:
\begin{theorem}
	\label{thm:rhodelselhat}
	For any replacement rule, $\left\{p_{\left(R, \alpha\right)} \left(\vx\right) \right\}_{\left(R,\alpha\right)}$, satisfying Assumption \ref{ass:neutral}, and any mutational bias $\nu$, the following success criteria are equivalent:
	\renewcommand{\labelenumi}{(\alph{enumi})}
	\begin{enumerate}
		\item $\rho_A > \rho_a$;
		\item $\E_\RMC\left[ \delsel \right] > 0$;
		\item $\E_\RMC\left[ \delhatsel \right] > 0$;
		\item $\frac{d}{du}  \E_\MSS\left[\delsel\right]\big|_{u=0} > 0$;
		\item $\frac{d}{du}  \E_\MSS\left[\delhatsel\right] \big|_{u=0} > 0$;
		\item $\lim_{u \to 0} \E_\MSS\left[ x \right] \; (= \lim_{u \to 0} \E_\MSS\left[ \hat{x} \right]) > \nu$.
	\end{enumerate}
\end{theorem}

The equivalence of (d) and (f) was previously shown by \cite{tarnita2014measures}.

\begin{proof}
	Analogously to the proof of Theorem \ref{thm:rhodelsel}, we temporarily fix $u>0$ and take the expectation of both sides of Eq.~\eqref{eq:mutselRVdecomp} under the MSS distribution. Upon rearranging, we obtain
	\begin{equation}
		\E_\MSS \left[ \delhatsel \right] = u \E_\MSS \left[ \sum_{g \in G} \left(x_g-\nu\right) \hat{b}_g \right].
	\end{equation}
	Since $\delhatsel\left(\va\right) =\delhatsel\left(\vA\right) =0$, application of Lemma \ref{lem:lowu} yields
	\begin{equation}
		\label{eq:RMCMSS}
		\E_\RMC \left[ \delhatsel \right] = K \frac{d \E_\MSS\left[\delhatsel\right]}{du} \Big|_{u=0} 
		= K \lim_{u \to 0}\E_\MSS \left[ \sum_{g \in G} \left(x_g-\nu\right) \hat{b}_g \right],
	\end{equation}
	for some $K>0$. We observe that, by Assumption \ref{ass:neutral}, $b\left(\vA\right) =b\left(\va\right)$ and $\hat{b}\left(\vA\right) = \hat{b}\left(\va\right)$; we denote the latter quantity by $\hat{b}^\circ$. Theorem \ref{thm:lowu} now gives
	\begin{align}
		\lim_{u \to 0}\E_\MSS \left [ \sum_{g \in G} \left(x_g-\nu\right) \hat{b}_g \right] 
		&= \hat{b}^\circ \left( \left(1-\nu\right) \lim_{u \to 0} \pi_\MSS\left(\vA\right) - \nu \lim_{u \to 0} \pi_\MSS\left(\va\right) \right) \nonumber \\
		\label{eq:MSSrho}
		& = \frac{\hat{b}^\circ \nu\left(1-\nu\right) }{\nu \rho_{A}+\left(1-\nu\right) \rho_{a}} (\rho_A - \rho_a).
	\end{align}
	The coefficient of $\rho_A - \rho_a$ above is positive, which proves the equivalence of (a), (c), and (e). The rest of the proof follows from Theorem \ref{thm:rhodelsel}.
\end{proof}

We note that Theorem \ref{thm:rhodelselhat} requires Assumption \ref{ass:neutral} and is therefore less general than its non-RV-weighted analogue, Theorem \ref{thm:rhodelsel} (see \citealp{tarnita2014measures}, for related discussion).

\section{Neutral drift}\label{sec:neutral}
Neutral drift describes a situation where the alleles present in the population do not affect the drivers of selection (births and deaths). Neutral drift is interesting and important in its own right (e.g.~\citealp{kimura1968evolutionary,allen2015molecular,mcavoy2018stationary}), and will also serve as a baseline for studying weak selection. In our framework, neutral drift is characterized by the property that the probabilities of replacement events are independent of the population state:

\begin{definition}
	A replacement rule $\left\{p_{\left(R,\alpha\right)}\left(\vx\right)\right\}_{\left(R,\alpha\right)}$ \emph{represents neutral drift} if $p_{\left(R,\alpha\right)}\left(\vx\right)$ is independent of the state $\vx$ for every $(R,\alpha)$.
\end{definition}

Clearly, a replacement rule representing neutral drift satisfies Assumption 1; indeed, we have $p_{\left(R,\alpha\right)}\left(\vx\right) = p^\circ_{\left(R,\alpha\right)}$ for each replacement event, $\left(R,\alpha\right)$, and each state, $\vx$. For this reason, we will denote a replacement rule representing neutral drift by its (fixed) probability distribution, $\left\{p_{\left(R,\alpha\right)}^\circ\right\}_{\left(R,\alpha\right)}$, over replacement events.

\subsection{Ancestral random walks and uniqueness of reproductive value}
A convenient property of neutral drift is that it can be analyzed backwards in time, using the perspective of coalescent theory \citep{kingman1982coalescent,WakeleyCoalescent}. Let  $\left\{p_{\left(R,\alpha\right)}^\circ\right\}_{\left(R,\alpha\right)}$ be a replacement rule representing neutral drift. For a given site, $g \in G$, the probability $a_{gh}$ that the parent copy of the allele in $g$ occupied site $h$ is
\begin{align}
	a_{gh} \coloneqq \Prob^\circ\left[\alpha\left(g\right) =h \, | \, g\in R\right] = \frac{e_{hg}^\circ}{d_g^\circ} .
\end{align}
(Note that $d_g^\circ>0$ as a consequence of the Fixation Axiom.) We define the \emph{ancestral random walk} as a Markov chain, $\mathcal{A}$, on $G$ with transition probabilities $a_{gh}$ from $g$ to $h$. The trajectory of the ancestral random walk from a given site $g$ represents the ancestry of $g$ traced backwards in time. The ancestral random walk is a special case of a \emph{coalescing random walk}
\citep{holley1975ergodic,cox1989coalescing} for which there is only one walker.

The ancestral random walk enables an intuitive proof for the uniqueness and non-negativity of reproductive values:

\begin{proposition}
	\label{lem:RVunique}
	For a given replacement rule satisfying Assumption \ref{ass:neutral}, the reproductive values, $\left\{v_g\right\}_{g \in G}$, are uniquely defined by Eq.~\eqref{eq:vsystem} and are nonnegative for every $g \in G$.
\end{proposition}
\begin{proof}
	First suppose that the given replacement rule represents neutral drift. Let $\mathcal{A}$ be the corresponding ancestral random walk, with transition probabilities $a_{gh} = e_{hg}^\circ/d_g^\circ$. The Fixation Axiom implies that there exists a $g \in G$ such that, for all $h \in G$, there is a $k \geqslant 0$ and a sequence $\ell_1, \ldots, \ell_k \in G$ such that $a_{h \ell_k} a_{\ell_k \ell_{k-1}} \cdots a_{\ell_2 \ell_1} a_{\ell_1 g} >0$. In other words, there exists $g \in G$ such that for every $h\in G$, there is a finite sequence of transitions in $\mathcal{A}$, each with positive probability, from $h$ to $g$. It follows that $\mathcal{A}$ has a single closed communicating class, and therefore it has a unique stationary probability distribution, $\{z_g\}_{g \in G}$, which is the unique solution to the system of equations
	\begin{subequations}
		\begin{align}
			z_g &= \sum_{h \in G} \frac{e_{gh}^\circ}{d_h^\circ} z_h ; \\
			\sum_{g \in G} z_g  &= 1.
		\end{align}
	\end{subequations}
	Setting 
	\begin{align}
		v_g \coloneqq \frac{ z_g /d_g^\circ}{\sum_{\ell \in G}\left( z_\ell /d_\ell^\circ\right) } ,
	\end{align}
	it follows that $\left\{v_g\right\}_{g \in G}$ is the unique solution to Eq.~\eqref{eq:vsystem}. The $z_g$ are nonnegative since they comprise a probability distribution, and it follows that the $v_g$ are nonnegative as well.
	
	If the given replacement rule does not represent neutral drift, we define a new replacement rule $\left\{\tilde{p}_{\left(R,\alpha\right)} \right\}_{\left(R,\alpha\right)}$ by $\tilde{p}_{\left(R,\alpha\right)} \coloneqq p_{\left(R,\alpha\right)}\left(\va\right) =p_{\left(R,\alpha\right)}\left(\vA\right)$. This new replacement rule represents neutral drift by definition. The above argument again shows that the reproductive values are uniquely defined by Eq.~\eqref{eq:vsystem} and are nonnegative.
\end{proof}

As a corollary to the proof, we can see that the states with positive reproductive value are precisely those that are recurrent under the ancestral random walk, which in turn are those that are able to spread their contents throughout the population in the sense of the Fixation Axiom. This fact hints at a connection between reproductive value and fixation probability, which we will make explicit in Section \ref{sec:RVfixprob}.

\subsection{Change due selection vanishes under neutral drift}
A key result for neutral drift is that the RV-weighted change due to selection, $\delhatsel\left(\vx\right)$, is zero in every state, $\vx$. This property mathematically expresses the \emph{neutrality} (absence of selection) in neutral drift. Notably, the analogous property does not hold for unweighted change due to selection, $\delsel\left(\vx\right)$. Indeed, this property uniquely defines reproductive value, and is one of the primary motivations studying RV-weighted quantities (see also \citealp{tarnita2014measures}). We formalize these observations in the following proposition:

\begin{theorem}
	\label{thm:delhatselneut}
	If the replacement rule $\left\{p^\circ_{\left(R,\alpha\right)}\right\}_{\left(R,\alpha\right)}$ represents neutral drift, then $\delhatsel\left(\vx\right) = 0$ for each state $\vx \in \left\{0,1\right\}^G$. Furthermore, if $\left\{\tilde{v}_g\right\}_{g \in G}$ is any weighting of genetic sites $g \in G$ such that $\tilde{\del}_\mathrm{sel}\left(\vx\right)$ (the change in $\tilde{v}$-weighted frequency) is zero for each state $\vx$, then $\left\{\tilde{v}_g\right\}_{g \in G}$ are a constant multiple of $\left\{v_g\right\}_{g \in G}$.
\end{theorem}

\begin{proof}
	For the first claim, combining Eqs.~\eqref{eq:RVbdneutral} and \eqref{eq:delhatsel} gives
	\begin{equation}
		\delhatsel(\vx) = \sum_{g \in G} x_g \left(\hat{b}_g^\circ - \hat{d}_g^\circ \right) = 0. 
	\end{equation}
	
	For the second claim, consider an arbitrary weighting of genetic sites $\left\{\tilde{v}_g\right\}_{g \in G}$. In analogy with Eq.~\eqref{eq:delhatsel}, the expected change in $\tilde{x} = \sum_{g \in G} \tilde{v}_g x_g$ can be written
	\begin{equation}
		\label{eq:deltildesel}
		\tilde{\del}_\mathrm{sel}\left(\vx\right) = \sum_{g,h \in G}x_g \left( e_{gh}^\circ \tilde{v}_h - e_{hg}^\circ \tilde{v}_g \right).
	\end{equation}
	
	If $\tilde{\del}_\mathrm{sel}\left(\vx\right) =0$ for each state $\vx$, then, by substituting $\vx=\mathbf{1}_{\left\{g\right\}}$ into Eq.~\eqref{eq:deltildesel} and rearranging, we obtain
	\begin{equation}
		\label{eq:tildevrecur}
		\sum_{h \in G} e_{gh}^\circ  \tilde{v}_h = \sum_{h \in G} e_{hg}^\circ  \tilde{v}_g = d_g^\circ \tilde{v}_g \quad
		\text{for each $g \in G$.}
	\end{equation}
	The above equation is equivalent to Eq.~\eqref{eq:vrecur}, with $\tilde{v}_g$ in place of $v_g$. Proposition \ref{lem:RVunique} guarantees the solution to Eq.~\eqref{eq:vrecur} is unique up to a constant multiple. Therefore, the $\left\{\tilde{v}_g\right\}_{g \in G}$ are a constant multiple of $\left\{v_g\right\}_{g \in G}$.
\end{proof}

\subsection{Reproductive value and fixation probability}\label{sec:RVfixprob}
As an application of Theorem \ref{thm:delhatselneut}, we deduce a remarkable relationship between reproductive value and fixation probability: under neutral drift with no mutation, the reproductive value of a site is proportional to the fixation probability of a novel type initiated at that site. This relationship was previously noted by \cite{maciejewski2014reproductive} and \cite{allen2015molecular}; here, we provide an elegant proof using martingales:
\begin{theorem} 
	\label{thm:neutfixprob}
	Let $\left\{p^\circ_{\left(R,\alpha\right)}\right\}_{\left(R,\alpha\right)}$ be a replacement rule representing neutral drift and let $\cM$ be the associated evolutionary Markov chain with $u=0$. Then, for any state $\vx \in \left\{0,1\right\}^G$,
	\begin{equation}
		\label{eq:fixRV1}
		\lim_{t \to \infty} P^{\left(t\right)}_{\vx \to \vA} = \hat{x} .
	\end{equation}
	In particular, the fixation probability of a single mutation arising in site $g$ is
	\begin{equation}
		\label{eq:fixRV2}
		\lim_{t \to \infty} P^{\left(t\right)}_{\mathbf{1}_{\left\{g\right\}} \to \vA} = \frac{v_g}{n} ,
	\end{equation}
	and the overall fixation probabilities of $A$ and $a$ are
	\begin{align}
		\label{eq:neutfixprob}
		\rho_{A} = \rho_{a} = \frac{\hat{b}^{\circ}}{nb^{\circ}} .
	\end{align}
\end{theorem}

\begin{proof}
	Consider $\cM$ started from an arbitrary initial state $\vX\left(0\right) =\vx_0$. Let $\vX\left(t\right)$ denote the state at time $t$, and let $\hat{X}\left(t\right) \coloneqq \frac{1}{n}\sum_{g \in G} v_g X_g\left(t\right)$ denote the RV-weighted frequency at time $t$. We claim that $\hat{X}\left(t\right)$ is a martingale, which can be seen by writing Eq.~\eqref{eq:mutselRVdecomp} as
	\begin{align}
		\E\left[\hat{X}\left(t+1\right) -\hat{X}\left(t\right)\ |\ \vX\left(t\right) =\vx\right] &= \frac{1}{n} \delhat\left(\vx\right) \nonumber \\ 
		& = \frac{1}{n} \delhatsel\left(\vx\right) - \frac{u}{n} \sum_{g \in G} \left(x_g - \nu \right) \hat{b}_g^\circ .
	\end{align}
	The first term in the final expression is zero by Theorem \ref{thm:delhatselneut}, and the second is zero since $u=0$; thus $\hat{X}(t)$ is a martingale.
	
	Since $u=0$, $\cM$ eventually becomes absorbed either in state $\vA$, for which $\hat{x}=1$, or $\va$, for which $\hat{x}=0$. The martingale property then implies
	\begin{align}
		\hat{x}_0 = \lim_{t \to \infty} \E\left[\hat{X}\left(t\right)\ |\ \vX\left(0\right) =\vx_0\right] =
		\lim_{t \to \infty} P^{\left(t\right)}_{\vx_0 \to \vA} ,
	\end{align}
	which proves Eq.~\eqref{eq:fixRV1}. Eq.~\eqref{eq:fixRV2} follows from setting $\vx_0 = \mathbf{1}_{\left\{g\right\}}$, and Eq.~\eqref{eq:neutfixprob} follows from Definitions \ref{def:mutapper} and \ref{def:meanFixationProbabilities}.
\end{proof}

Note that, according to Eq.~\eqref{eq:neutfixprob}, the fixation probability of a neutral mutation is not necessarily $1/N$, and depends on the spatial structure. It follows that spatial structure can affect a population's rate of neutral substitution (or ``molecular clock"; \citealp{kimura1968evolutionary}). This effect is discussed in detail by \cite{allen2015molecular}.

\subsection{Symmetry of neutral distributions}
Since, for neutral drift, the alleles $a$ and $A$ are interchangeable, we expect the neutral RMC and MSS distributions to be insensitive to interchanging the roles of $a$ and $A$. To formalize this property, we define the \emph{complement} of a state $\vx$, denoted $\bar{\vx}$, by $\bar{x}_{g}=1-x_g$ for each $g \in G$. In other words, $\bar{\vx}$ is formed from $\vx$ by replacing all $a$'s with $A$'s and vice versa. In particular, $\bar{\va} = \vA$ and $\bar{\vA} = \va$. The symmetry property for the RMC distribution can then be stated as follows:

\begin{proposition}
	\label{prop:symmetry}
	If the replacement rule $\left\{p^\circ_{\left(R,\alpha\right)}\right\}_{\left(R,\alpha\right)}$ represents neutral drift, then for each state $\vx \in \left\{0,1\right\}^G$, $\pi_\RMC \left(\bar{\vx} \right) = \pi_\RMC\left(\vx\right)$.
\end{proposition}

\begin{proof}
	For a replacement rule representing neutral drift and $u=0$, it follows from how transitions are defined that $P_{\vx \to \vy}= P_{\bar{\vx} \to \bar{\vy}}$ for all states $\vx,\vy \in \left\{0,1\right\}^G$. We further observe that $\mu_a\left(\vx\right) =\mu_A \left( \bar{\vx} \right)$ and $\mu_A\left(\vx\right) =\mu_a \left( \bar{\vx} \right)$. Substituting into the recurrence relation for the RMC distribution, Eq.~\eqref{eq:RMCrecur}, we obtain
	\begin{align}\label{eq:recurrenceRMC}
		\pi_\RMC \left(\vx\right) = \sum_{\vy \notin \left\{\va, \vA\right\}} \pi_\RMC\left(\vy\right) \left( P_{\bar{\vy} \to \bar{\vx}} + P_{\bar{\vy} \to \vA} \mu_a \left(\bar{\vx} \right) + P_{\bar{\vy} \to \va} \mu_A \left(\bar{\vx} \right) \right) ,
	\end{align}
	for all states $\vx$, with the transition probabilities evaluated at $u=0$. Since Eq.~\eqref{eq:recurrenceRMC} uniquely determines $\left\{ \pi_\RMC \left( \bar{\vx} \right) \right\}$, we have $\pi_\RMC \left( \bar{\vx} \right) = \pi_\RMC \left( \vx \right)$ for all $\vx \in \left\{0,1\right\}^G$.
\end{proof}

In particular, under neutral drift, we have $\E_\RMC\left[x_g\right] = 1/2$ for each site $g$; that is, $g$ is equally likely to be occupied by either allele in the RMC distribution for neutral drift.

For the MSS distribution, interchanging the roles of $a$ and $A$ leads to a different symmetry property, one that incorporates the mutational bias, $\nu$:

\begin{proposition}
	\label{prop:MSSsymmetry}
	Let $\left\{p_{\left(R,\alpha\right)}\left(\vx\right)\right\}_{\left(R,\alpha\right)}$ be a replacement rule representing neutral drift and let the mutation probability $u>0$ be fixed. Then, for each state $\vx \in \left\{0,1\right\}^G$, the value of $\pi_\MSS\left(\vx\right)$ with mutational bias $\nu$ equals the value $\pi_\MSS\left(\bar{\vx}\right)$ with mutational bias $1-\nu$.
\end{proposition}

We omit the proof, which is similar to that Proposition \ref{prop:symmetry} but uses the recurrence relations \eqref{eq:MSSrecur} in place of \eqref{eq:RMCrecur}. Proposition \ref{prop:MSSsymmetry} implies in particular that $\E_\MSS\left[x_g\right] = \nu$ for each site $g$. The apparent discrepancy between the results $\E_\RMC\left[x_g\right] = 1/2$ and $\E_\MSS\left[x_g\right] = \nu$ can be resolved by recalling that, as $u \to 0$, the MSS distribution becomes concentrated on state $\vA$ (with probability approaching $\nu$) and state $\va$ (with probability approaching $1-\nu$); in contrast, the RMC distribution excludes these monoallelic states and is independent of $\nu$.

\section{Weak selection}\label{sec:weaksel}
We say that selection is weak if the process of natural selection between $A$ and $a$ approximates neutral drift. Weak selection is mathematically convenient because it allows the use of perturbative techniques; it is also biologically relevant since, for many systems of interest, mutations have a relatively small effect on reproductive success.

\subsection{Formalism} \label{sec:weakseldef}
To formalize the notion of weak selection, we introduce a selection strength parameter, $\delta$, which takes values in some half-open neighborhood $\left[0,\varepsilon\right)$ of zero. We consider a $\delta$-indexed family of replacement rules, subject to the following assumption:

\begin{assumption}[Assumptions for weak selection] For each replacement event $\left(R,\alpha\right)$, the probabilities $p_{\left(R,\alpha\right)}\left(\vx\right)$ satisfy the following:
	\label{ass:weak}
	\renewcommand{\labelenumi}{(\alph{enumi})}
	\begin{enumerate}
		\item $p_{\left(R,\alpha\right)}\left(\vx\right)$ varies smoothly with respect to $\delta\in\left[0,\varepsilon\right)$ for each state $\vx$;
		\item $p_{\left(R,\alpha\right)}\left(\vx\right)$ is independent of $\vx$ for $\delta =0$.
	\end{enumerate}
\end{assumption}

Part (b) guarantees that the replacement rule represents neutral drift when $\delta =0$. For now, we do not require that Assumption \ref{ass:neutral} be satisfied for all values of $\delta$. Assumption \ref{ass:neutral} will appear later as a condition of Corollary \ref{cor:delhatselweak} below.

The following proposition shows that, under Assumption \ref{ass:weak}, other fundamental quantities of interest also vary smoothly with respect to $\delta$:
\begin{proposition}\label{lem:piweak}
	For a $\delta$-indexed family of replacement rules, $\left\{p_{\left(R,\alpha\right)}\left(\vx\right)\right\}_{\left(R,\alpha\right)}$, satisfying Assumption \ref{ass:weak}, and for any mutational bias, $0<\nu <1$,
	\renewcommand{\labelenumi}{(\alph{enumi})}
	\begin{enumerate}
		\item $\rho_A$, $\rho_a$, and $\pi_\RMC\left(\vx\right)$ for each $\vx \in \left\{0,1\right\}^G$ are smooth functions of $\delta \in \left[0,\varepsilon\right)$;
		\item For each $\vx \in \left\{0,1\right\}^G$, $\pi_\MSS\left(\vx\right)$ extends uniquely to a smooth function of $\left(u,\delta\right) \in \left[0,1\right] \times \left[0,\varepsilon\right)$;
		\item For each $\vx \in \left\{0,1\right\}^G \setminus \left\{\va, \vA\right\}$, $\Prob_\MSS\left[\vX = \vx\ |\ \vX \notin \left\{\va, \vA\right\}\right]$ extends uniquely to a smooth function of $\left(u,\delta\right) \in \left[0,1\right] \times \left[0,\varepsilon\right)$.
	\end{enumerate}
\end{proposition}
\begin{proof}
	We first observe that, from the definition of the evolutionary Markov chain and the formalism for weak selection, the transition probabilities $P_{\vx \to \vy}$ are smooth functions of $\left(u,\delta\right) \in \left[0,1\right] \times \left[0, \varepsilon\right)$.
	
	The fixation probabilities, $\rho_A$ and $\rho_a$, being absorption probabilities for a finite Markov chain, are bounded, rational functions of the transition probabilities (see, for example, Theorem 3.3.7 of \citealp{kemeny1960finite}), and are therefore  smooth functions of $\delta \in \left[0, \varepsilon\right)$.
	
	We turn now to the stationary probabilities $\pi_\MSS\left(\vx\right)$. The system of Eqs.~\eqref{eq:MSSsystem} and \eqref{eq:augment} define a unique continuous extension of $\pi_\MSS\left(\vx\right)$ to $0 \leqslant u \leqslant 1$. Since the equations of this system vary smoothly with $u$ and $\delta$ and the solution is  unique, this extension of $\pi_\MSS\left(\vx\right)$ is smooth in $\left(u,\delta\right) \in \left[0,1\right] \times \left[0, \varepsilon\right)$.
	
	The argument for $\Prob_\MSS\left[\vX = \vx\ |\ \vX \notin \left\{\va, \vA\right\}\right]$ is similar, except that the relevant system of equations is Eq.~\eqref{eq:conditionalrecur}---which is replaced by Eq.~\eqref{eq:RMCrecur} for $u=0$---together with the additional equation $\sum_{\vx \notin \left\{\va, \vA\right\}} \Prob_\MSS\left[\vX = \vx\ |\ \vX \notin \left\{\va, \vA\right\}\right] = 1$. This system of equations has a unique solution for each $\left(u,\delta\right) \in \left[0,1\right] \times \left[0, \varepsilon\right)$, which coincides with $\{\Prob_\MSS\left[\vX = \vx\ |\ \vX \notin \left\{\va, \vA\right\}\right]\}$ for $u>0$ and with $\{ \pi_\RMC\left(\vx\right) \}$ for $u=0$. Thus, for each $\vx \notin \left\{\va, \vA\right\}$, $\Prob_\MSS\left[\vX = \vx\ |\ \vX \notin \left\{\va, \vA\right\}\right]$ extends uniquely to a smooth function of $\left(u,\delta\right) \in \left[0,1\right] \times \left[0, \varepsilon\right)$, which coincides with $\pi_\RMC\left(\vx\right)$ at $u=0$.
\end{proof}

We will study weak selection as a perturbation of neutral drift ($\delta=0$). In formulating weak-selection expansions of various quantities, we will use a circle (${}^\circ$) to indicate the value at $\delta =0$ and a prime ($'$) to denote the first-order coefficient in $\delta$ as $\delta \to 0^+$. (In light of  Assumption \ref{ass:weak}b, this use of ${}^\circ$ is consistent with the previous use in Section \ref{sec:neutral}.) For example, we have the following weak selection expansions:
\begin{subequations}
	\begin{align}
		p_{\left(R,\alpha\right)}\left(\vx\right) & = p^\circ_{\left(R,\alpha\right)} + \delta p'_{\left(R,\alpha\right)}\left(\vx\right) + \mathcal{O}\left(\delta^2\right) ; \\
		\label{eq:piRMCweak}
		\pi_\RMC\left(\vx\right) & = \pi^\circ_\RMC\left(\vx\right) + \delta \pi'_\RMC\left(\vx\right) + \mathcal{O}\left(\delta^2\right) .
	\end{align}
\end{subequations}
We say that a statement holds \emph{under weak selection} if it holds to first order in $\delta$ as $\delta \to 0^+$. We deal only with first-order expansions here; for the mathematical theory of higher-order perturbations of a Markov chain, see \cite{silvestrov2017nonlinearly}.

\subsection{Success criteria for weak selection}
Turning now to quantities describing selection, we have the following weak-selection expansion of fitness:
\begin{equation}
	\label{eq:fitweak}
	w_g\left(\vx\right) = v_g + \delta w'_g\left(\vx\right) + \mathcal{O}\left(\delta^2\right) ,
\end{equation}
with
\begin{align}
	w'_g\left(\vx\right) & = \hat{b}_g'\left(\vx\right) - \hat{d}_g'\left(\vx\right) = \sum_{h \in G} \left(e_{gh}'\left(\vx\right) v_h - e_{hg}'\left(\vx\right) v_g \right).
\end{align}
Above and throughout this section, the reproductive values $v_g$ are understood to be computed at $\delta = 0$, and not to vary with $\delta$. We note that, in light of Eq.~\eqref{eq:fitsum}, $\sum_{g \in G} w_g' \left(\vx\right) = 0$ for each state, $\vx$.

For the RV-weighted change due to selection, $\delhatsel \left(\vx\right)$, we note that Theorem \ref{thm:delhatselneut} implies that $\delhatsel^\circ \left(\vx\right) =0$ for all states $\vx$. We therefore have the weak-selection expansion
\begin{equation}
	\label{eq:delhatselexpand}
	\delhatsel \left(\vx\right) = \delta \delhatsel' \left(\vx\right) + \mathcal{O}\left(\delta^2\right) ,
\end{equation}
with
\begin{align}\label{eq:delhatselweak}
	\delhatsel' \left(\vx\right) & = \sum_{g\in G}x_{g}w_{g}'\left(\vx\right) \nonumber \\
	& = \sum_{g,h \in G}x_g \left( e'_{gh}\left(\vx\right) v_h - e'_{hg}\left(\vx\right) v_g \right) \nonumber \\
	& = \frac{1}{2} \sum_{g,h \in G}\left(x_g - x_h\right) \left( e'_{gh}\left(\vx\right) v_h - e'_{hg}\left(\vx\right) v_g \right) .
\end{align}

Our second main result is a weak-selection analogue of Theorems \ref{thm:rhodelsel} and \ref{thm:rhodelselhat}. It proves the equivalence of four success criteria: one based on fixation probability, one based on expected frequency, and two based on change due to selection.
\begin{theorem}
	\label{thm:delhatselweak}
	For any replacement rule $\left\{p_{\left(R, \alpha\right)} \left(\vx\right) \right\}_{\left(R,\alpha\right)}$ satisfying Assumption \ref{ass:weak}, and any mutational bias $\nu$, the following success criteria are equivalent:
	\renewcommand{\labelenumi}{(\alph{enumi})}
	\begin{enumerate}
		\item $\rho_A > \rho_a$ under weak selection;
		\item $\displaystyle \lim_{u \to 0} \E_\MSS\left[x\right] > 
		\frac{\nu b\left(\va\right)}{\nu b\left(\va\right) + \left(1-\nu\right) b\left(\vA\right)} $ under weak selection;
		\item $\displaystyle  \E_\RMC^\circ\left[\delhatsel '\right] 
		> K\nu\left(1-\nu\right) \left( \hat{b}'\left(\mathbf{A}\right) -\hat{b}'\left(\mathbf{a}\right)
		- \frac{\hat{b}^{\circ}}{b^\circ}\left(b'\left(\mathbf{A}\right) -b'\left(\mathbf{a}\right)\right) \right)$;
		\item $\displaystyle \frac{d}{du} \, \E_\MSS^\circ\left[\delhatsel '\right]  \Big|_{u=0}
		> \nu\left(1-\nu\right) \left( \hat{b}'\left(\mathbf{A}\right) -\hat{b}'\left(\mathbf{a}\right)
		- \frac{\hat{b}^{\circ}}{b^\circ}\left(b'\left(\mathbf{A}\right) -b'\left(\mathbf{a}\right)\right) \right)$,
	\end{enumerate}
	with $K$ defined as in Eq.~\eqref{eq:Kdef}.
\end{theorem}

Conditions (a) and (b) above are the weak-selection versions of the corresponding criteria in Theorem \ref{thm:rhodelsel}. Conditions (c) and (d) involve expectations of $\delhatsel '\left(\vx\right)$ over the neutral RMC and MSS distributions, respectively. However, these latter conditions also involve terms on the right-hand side that were not seen in our previous results. To gain intuition for these additional terms, it is helpful to note that Condition (c) can be rewritten as
\begin{equation}
	\label{eq:weakrewrite}
	\E_\RMC^\circ\left[\delhatsel '\right]  
	+ Kb^\circ \nu\left(1-\nu\right) \frac{d}{d\delta} \left( \frac{\hat{b}\left(\va\right)}{b\left(\va\right)} - \frac{\hat{b}\left(\vA\right)}{b\left(\vA\right)} \right)\bigg|_{\delta=0} > 0.
\end{equation}
Eq.~\eqref{eq:weakrewrite} reveals that $\E_\RMC^\circ\left[\delhatsel '\right]$ captures only part of the effects of weak selection on allele $A$. The other part, represented by the second term in Eq.~\eqref{eq:weakrewrite}, has to do with the average reproductive value of offspring created in the monoallelic states. For example, the average reproductive value of new offspring in state $\va$ is $\hat{b}\left(\va\right) /b\left(\va\right)$. If this quantity increases with $\delta$, new $A$-mutants  arising in state $\va$ will have additional reproductive value for $\delta>0$, relative to the neutral drift ($\delta =0$) case. Such effects are accounted for in the second term of Eq.~\eqref{eq:weakrewrite}, or equivalently, in the right-hand sides of Conditions (c) and (d). In short, fitness-based quantities such as $\delhatsel$ account for only part of the direction of selection. This phenomenon is discussed in detail by \cite{tarnita2014measures}, who also proved the equivalence of (b) and (d) in the special case that $b(\va) = b(\vA)$ for all $\delta \geqslant 0$.

We also observe that Condition (b) above involves two limits, $u \to 0$ and $\delta \to 0$. These limits can be freely interchanged according to Proposition \ref{lem:piweak}, so there is no concern regarding limit orderings.

\begin{proof}[Proof of Theorem \ref{thm:delhatselweak}]
	We first show the equivalence of (a) and (b). Theorem \ref{thm:lowu} gives
	\begin{equation}
		\label{eq:EMSSx}
		\lim_{u \to 0} \E_\MSS\left[x\right] = \frac{\nu b\left(\va\right)\rho_{A}}{ \nu b\left(\va\right) \rho_{A}+\left(1-\nu\right) b\left(\vA\right)\rho_{a}} .
	\end{equation}
	Note that for $\delta =0$, we have $\rho_A = \rho_a = \hat{b}^{\circ}/(n b^\circ)$ by Theorem \ref{thm:neutfixprob}, and both sides of Condition (b) become equal to $\nu$. It therefore suffices to show that the first $\delta$-derivatives of $\rho_A - \rho_a$ and $\lim_{u \to 0} \E_\MSS\left[x\right]- \nu b\left(\va\right) /\left(\nu b\left(\va\right) + \left(1-\nu\right) b\left(\vA\right)\right)$ have the same sign at $\delta = 0$. Applying Eq.~\eqref{eq:EMSSx}, we compute:
	\begin{align}
		\frac{d}{d\delta} &\left( \lim_{u \to 0} \E_\MSS\left[x\right] 
		-  \frac{\nu b\left(\va\right)}{\nu b\left(\va\right) + \left(1-\nu\right) b\left(\vA\right)} \right)\Big|_{\delta=0} \nonumber \\
		& = \nu\left(1-\nu\right) \left( \frac{nb^\circ }{\hat{b}^\circ} \left( \rho_A' - \rho_a'\right)
		+ \frac{b'\left(\va\right) - b'\left(\vA\right)}{b^\circ} \right) - \nu\left(1-\nu\right) \frac{b'\left(\va\right) - b'\left(\vA\right)}{b^\circ} \nonumber \\
		& = \nu\left(1-\nu\right) \frac{nb^\circ }{\hat{b}^\circ} \left( \rho_A' - \rho_a'\right) .
	\end{align}
	This shows (a) $\Leftrightarrow$ (b). We now show (a) $\Leftrightarrow$ (c). From Eq.~\eqref{eq:RMCMSS}, we have
	\begin{align}
		\nonumber
		\E_\RMC \left[  \delhatsel \right ] & = K \lim_{u \to 0}\E_\MSS \left [ \sum_{g \in G} \left(x_g-\nu\right) \hat{b}_g \right ] \nonumber \\
		& = K \left(\left(1-\nu\right)\hat{b}\left(\vA\right)\lim_{u\rightarrow 0}\pi_{\MSS}\left(\vA\right) -\nu\hat{b}\left(\va\right) \lim_{u\rightarrow 0}\pi_{\MSS}\left(\va\right) \right) \nonumber \\
		& = K \nu \left(1-\nu\right) \frac{\hat{b}\left(\vA\right) b\left(\va\right) \rho_A - \hat{b}\left(\va\right) b\left(\vA\right) \rho_a}{\nu b\left(\va\right)\rho_{A}+\left(1-\nu\right) b\left(\vA\right)\rho_{a}} ,
	\end{align}
	where the last line comes from Theorem~\ref{thm:lowu}. Differentiating both sides with respect to $\delta$ at $\delta =0$ gives
	\begin{multline}
		\label{eq:delhatselderiv1}
		\frac{d}{d\delta} \, \mathbb{E}_{\RMC}\left[\delhatsel\right] \Big\vert_{\delta =0} \\
		= K\nu\left(1-\nu\right) \left(nb^{\circ}\left(\rho_{A}'-\rho_{a}'\right) +\hat{b}'\left(\vA\right) -\hat{b}'\left(\va\right)
		- \frac{\hat{b}^{\circ}}{b^{\circ}}\left(b'\left(\vA\right) -b'\left(\va\right)\right)\right).
	\end{multline}
	We rewrite the left-hand side of this equation as
	\begin{align}
		\frac{d}{d\delta} \, \mathbb{E}_{\RMC}\left[\delhatsel\right] \Big\vert_{\delta =0}
		&= \sum_{\vx \in \left\{0,1\right\}^G\setminus\left\{\va ,\vA\right\}} 
		\frac{d}{d\delta} \left( \pi_\RMC\left(\vx\right) \, \delhatsel \left(\vx\right) \right) \Big\vert_{\delta =0} \nonumber \\
		& = \sum_{\vx \in \left\{0,1\right\}^G\setminus\left\{\va ,\vA\right\}} 
		\pi_\RMC^\circ\left(\vx\right) \delhatsel'\left(\vx\right) \nonumber \\
		& = \mathbb{E}_{\RMC}^\circ\left[\delhatsel'\right] . \label{eq:delhatselderiv2}
	\end{align}
	In the second line above, all terms of the form $\pi_\RMC'\left(\vx\right)\delhatsel^\circ\left(\vx\right)$ vanish since $\delhatsel^\circ\left(\vx\right) =0$ for all states $\vx$ by Theorem \ref{thm:delhatselneut}. Combining Eqs.~\eqref{eq:delhatselderiv1} and \eqref{eq:delhatselderiv2} gives the equivalence of (a) and (c). Finally, (a) and (d) are equivalent by Lemma \ref{lem:lowu}, completing the proof.
\end{proof}

If the conditions of Theorem \ref{thm:delhatselweak} are satisfied, we say that allele $A$ is \emph{favored under weak selection}. The power of Theorem \ref{thm:delhatselweak} lies in the fact that, for many models of interest, Conditions (c) and (d) are easier to evaluate than (a) and (b) because (c) and (d) involve expectations taken at neutrality ($\delta=0$), meaning that the probabilities of replacement events are independent of the state. At neutrality, the recurrence relations governing the RMC and MSS distributions simplify greatly, in many cases allowing conditions (c) or (d) to be simplified to closed form (see Section \ref{sec:examples} for examples).

Still, Conditions (c) and (d) are not as simple as one might hope. Evaluating them requires computing the full value (not just the sign) of either $\frac{d}{du} \E_\MSS^\circ\left[\delhatsel '\right]\big|_{u=0}$ or both $\E_\RMC^\circ\left[\delhatsel '\right] $ and $K$. However, if we assume that the replacement probabilities in the monoallelic states satisfy Assumption \ref{ass:neutral} and are independent of $\delta$, only the sign of $\frac{d}{du} \E_\MSS^\circ\left[\delhatsel '\right]\big|_{u=0}$ or $\E_\RMC^\circ\left[\delhatsel '\right]$ is needed, as shown below:

\begin{corollary}
	\label{cor:delhatselweak}
	Let $\left\{p_{\left(R, \alpha\right)} \left(\vx\right) \right\}_{\left(R,\alpha\right)}$ be a replacement rule satisfying Assumption \ref{ass:weak}. Suppose that $p_{\left(R,\alpha\right)}\left(\va\right)$ and $p_{\left(R,\alpha\right)}\left(\vA\right)$ are equal to each other and independent of $\delta$, for all replacement events $\left(R,\alpha\right)$ and all sufficiently small $\delta \geqslant 0$. Then, for any mutational bias $\nu$, the following success criteria are equivalent:
	\renewcommand{\labelenumi}{(\alph{enumi})}
	\begin{enumerate}
		\item $\rho_A > \rho_a$ under weak selection;
		\item $\lim_{u \to 0} \E_\MSS\left[ x \right] \; (= \lim_{u \to 0} \E_\MSS\left[ \hat{x} \right]) > \nu$ under weak selection;
		\item $\E_\RMC^\circ \left[ \delhatsel' \right] >0$;
		\item $\frac{d}{du} \E_\MSS^\circ\left[\delhatsel '\right]\big|_{u=0} > 0$.
	\end{enumerate}
\end{corollary}

The proof follows immediately from Theorem \ref{thm:delhatselweak}. Aspects of this result, in the special case that all sites have the same reproductive value, were proven in Theorem 2 of \cite{Eusociality} and in Eq.~(16) of \cite{van2015social}. Additionally, a number of instances of this result have been obtained for particular models \citep{RoussetBilliard,leturque2002dispersal,LessardFixation,TaylorFixProb,AntalPhenotype,wakano2013mathematical,debarre2014social,allen2017evolutionary}.

\section{Results for individuals}\label{sec:ind}
All of the above results apply equally to haploid, diploid, haplodiploid, or polyploid populations, because the analysis is based on genetic sites rather than individuals. However, any formalism for natural selection must reserve a special role for the individual, as the entity that carries and is shaped by its genetic material. Here, we connect the individual-level and gene-level perspectives with definitions and results that apply at the level of the individual.

We recall from Section \ref{sec:sites} that the set of genetic sites, $G$, is partitioned into a collection, $\left\{G_i \right\}_{i \in I}$, where $G_i$ is the set of sites residing in individual $i \in I$. The ploidy of individual $i$ is $n_i \coloneqq \left| G_i\right|$. The notation $g \sim h$ indicates that sites $g$ and $h$ reside in the same individual.

\subsection{Assumptions}
Moving to an individual-level perspective requires the introduction of additional assumptions. First, we assume that an individual's alleles survive or die all together (along with the individual itself). Thus, if one of an individual's genetic sites is replaced, then all of them are. This excludes the possibility that, for example, a virus causes a germline mutation in one of a diploid individual's alleles.

\begin{assumption}[Coherence of individuals]\label{ass:ind}
	If genetic sites $g,h \in G$ reside in the same individual, $g \sim h$, then for each state $\vx \in \left\{0,1\right\}^G$ and each replacement event $\left(R,\alpha\right)$ with $p_{\left(R,\alpha\right)}\left(\vx\right) >0$, either $g,h \in R$ or $g,h \notin R$.
\end{assumption}

The second assumption is that meiosis is fair: each of the alleles in an individual is equally likely to be passed on to offspring when this individual reproduces. This assumption excludes the possibility of meiotic drive \citep{sandler1957meiotic,lindholm2016ecology}. 

\begin{assumption}[Fair meiosis]\label{ass:fairmeiosis}
	Let $\left(R,\alpha_1\right)$ and $\left(R, \alpha_2\right)$ be replacement events with the same set $R \subseteq G$ of replaced individuals. If $\alpha_1\left(g\right) \sim \alpha_2\left(g\right)$ for all $g \in R$, then $p_{\left(R, \alpha_1\right)}\left(\vx\right) = p_{\left(R,\alpha_2\right)}\left(\vx\right)$ for each state $\vx \in \left\{0,1\right\}^G$.
\end{assumption}

In words, the probability of a replacement event depends only on the (individual) parent of each replaced site, not on which allele is inherited from that parent.

\subsection{Fitness and selection at the individual level}
As immediate consequences of Assumptions \ref{ass:ind} and \ref{ass:fairmeiosis}, we see that sites residing in the same individual have the same birth rate, death rate, reproductive value, and fitness:
\begin{lemma}
	\label{lem:indsame}
	Suppose the replacement rule $\left\{p_{\left(R,\alpha\right)}\left(\vx\right) \right\}_{\left(R, \alpha\right)}$ satisfies Assumptions \ref{ass:ind} and \ref{ass:fairmeiosis}. If sites $g, h \in G$ reside in the same individual, $g \sim h$, then $d_g\left(\vx\right) =d_h\left(\vx\right)$ and $e_{g \ell}\left(\vx\right) =e_{h \ell}\left(\vx\right)$ for each state $\vx \in \left\{0,1\right\}^G$ and each site $\ell \in G$. If, furthermore, Assumption \ref{ass:neutral} holds, then $v_{g}=v_{h}$ and $w_{g}\left(\vx\right) =w_{h}\left(\vx\right)$ for each $\vx \in \left\{0,1\right\}^G$.
	\renewcommand{\labelenumi}{(\alph{enumi})}
\end{lemma}

\begin{proof}
	Fix sites $g,h \in G$ with $g \sim h$. Assumption \ref{ass:ind} and Eq.~\eqref{eq:ddef} imply that that $d_{g}\left(\vx\right) =d_{h}\left(\vx\right)$ for all states $\vx \in \left\{0,1\right\}^G$. Assumption \ref{ass:fairmeiosis} implies that 
	\begin{equation}
		\sum_{\substack{\left(R,\alpha\right) \\ \alpha\left(\ell\right) =g}} p_{\left(R,\alpha\right)}\left(\vx\right) 
		= \sum_{\substack{\left(R,\alpha\right) \\ \alpha\left(\ell\right) =h}} p_{\left(R,\alpha\right)}\left(\vx\right) ,
	\end{equation}
	which is equivalent to $e_{g\ell}\left(\vx\right) = e_{h\ell}\left(\vx\right)$.
	
	If Assumption \ref{ass:neutral} holds, then the above arguments imply that $d_g^\circ = d_h^\circ$ and $e_{g\ell}^{\circ}=e_{h\ell}^{\circ}$ for all $\ell \in G$. It follows that $v_{g}=\frac{1}{d_{g}^{\circ}}\sum_{\ell\in G}e_{g\ell}^{\circ}v_{\ell}=\frac{1}{d_{h}^{\circ}}\sum_{\ell\in G}e_{h\ell}^{\circ}v_{\ell}=v_{h}$, as desired.
	
	The definitions in Eqs.~\eqref{eq:bhat} and \eqref{eq:dhat} now imply that $\hat{b}_g\left(\vx\right) =\hat{b}_h\left(\vx\right)$ and $\hat{d}_g\left(\vx\right) =\hat{d}_h\left(\vx\right)$ for all states $\vx$, from which it follows that $w_g\left(\vx\right) =w_h\left(\vx\right)$ for all $\vx$.
\end{proof}

Moving to individual-level quantities, we can identify the type of an individual $i \in I$, in state $\vx \in \left\{0,1\right\}^G$, by its fraction of $A$ alleles:
\begin{equation}
	\label{eq:Xdef}
	X_i \coloneqq \frac{1}{n_i} \sum_{g \in G_i} x_g.
\end{equation}
For example, a diploid heterozygote (genotype $Aa$) has $X_i = 1/2$. The reproductive value and fitness of individual $i \in I$ are defined by summing the corresponding quantities over all sites in $I$:
\begin{subequations}
	\label{eq:Wdef}
	\begin{align}
		V_i & \coloneqq \sum_{g \in G_i} v_g ; \\
		W_i\left(\vx\right) & \coloneqq \sum_{g \in G_i} w_g\left(\vx\right) .
	\end{align}
\end{subequations}

If Assumptions \ref{ass:neutral}, \ref{ass:ind}, and \ref{ass:fairmeiosis} hold, then all sites in the same individual have the same reproductive value and fitness according to Lemma \ref{lem:indsame}, and it follows that $V_i = n_i v_g$ and $W_i\left(\vx\right) = n_i w_g\left(\vx\right)$ for any $g \in G_i$.

We can use the above definitions to express the RV-weighted change due to selection, $\delhatsel$, using only quantities that apply at the level of the individual:

\begin{proposition}
	\label{prop:delselind}
	For any replacement rule satisfying Assumptions \ref{ass:neutral}, \ref{ass:ind}, and \ref{ass:fairmeiosis}, and any state $\vx$,
	\begin{equation}
		\delhatsel\left(\vx\right) = \sum_{i \in I} X_i \left(W_i\left(\vx\right) -V_i\right) .
	\end{equation}
\end{proposition}

\begin{proof}
	For each individual $i \in I$ we calculate
	\begin{align}
		X_i \left(W_i\left(\vx\right) -V_i\right) 
		& = \left( \frac{1}{n_i} \sum_{g \in G_i} x_g \right) \left( \sum_{h \in G_i} w_h(\vx) -   \sum_{h \in G_i} v_h\right) \nonumber \\
		& = \left( \sum_{g \in G_i} x_g \right) \left(  \frac{1}{n_i} \sum_{h \in G_i} \left(w_h\left(\vx\right) -v_h\right) \right) .
	\end{align}
	Lemma \ref{lem:indsame} implies that $\frac{1}{n_i} \sum_{h \in G_i} \left(w_h\left(\vx\right) -v_h\right) =w_g\left(\vx\right) -v_g$ for any $g \in G_i$; thus, the right-hand side above is equal to $\sum_{g \in G_i} x_g \left(w_g\left(\vx\right) -v_g\right)$. Now summing over all individuals $i \in I$ we have
	\begin{align}
		\sum_{i \in I} X_i \left(W_i\left(\vx\right) -V_i\right) & = \sum_{i \in I} \sum_{g \in G_i} x_g \left(w_g\left(\vx\right) -v_g\right) \nonumber \\
		& = \sum_{g \in G} x_g \left(w_g\left(\vx\right) -v_g\right) \nonumber \\
		& = \delhatsel\left(\vx\right) ,
	\end{align}
	as desired.
\end{proof}

We can use Proposition \ref{prop:delselind} to restate the criteria for success in Theorems \ref{thm:rhodelselhat} and \ref{thm:delhatselweak} using individual-level quantities. For example, if Assumptions \ref{ass:neutral}, \ref{ass:ind}, and \ref{ass:fairmeiosis}) hold, then $A$ is favored by selection if and only if $\sum_{i \in I} \E_\RMC \left [  X_i \left(W_i -V_i\right)\right] > 0$. Likewise, if Assumptions \ref{ass:weak}, \ref{ass:ind}, \ref{ass:fairmeiosis}, and the assumptions of Corollary \ref{cor:delhatselweak} hold, then $A$ is favored under weak selection if and only if $\sum_{i \in I} \E_\RMC^\circ \left [ X_i W_i' \right] > 0$.

If Assumptions \ref{ass:ind} and \ref{ass:fairmeiosis} do not hold, then sites in the same individual may have different reproductive value and/or fitness. These differences would not be reflected in the individual-level quantities $V_i$ and $W_i$; thus, the criteria for success may not be expressible in terms of these quantities.

\section{Examples} \label{sec:examples}
We illustrate the application of our formalism to two examples: evolutionary games on an arbitrary weighted graph \citep{allen2017evolutionary}, and a haplodiploid population in which alleles may affect males and females differently. In each case, we show how Conditions (c) and (d) of Corollary \ref{cor:delhatselweak} can be evaluated to obtain tractable conditions for success under natural selection.

\subsection{Games on graphs}\label{sec:games}
Evolutionary games on graphs \citep{NowakMay,blume1993statistical,SantosScaleFree,Ohtsuki,SzaboFath,chen2013sharp,allen2014games,debarre2014social,pena2016evolutionary} are a well-studied mathematical model for the evolution of social behavior. Individuals play a game with neighbors, and payoffs from this game determine reproductive success. Analytical results were first obtained for regular graphs \citep{Ohtsuki,Taylor,cox2013voter,chen2013sharp,allen2014games,debarre2014social,durrett2014spatial,pena2016evolutionary} and have recently been extended to arbitrary weighted graphs \citep{allen2017evolutionary,fotouhi2018conjoining}.

\subsubsection{Model}
Population structure is represented as a weighted (undirected) graph $G$, which we assume to be connected. Each vertex is always occupied by a single haploid individual (see Fig.~\ref{fig:graphs}). The edge weight between vertices $g,h \in G$, denoted $\omega_{gh} \geqslant 0$, indicates the strength of spatial relationship between these vertices. It is helpful to define the \emph{weighted degree} of vertex $g$ as $\omega_g = \sum_{h \in G} \omega_{gh}$. The random-walk step probability from $g$ to $h$ is denoted $p_{gh}=\omega_{gh}/\omega_g$. The probability that an $m$-step random walk from $g$ terminates at $h$ is denoted $p_{gh}^{\left(m\right)}$.

\begin{figure}
	\centering
	\includegraphics[width=\textwidth]{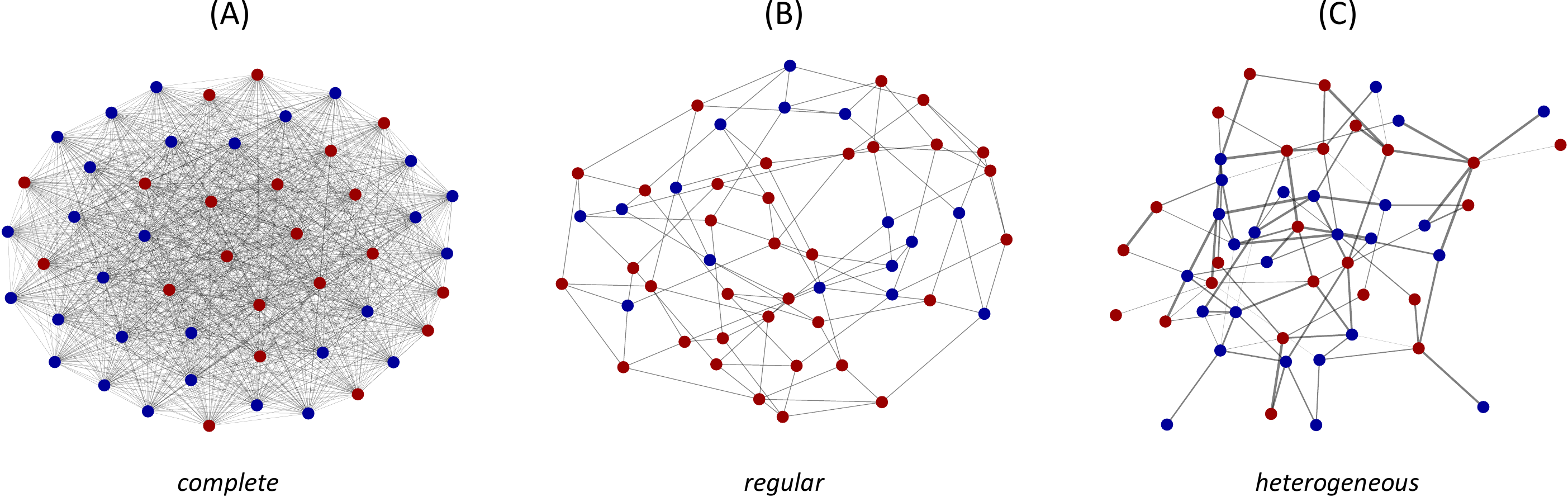}
	\caption{Three examples of graph-structured populations. In each population, blue indicates that a location is occupied by the allele $A$, while red indicates the allele at that location is $a$. (A) A complete graph, wherein each player is a neighbor of (i.e. shares a link with) every other player in the population. (B) A (non-complete) regular graph, for which all individuals have the same number of neighbors ($4$, in this instance). (C) A weighted heterogeneous graph, for which both the number of neighbors and the weights of the connections (given by the shading of the links) may vary from player to player. Historically, the analysis of evolutionary games on graphs has proceeded in order of increasing asymmetry, from (A) to (B) to (C).\label{fig:graphs}}
\end{figure}

In each state of the process, each individual interacts with each of its neighbors according to a game a $2 \times 2$ matrix game of the form
\begin{equation}
	\label{eq:game}
	\bordermatrix{%
		& A & a \cr
		A &\ f_{AA} & \ f_{Aa} \cr
		a &\ f_{aA} & \ f_{aa} \cr
	} .
\end{equation}

In each state $\vx$, each individual $g$ retains the edge-weighted average payoff it receives from neighbors, given by
\begin{equation}
	\label{eq:payoff}
	f_g (\vx) = \sum_{\ell \in G} p_{g\ell} \left(f_{AA} x_g x_\ell + f_{Aa} x_g \left(1-x_\ell\right) + f_{aA} \left(1-x_g\right) x_\ell + f_{aa} \left(1-x_g\right) \left(1-x_\ell\right) \right) .
\end{equation}
Game payoff is translated into fecundity by $F_g\left(\vx\right) = 1 + \delta f_g\left(\vx\right)$, where $\delta>0$ is a parameter representing the strength of selection.

Reproduction and replacement proceed according to a specified update rule \citep{Ohtsuki}. For Birth-Death (BD) updating, an individual $g$ is chosen at random, with probability proportional to reproductive rate $F_g\left(\vx\right)$, to (asexually) produce an offspring. This offspring replaces a random neighbor of $g$, chosen proportionally to edge weight $\omega_{gh}$. For Death-Birth (DB) updating, an individual $h$ is chosen, uniformly at random, to be replaced. Then, a neighbor $g$ is chosen, proportionally to $\omega_{gh} F_g\left(\vx\right)$, to produce an offspring to fill the vacancy. Mutations are resolved in accordance with our framework (Section \ref{sec:mutation}), leading to a new state.

\begin{figure}
	\centering
	\includegraphics[width=\textwidth]{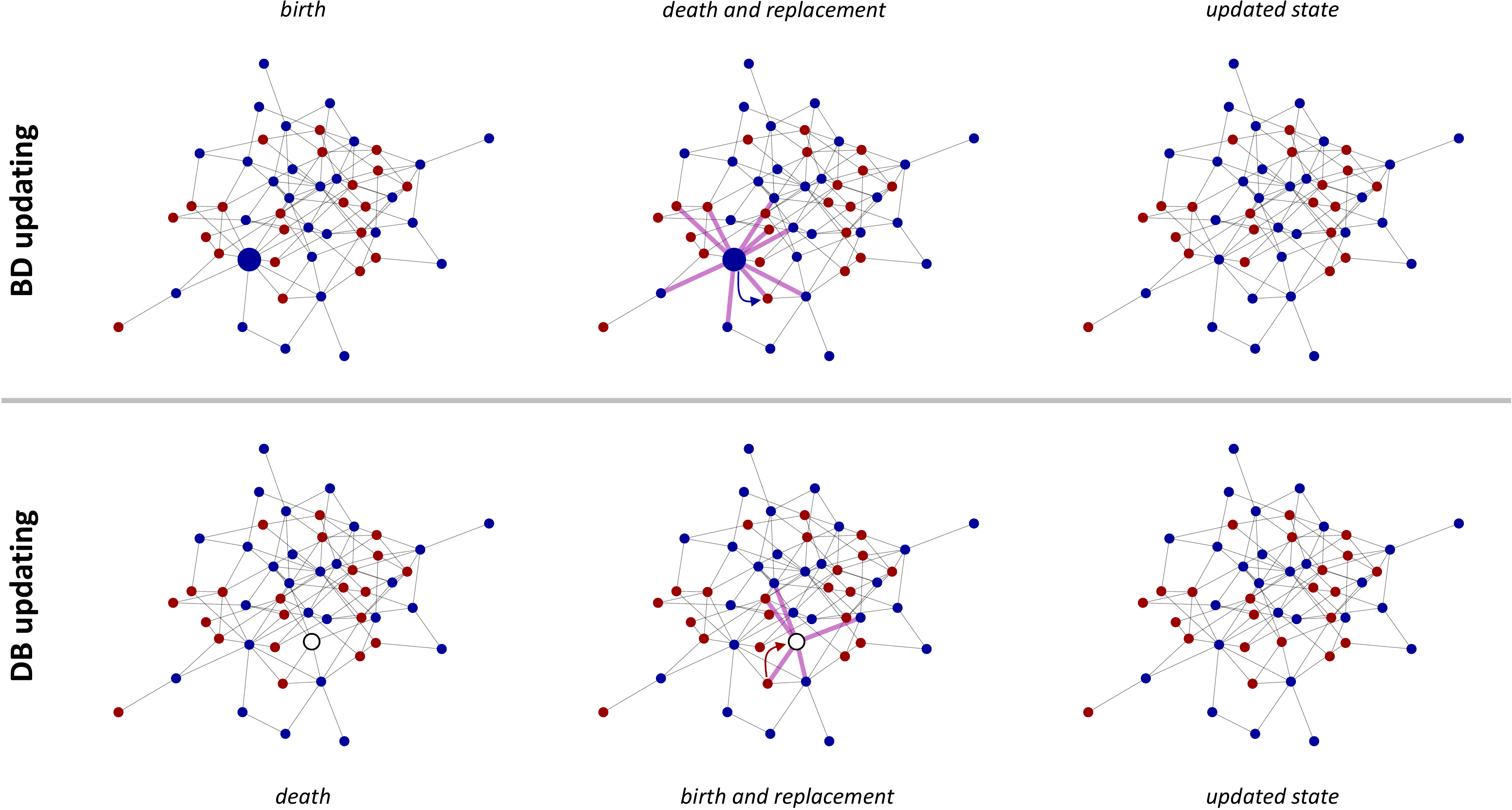}
	\caption{Birth-Death (BD) and Death-Birth (DB) updating on a graph. Each node is occupied by a haploid individual with allele $A$ (blue) or $a$ (red). Under BD updating, an individual is first chosen for reproduction with probability proportional to fecundity (large blue node). The offspring then replaces a neighbor, with probability determined by the weights of the outgoing links. Under DB updating, an individual is chosen for death uniformly at random from the population (empty node). The neighbors of this individual then compete to reproduce and fill the vacancy, with probability determined by both fecundity and the weights of the incoming links to the empty node.}
\end{figure}

\subsubsection{Basic quantities}
Since only one individual is replaced per time-step, any replacement event with positive probability has the form $(\{h\}, h \mapsto g)$, meaning that the occupant of site $h$ is replaced by the offspring of site $g$. The nonzero probabilities in the replacement rule are given by
\begin{equation}
	\label{eq:egames}
	p_{\left(\left\{h\right\},h \mapsto g\right)}\left(\vx\right) = 
	\begin{cases} 
		\displaystyle \left( \frac{1 + \delta f_g\left(\vx\right)}{\sum_{\ell \in G} \left(1 + \delta f_\ell\left(\vx\right)\right)} \right) p_{gh}
		& \text{for BD updating} , \\[5mm]
		\displaystyle \frac{1}{n} \left( \frac{\omega_{gh} \left (1+\delta f_g\left(\vx\right) \right)}{\sum_{\ell \in G} \omega_{\ell h} \left(1+\delta f_\ell \left(\vx\right)\right) }\right) 
		& \text{for DB updating} .
	\end{cases}
\end{equation}
Above, $f_g(\vx)$ is the payoff to $g$ in state $\vx$, given by Eq.~\eqref{eq:payoff}. We note also that $e_{gh}\left(\vx\right) = p_{\left(\left\{h\right\},h \mapsto g\right)}\left(\vx\right)$ for each $g,h \in G$.

For DB updating, new mutants are equally likely to appear at each vertex since each vertex is equally likely to be replaced in each state. Thus, the mutant appearance distribution for DB updating is
\begin{align}
	\mu_A \left(\vx\right) = \begin{cases} 
		\frac{1}{n} & \text{if $\vx = \mathbf{1}_{\left\{g\right\}}$ for some $g \in G$} , \\
		0 & \text{otherwise} .
	\end{cases}
\end{align}
The formula for $\mu_a\left(\vx\right)$ is analogous. For BD updating, the mutant appearance is distribution is nonuniform and given by 
\begin{equation}
	\label{eq:mutappearBD}
	\mu_A \left(\vx\right) = \begin{cases} 
		\frac{1}{n} \sum_{h \in G} p_{hg} & \text{if $\vx = \mathbf{1}_{\left\{g\right\}}$ for some $g \in G$} , \\
		0 & \text{otherwise} ,
	\end{cases}
\end{equation}
and analogously for $\mu_a\left(\vx\right)$.

Both BD and DB updating satisfy Assumption \ref{ass:weak}, and therefore have a well-defined neutral process. The probability that vertex $h$ is replaced by the offspring of vertex $g$, given by Eq.~\eqref{eq:egames}, reduces under the neutral process to 
\begin{equation}
	\label{eq:egamesneut}
	e^\circ_{gh} = 
	\begin{cases}
		p_{gh}/n & \text{for BD updating} , \\
		p_{hg}/n & \text{for DB updating} .
	\end{cases}
\end{equation}

The reproductive value of a vertex $g$ is proportional to its weighted degree for DB updating, and \emph{inversely} proportional to its weighted degree for BD updating:
\begin{align}
	v_g = 
	\begin{cases}
		\displaystyle n \frac{\omega_g^{-1}}{\tilde{\Omega}} & \text{for BD updating} , \\[5mm]
		\displaystyle n \frac{\omega_g}{\Omega} & \text{for DB updating,}
	\end{cases}
\end{align}
where $\Omega= \sum_{h \in G} \omega_h$ and $\tilde{\Omega}= \sum_{h \in G} \omega^{-1}_h$ are the total weighted degree and total inverse weighted degree, respectively (see Fig.~\ref{fig:graphRV}). These reproductive values were first discovered by \cite{maciejewski2014reproductive} for unweighted graphs and generalized to weighted graphs by \cite{allen2017evolutionary}.

\begin{figure}
	\centering
	\includegraphics[width=\textwidth]{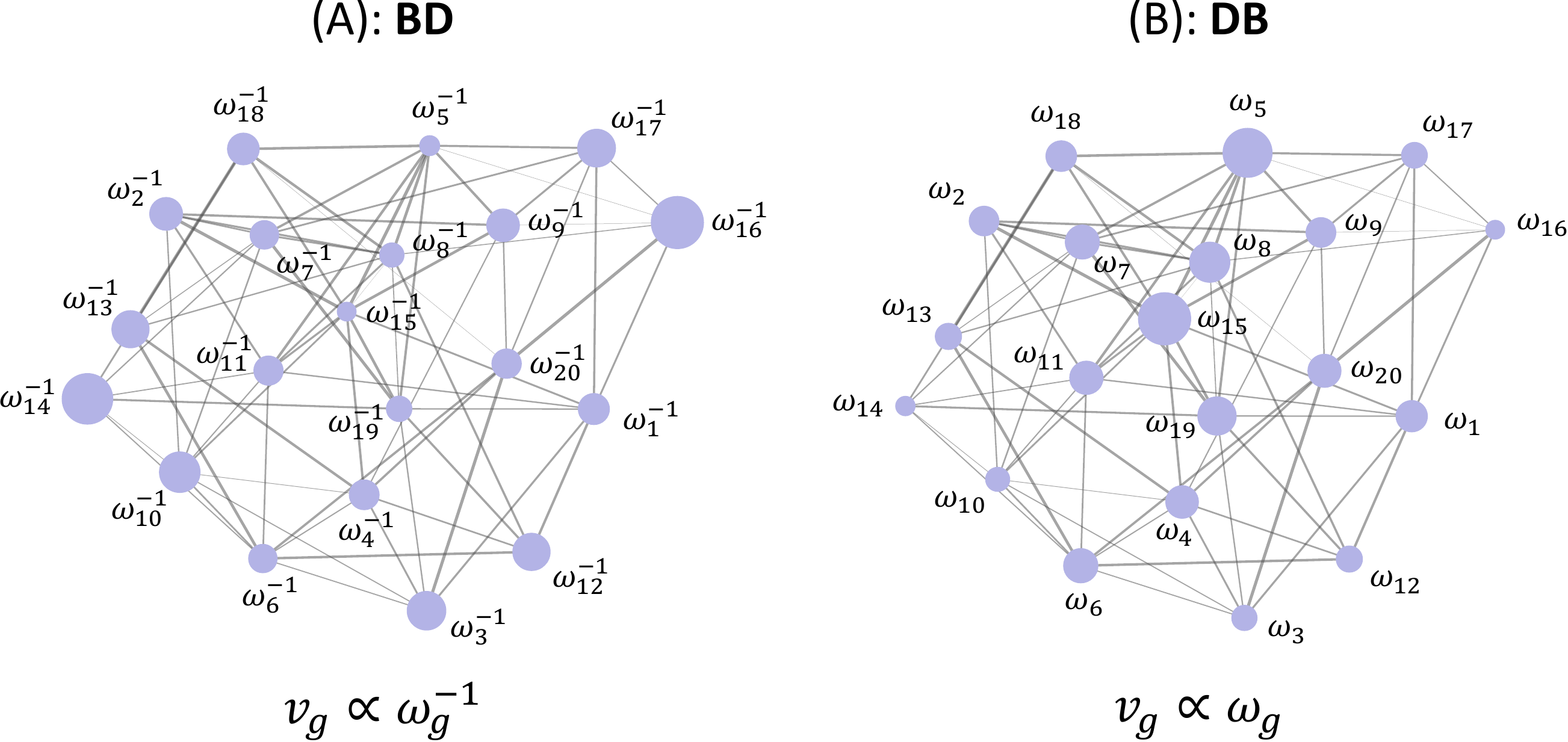}
	\caption{Reproductive values on a weighted graph for BD and DB updating. The \emph{weighted degree} $\omega_g$ of vertex $g$ is defined as $\omega_{g}\coloneqq\sum_{h\in G}\omega_{gh}$, where $\omega_{gh}$ is the weight of the edge between vertices $g$ and $h$. (A) For BD updating, the reproductive value of a vertex is inversely proportional to weighted degree: $v_g = n \omega_g^{-1}/\sum_{g \in G} \omega_g^{-1}$. (B) For DB updating the reproductive value of a vertex is directly proportional to its weighted degree: $v_g = n \omega_g/\sum_{g \in G} \omega_g$. For both panels, vertices are sized proportionally to reproductive value.
		\label{fig:graphRV}}
\end{figure}

A weak-selection expansion of Eq.~\eqref{eq:egames} yields
\begin{equation}
	\label{eq:egamesweak}
	e_{gh}' \left(\vx\right) = \begin{cases}
		\displaystyle \frac{p_{gh}}{n} \left( f_g\left(\vx\right) - \frac{1}{n} \sum_{\ell \in G}  f_\ell\left(\vx\right) \right) 
		& \text{for BD updating} , \\[7mm]
		\displaystyle \frac{p_{hg}}{n}  \left( f_g\left(\vx\right)- \sum_{\ell \in G} p_{h \ell } f_\ell \left(\vx\right) \right)
		& \text{for DB updating} .
	\end{cases}
\end{equation}

The fitness of vertex $g$ in state $\vx$ has the weak-selection expansion
\begin{align}
	w_g\left(\vx\right) = v_g + \delta w_g'\left(\vx\right) + \mathcal{O}\left(\delta^2\right) ,
\end{align}
where
\begin{equation}
	w_g'\left(\vx\right) = \begin{cases}
		\displaystyle \frac{1}{n \tilde{\Omega}}
		\sum_{h \in G} \frac{\omega_{gh}}{\omega_g \omega_h} \left( f_g\left(\vx\right) - f_h\left(\vx\right)\right)
		& \text{for BD updating} , \\[7mm]
		\displaystyle \frac{1}{n \Omega} \sum_{g,h \in G} 
		\omega_g p_{gh}^{\left(2\right)} \left( f_g\left(\vx\right) - f_h\left(\vx\right) \right) 
		& \text{for DB updating}.
	\end{cases}
\end{equation}

The first-order term of the RV-weighted change due to selection can be written
\begin{equation}
	\label{eq:delhatselgraphs}
	\delhatsel'\left(\vx\right) = \begin{cases}
		\displaystyle \frac{1}{2n \tilde{\Omega}} 
		\sum_{g,h \in G} \frac{\omega_{gh}}{\omega_g \omega_h} \left(x_g - x_h\right) \left( f_g\left(\vx\right) - f_h\left(\vx\right)\right)
		& \text{for BD updating} , \\[7mm]
		\displaystyle \frac{1}{2n \Omega} \sum_{g,h \in G} 
		\omega_g p_{gh}^{\left(2\right)} \left(x_g - x_h\right) \left( f_g\left(\vx\right) -  f_h\left(\vx\right) \right) 
		& \text{for DB updating} .
	\end{cases}
\end{equation}

\subsubsection{Condition for success: Birth-Death}\label{sec:BD}
Applying Corollary \ref{cor:delhatselweak}, we obtain that for BD, type $A$ is favored under weak selection if and only if
\begin{equation}
	\label{eq:BDcond1}
	\sum_{g,h \in G} \frac{\omega_{gh}}{\omega_g \omega_h} \E_\RMC^\circ \left[ \left(x_g - x_h\right) \left( f_g\left(\vx\right) - f_h\left(\vx\right) \right)\right] > 0 .
\end{equation}

To express Condition \eqref{eq:BDcond1} in terms of the entries of the payoff matrix \eqref{eq:game}, we use Eq.~\eqref{eq:payoff} to calculate
\begin{align}
	\label{eq:payoffdiff1}
	& \E_\RMC^\circ \left [ \left(x_g - x_h\right)  \left(f_g\left(\vx\right) -f_h\left(\vx\right) \right) \right] \nonumber \\
	& = \sum_{\ell \in G}  \Big( f_{AA} \left(p_{g\ell} \E_\RMC^\circ \left[x_g \left(1-x_h\right) x_\ell\right] + p_{h\ell} \E_\RMC^\circ \left[\left(1-x_g\right) x_h x_\ell\right]\right)) \nonumber \\
	& \qquad + f_{Aa} \left( p_{g\ell} \E_\RMC^\circ \left[x_g\left(1-x_h\right) \left(1-x_\ell\right)\right] + p_{h\ell} \E_\RMC^\circ \left[\left(1-x_g\right) x_h \left(1-x_\ell\right)\right] \right) \nonumber \\
	& \qquad - f_{aA} \left( p_{g\ell} \E_\RMC^\circ \left[\left(1-x_g\right) x_h  x_\ell\right] + p_{h\ell} \E_\RMC^\circ \left[x_g \left(1-x_h\right) x_\ell\right]\right) \nonumber \\
	& \qquad - f_{aa} \left( p_{g\ell} \E_\RMC^\circ \left[ \left(1-x_g\right) x_h \left(1-x_\ell\right)\right] + p_{h\ell} \E_\RMC^\circ \left[x_g \left(1-x_h\right) \left(1-x_\ell\right)\right]\right) \Big) .
\end{align}

We can reduce to pairwise quantities by noting that Proposition \ref{prop:symmetry} implies
\begin{align}
	\E_\RMC^\circ \left[x_g x_h x_\ell\right] = \E_\RMC^\circ \left[\left(1-x_g\right)\left(1- x_h\right)\left(1- x_\ell\right)\right],
\end{align}
which leads to the identity
\begin{align}
	\E_\RMC^\circ \left[x_g x_h x_\ell\right] = \frac{1}{2}\left( \E_\RMC^\circ \left[x_gx_h\right] + \E_\RMC^\circ \left[x_g x_\ell\right] + \E_\RMC^\circ \left[x_h x_\ell\right] \right) - \frac{1}{4} .
\end{align}

Applying this identity, Eq.~\eqref{eq:payoffdiff1} reduces to
\begin{multline}
	\label{eq:payoffdiff2}
	\E_\RMC^\circ \left[ \left(x_g - x_h\right) \left(f_g\left(\vx\right) -f_h\left(\vx\right)\right) \right] 
	= \frac{1}{2}   \Big(\left(f_{AA} + f_{Aa} - f_{aA} - f_{aa}\right) \left(1 - 2\E_\RMC^\circ \left[x_gx_h\right] \right)\\
	+ \left(f_{AA} - f_{Aa} + f_{aA} - f_{aa}\right) 
	\sum_{\ell \in G} \left(p_{g\ell} - p_{h \ell}\right)\left( \E_\RMC^\circ \left[x_g x_\ell\right] - \E_\RMC^\circ \left[x_h x_\ell\right] \right) \Big).
\end{multline}
Substituting and rearranging, we can rewrite Condition \eqref{eq:BDcond1} as 
\begin{multline}
	\label{eq:BDcond2}
	\sum_{g,h \in G} \frac{\omega_{gh}}{\omega_g \omega_h} \Big(\left(f_{AA} + f_{Aa} - f_{aA} - f_{aa}\right) \left( \frac{1}{2} - \E_\RMC^\circ \left[x_gx_h\right] \right)\\
	+ \left(f_{AA} - f_{Aa} + f_{aA} - f_{aa}\right) 
	\sum_{\ell \in G} p_{g\ell} \left( \E_\RMC^\circ \left[x_g x_\ell\right] - \E_\RMC^\circ \left[x_h x_\ell\right] \right) \Big ) > 0.
\end{multline}

It remains to describe how to compute the pairwise expectations $\E_\RMC^\circ\left[x_gx_h\right]$. Let us define the centered variables $\underline{x}_g = x_g - \nu$. Working through the possibilities of a single replacement event under neutral drift, one arrives at the following recurrence relation: for all pairs of sites $g \neq h$,
\begin{equation}
	\label{eq:MSSrecurBD}
	\E_\MSS^\circ \left[\underline{x}_g\underline{x}_h\right]  = \frac{1-u}{\sum_{\ell \in G} \left(p_{\ell g} + p_{\ell h}\right)}
	\sum_{\ell \in G} \left(p_{\ell g} \E_\MSS^\circ \left[\underline{x}_\ell \underline{x}_h\right] + p_{\ell h}\E_\MSS^\circ\left[ \underline{x}_g \underline{x}_\ell\right] \right) .
\end{equation}

Define the state function
\begin{equation}
	\label{eq:phigh}
	\phi_{gh}\left(\vx\right) = \underline{x}_g\underline{x}_h - \frac{1}{\sum_{\ell \in G} \left(p_{\ell g} + p_{\ell h}\right)}
	\sum_{\ell \in G} \left(p_{\ell g} \underline{x}_\ell \underline{x}_h + p_{\ell h} \underline{x}_g \underline{x}_\ell \right) .
\end{equation}
From the recurrence relation \eqref{eq:MSSrecurBD}, we have 
\begin{equation}
	\label{eq:Ephi}
	\E_\MSS^\circ\left[\phi_{gh}\right] = -\frac{u}{1-u} \E_\MSS^\circ\left[\underline{x}_g\underline{x}_h\right]
\end{equation}
Note also that $\phi_{gh}\left(\va\right) = \phi_{gh}\left(\vA\right) = 0$. Therefore, by Lemma \ref{lem:lowu}, there exists $K>0$ such that, for all $g,h \in G$ with $g \neq h$,
\begin{align}
	\E_\RMC^\circ\left[\phi_{gh}\right] & = K \frac{d}{du} \bigg|_{u=0} \E_\MSS^\circ\left[\phi_{gh}\right] \nonumber \\
	& = K \frac{d}{du} \bigg|_{u=0} \left(  -\frac{u}{1-u} \E_\MSS^\circ\left[ \underline{x}_g\underline{x}_h\right] \right) \nonumber \\
	& = - K \E_\MSS^\circ\left[ \underline{x}_g\underline{x}_h\right]\bigg|_{u=0} \nonumber \\
	& = - K \nu \left(1-\nu\right).
\end{align}
Substituting in from Eq.~\eqref{eq:phigh}, we see that for all $g,h \in G$ with $g \neq h$,
\begin{align}
	\label{eq:RMCrecurBD}
	\E_\RMC^\circ \left[\underline{x}_g\underline{x}_h\right]  &= \frac{1}{\sum_{\ell \in G} \left(p_{\ell g} + p_{\ell h}\right)}
	\sum_{\ell \in G} \left(p_{\ell g} \E_\RMC^\circ \left[\underline{x}_\ell \underline{x}_h\right] 
	+ p_{\ell h}\E_\RMC^\circ\left[ \underline{x}_g \underline{x}_\ell\right] \right) \nonumber \\
	&\quad\quad - K \nu \left(1-\nu\right) .
\end{align}
We define the quantity $\tau_{gh}$, for all pairs $g,h \in G$, by
\begin{equation}
	\label{eq:taudef}
	\tau_{gh} = \frac{\frac{1}{2} - \E_\RMC^\circ \left[x_gx_h\right]}{K \nu \left(1-\nu\right)} 
	= \frac{\frac{1}{2} - \E_\RMC^\circ \left[\underline{x}_g\underline{x}_h\right] - \nu\left(1-\nu\right)}{K \nu \left(1-\nu\right)}.
\end{equation}
Note that $\tau_{gg}=0$ for all $g \in G$ since $\E_\RMC^\circ \left[x_g^2\right]=\E_\RMC^\circ \left[x_g\right]=\frac{1}{2}$ by Proposition \ref{prop:symmetry}. Eq.~\eqref{eq:RMCrecurBD} then leads to the recurrence relations:
\begin{equation}
	\label{eq:taurecur}
	\tau_{gh} = \begin{cases}
		1 + \frac{\sum_{\ell \in G} \left(p_{\ell g} \tau_{\ell h}
			+ p_{\ell h} \tau_{g \ell} \right)}
		{\sum_{\ell \in G} \left(p_{\ell g} + p_{\ell h}\right)} & g \neq h , \\
		0 & g=h .
	\end{cases}
\end{equation}
The system of linear equations \eqref{eq:taurecur} can be solved for the $\tau_{gh}$ on any given graph; the solution exists and is unique provided that $G$ is connected. The condition for success can then be rewritten in terms of the $\tau_{gh}$: By Corollary \ref{cor:delhatselweak}, $A$ is favored under weak selection if and only if
\begin{multline}
	\label{eq:BDcondtau}
	\sum_{g,h \in G} \frac{\omega_{gh}}{\omega_g \omega_h} \Big(\left(f_{AA} + f_{Aa} - f_{aA} - f_{aa}\right)\tau_{gh} \\
	+ (f_{AA} - f_{Aa} + f_{aA} - f_{aa}) 
	\sum_{\ell \in G} p_{g\ell}\left(\tau_{h \ell} - \tau_{g \ell}\right) \Big ) > 0 .
\end{multline}

The $\tau_{gh}$ can be interpreted as coalescence times for a discrete-time coalescing random walk (CRW) process; however, this interpretation is not necessary here. We note that Eq.~\eqref{eq:taurecur} differs from the recurrence for the CRW for BD updating presented in \cite{allen2017evolutionary}. The difference arises from the way initial mutants are introduced; in \cite{allen2017evolutionary} it is assumed that the location of the initial mutant is chosen uniformly among vertices; here, the initial mutant location is chosen according to the mutant appearance distribution given in Eq.~\eqref{eq:mutappearBD}.

We observe that Condition \eqref{eq:BDcondtau} can be written in the form
\begin{equation}
	\label{eq:sigma}
	\sigma f_{AA} + f_{Aa} > f_{aA} + \sigma f_{aa},
\end{equation}
with 
\begin{equation}
	\label{eq:sigmaBDtau}
	\sigma = \frac{ \displaystyle \sum_{g,h \in G} \frac{\omega_{gh}}{\omega_g \omega_h}
		\left( \tau_{gh} - \sum_{\ell \in G} \left(p_{g\ell} - p_{h \ell}\right) \tau_{g \ell} \right)}
	{\displaystyle \sum_{g,h \in G} \frac{\omega_{gh}}{\omega_g \omega_h}
		\left( \tau_{gh} + \sum_{\ell \in G} \left(p_{g\ell} - p_{h \ell}\right) \tau_{g \ell} \right)} .
\end{equation}
Eq.~\eqref{eq:sigma} is an instance of the Structure Coefficient Theorem \citep{Corina,NowakStructured,allen2013adaptive}, which states that, for a quite general class of evolutionary game models, the condition for success under weak selection takes the form \eqref{eq:sigma} for some ``structure coefficient," $\sigma$.

\subsubsection{Conditions for success: Death-Birth}
For DB updating, applying Corollary \ref{cor:delhatselweak} to Eq.~\eqref{eq:delhatselgraphs}, we obtain that $\rho_A > \rho_a$ under weak selection if and only if
\begin{equation}
	\label{eq:DBcond1}
	\sum_{g,h \in G} 
	\omega_g p_{gh}^{\left(2\right)} \E_\RMC^\circ \left[ \left(x_g - x_h\right) \left( f_g\left(\vx\right) - f_h\left(\vx\right) \right)\right] > 0.
\end{equation}
Applying Eq.~\eqref{eq:payoffdiff2}, this condition becomes
\begin{multline}
	\label{eq:DBcond2}
	\sum_{g,h \in G} \omega_g p_{gh}^{(2)} 
	\Bigg(\left(f_{AA} + f_{Aa} - f_{aA} - f_{aa}\right) \left(1 - 2\E_\RMC^\circ \left[x_gx_h\right] \right)\\
	+ \left(f_{AA} - f_{Aa} + f_{aA} - f_{aa}\right) 
	\sum_{\ell \in G} (p_{g\ell} - p_{h \ell})\left( \E_\RMC^\circ \left[x_g x_\ell\right] - \E_\RMC^\circ \left[x_h x_\ell\right] \right) \Bigg) > 0.
\end{multline}

To obtain the quantities $\E_\RMC^\circ \left[x_gx_h\right]$, we note that DB updating has a recurrence equation analogous to Eq.~\eqref{eq:MSSrecurBD}:
\begin{equation}
	\label{eq:MSSrecurDB}
	\E_\MSS^\circ \left[\underline{x}_g\underline{x}_h\right]  = \frac{1-u}{2}
	\sum_{\ell \in G} \left(p_{g\ell} \E_\MSS^\circ \left[\underline{x}_\ell \underline{x}_h\right] + p_{h\ell}\E_\MSS^\circ\left[ \underline{x}_g \underline{x}_\ell\right] \right) ,
\end{equation}
where $\underline{x}_g=x_g-\nu$ as before. Upon defining
\begin{equation}
	\label{eq:phighDB}
	\phi_{gh}(\vx) = \underline{x}_g\underline{x}_h - \frac{1}{2}
	\sum_{\ell \in G} \left(p_{g\ell} \underline{x}_\ell \underline{x}_h + p_{h\ell} \underline{x}_g \underline{x}_\ell \right),
\end{equation}
we find that Eq.~\eqref{eq:Ephi} holds for DB updating as well. Following the argument of Section \ref{sec:BD}, we obtain the following analogue of Eq.~\eqref{eq:RMCrecurBD}:
\begin{equation}
	\label{eq:RMCrecurDB}
	\E_\RMC^\circ \left[\underline{x}_g\underline{x}_h\right]  = \frac{1}{2}
	\sum_{\ell \in G} \left(p_{g \ell} \E_\RMC^\circ \left[\underline{x}_\ell \underline{x}_h\right] 
	+ p_{h\ell}\E_\RMC^\circ\left[ \underline{x}_g \underline{x}_\ell\right] \right)
	- K \nu \left(1-\nu\right).
\end{equation}
Defining the quantities $\tau_{gh}$ for $g,h \in G$ according to Eq.~\eqref{eq:taudef}, we arrive at the DB analogue of the recurrence relations \eqref{eq:taurecur}:
\begin{equation}
	\label{eq:taurecurDB}
	\tau_{gh} = \begin{cases}
		1 + \frac{1}{2}\sum_{\ell \in G} \left(p_{g \ell } \tau_{\ell h} +  p_{h \ell} \tau_{g \ell} \right) & g \neq h , \\
		0 & g=h .
	\end{cases}
\end{equation}
These $\tau_{gh}$ are precisely the pairwise coalescence times studied by \cite{allen2017evolutionary}, and they can be obtained for any given (weighted, connected) graph by solving Eq.~\eqref{eq:taurecurDB} as a linear system of equations. If we now define, for $m \geqslant 0$, 
\begin{align}
	\tau^{\left(m\right)} \coloneqq \sum_{g,h \in G} \frac{\omega_g}{\Omega} p_{gh}^{\left(m\right)} \tau_{gh} , 
\end{align}
then the condition for success \eqref{eq:DBcond2} can be rewritten as 
\begin{equation}
	\left(f_{AA} + f_{Aa} - f_{aA} - f_{aa}\right) \tau^{\left(2\right)} + \left(f_{AA} - f_{Aa} + f_{aA} - f_{aa}\right) \left(\tau^{\left(3\right)}-\tau^{\left(1\right)}\right) > 0 ,
\end{equation}
Again, the condition for success takes the form \eqref{eq:sigma}, with the structure coefficient for DB updating given by
\begin{equation}
	\label{eq:sigmaDB}
	\sigma = \frac{\tau^{(2)}+\tau^{(3)}-\tau^{(1)}}{\tau^{(2)}-\tau^{(3)}+\tau^{(1)}} ,
\end{equation}
which is exactly the result obtained by \cite{allen2017evolutionary}. The appearances of one-step, two-step, and three-step random walks in Eq.~\eqref{eq:sigmaDB} have an elegant interpretation in terms of interactions at various distances; see \cite{allen2017evolutionary}.

\subsection{Haplodiploid population}\label{sec:haplodiploid}
To illustrate the applicability of our framework beyond haploid populations, we analyze a simple model of evolution in a haplodiploid population with one male and one female parent per generation. This could represent, for example, a completely inbred, singly-mated, eusocial insect colony.

\subsubsection{Model}
The population consists of $N_{\F}$ diploid females and $N_{\M}$ haploid males. Thus, there are $n=2N_{\F}+N_{\M}$ genetic sites and $N=N_{\F}+N_{\M}$ individuals. Each genotype has an associated fecundity, denoted $F_{xy}$ for females and $F_{z}$ for males, with $x,y,z \in \left\{a,A\right\}$. (Here and throughout this example, $xy$ is understood as an unordered pair.) The fecundities for females are 
\begin{equation}
	\label{eq:femF}
	F_{aa} = 1 ; \qquad F_{Aa} = 1+\delta hs ; \qquad F_{AA} = 1+\delta s .
\end{equation}
Above, the parameter $s$ quantifies selection on the $A$ allele in females and  the parameter $h$ represents the degree of dominance. Fecundities for males are
\begin{equation}
	\label{eq:maleF}
	F_{a} = 1 ; \qquad F_{A} = 1+ \delta m ,
\end{equation}
where the parameter $m$ quantifies selection on the $A$ allele in males.

Each time-step, one female and one male are chosen at random, with probability proportional to fecundity, to be parents for the next generation. These parents produce a new generation of offspring, replacing all previous individuals. Females are produced sexually while males are produced asexually (parthenogenetically); thus, each female offspring inherits the allele of the male parent and one of the two alleles of the female parent, while each male offspring inherits one of the two alleles of the female parent.

\begin{figure}
	\centering
	\includegraphics[width=0.8\textwidth]{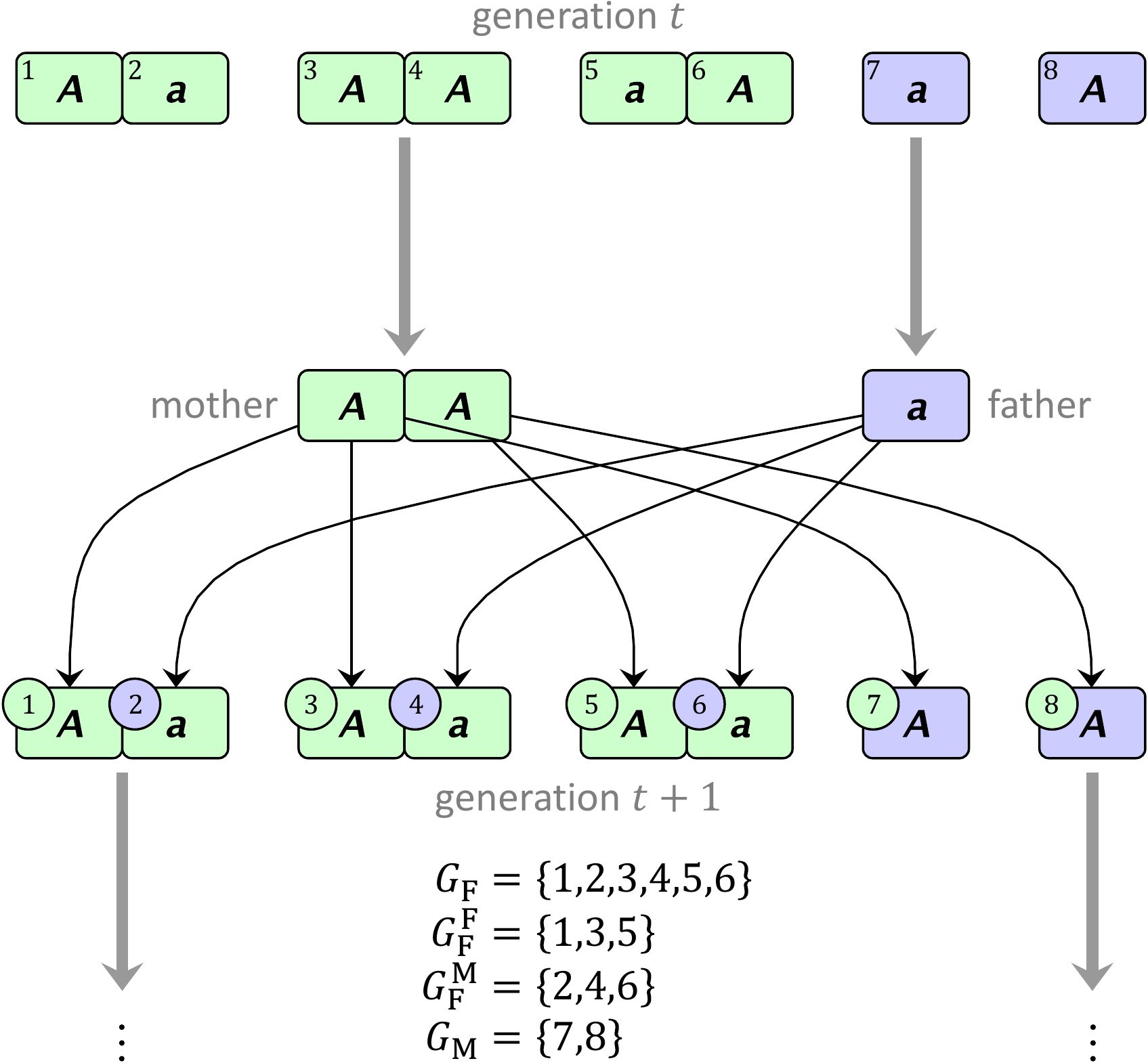}
	\caption{Model of selection in a haplodiploid population. Females (green) are diploid and males (purple) are haploid. At each time step, one female (mother) and one male (father) are selected to populate the subsequent generation. Each parent is chosen with probability proportional to its fecundity, which depends on its genotype according to Eqs.~\eqref{eq:femF}--\eqref{eq:maleF}. Each female in the next generation inherits one allele from the mother and one allele from the father, while each male inherits a single allele from the mother. The process then repeats. The genetic sites are numbered in their upper-left corner. In generation $t+1$, the color of the site label indicates the parent from which the allele was inherited (which, in females, determines whether the site is in $G_{\F}^{\F}$ or $G_{\F}^{\M}$).\label{fig:haplodiploid}}
\end{figure}

\subsubsection{Sites and replacement rule}
To translate this model into our formalism, we partition the set of sites $G$ as $G=G_\F \sqcup G_\M$, where $G_\F$ and $G_\M$ are the sets of sites in females and males, respectively. Similarly, we partition the set $I$ of individuals into females and males: $I = I_\F \sqcup I_\M$. It is notationally convenient (although not strictly necessary) to distinguish the sites in females according to which parent (male or female) they inherit alleles from. We therefore partition the sites in females as $G_\F= G_\F^\F \sqcup G_\F^\M$, where sites in $G_\F^\F$ house alleles from the female parent, and sites in $G_\F^\M$ house alleles from the male parent. 

For a given state $\vx \in \{0,1\}^G$, we denote the number of females of type $xy$ by $n_{xy}$, and the number of males of type $z$ by $n_{z}$, where $x,y,z \in \left\{A,a\right\}$. Clearly, $n_{aa}+n_{Aa}+n_{AA}=N_{\F}$ and $n_{a}+n_{A}=N_{\M}$ in every state.

We recall that, at each time-step, one male and one female are chosen to replace the entire population. Therefore, all replacement events $(R,\alpha)$ with nonzero probability have $R=G$; furthermore, there exists a pair of individuals $i \in I_\F$ and $j \in I_\M$, with genetic sites $G_i = \{g_i, g_i'\}$ and $G_j =\{g_j\}$, such that
\begin{equation}
	\begin{cases}
		\alpha\left(g\right) \in \{g_i, g_i'\} & \text{for } g \in G_{\M} \cup G_{\F}^{\F}, \\
		\alpha\left(g\right) =  g_j & \text{for } g \in  G_{\F}^{\M}.
	\end{cases}
\end{equation}
The probability that a particular replacement event $(G,\alpha)$ of the above form occurs in state $\vx$ can be written as
\begin{equation}
	p_{\left(G,\alpha\right)}(\vx) = \frac{1}{2^{N}} 
	\left( \frac{F_{x_{g_i}x_{g_i'}}} {n_{aa}F_{aa}+n_{Aa}F_{Aa} +n_{AA}F_{AA}} \right) 
	\left( \frac{F_{x_{g_j}}}{n_{a}F_a+n_{A}F_A} \right).
\end{equation}
Above, we have extended our notation for fecundity to numerical genotypes, so that $F_{11}=F_{AA}$, etc. The prefactor $1/2^N$ reflects the $2^{N}$ possible mappings from $G_\M \cup G_\F^\F$ to $\{g_i, g_i'\}$, representing the possible ways that each new offspring could inherit one of the two alleles in the mother.

\subsubsection{Evolutionary Markov chain}
We now establish some basic results for the evolutionary Markov chain. First, we note that the transition probabilities depend only the $5$-tuple \, $\left(n_{aa},n_{Aa},n_{AA},n_{a},n_{A}\right)$. We also observe that transitions consist of two steps: (i) selection of parents and (ii) production of offspring. The probabilities resulting from the second step depend only on the genotypes of the two chosen parents, for which there are six possibilities:
\begin{align}
	S_\mathrm{par} \coloneqq \Big\{ \left(aa,a\right) , \left(Aa,a\right) , \left(AA,a\right) , \left(aa,A\right) , \left(Aa,A\right) , \left(AA,A\right) \Big\}.
\end{align}
Transitions for the evolutionary Markov chain can therefore be written in the form
\begin{align}
	P_{\vx \rightarrow \vx'} &= \sum_{\left(xy,z\right)\in S_\mathrm{par}} P_{\vx \to \left(xy,z\right)} \, P_{\left(xy,z\right) \to \vx'} .
\end{align}
In other words, the transition matrix $\mathbf{M}$ for the evolutionary Markov chain factors as $\mathbf{M}=\mathbf{SR}$, where $\mathbf{S}=\left(P_{\vx \to (xy,z)}\right)$ represents the selection of parents and $\mathbf{R}=\left(P_{(xy,z) \to \vx}\right)$ represents (re)production of offspring.

The probabilities for selection (i.e.~the entries of $\mathbf{S}$) can be written as
\begin{equation}
	P_{\mathbf{x}\rightarrow\left(xy,z\right)} = \frac{n_{xy}F_{xy}}{n_{aa}F_{aa}+n_{Aa}F_{Aa} +n_{AA}F_{AA}} \,
	\frac{n_{z}F_z}{n_{a}F_a+n_{A}F_A}.
\end{equation}
The probabilities for reproduction (i.e. the entries of $\mathbf{R}$) can be written as
\begin{align}\label{eq:parentsToOffspring}
	P_{\left(xy,z\right)\rightarrow\mathbf{x}} &= \left(\left(1-q_{xy}\right)\left(1-q_{z}\right)\right)^{n_{aa}}\left(q_{xy}+q_{z}-2q_{xy}q_{z}\right)^{n_{Aa}} \left(q_{xy}q_{z}\right)^{n_{AA}} \nonumber \\
	&\quad \times \left(1-q_{xy}\right)^{n_{a}}q_{xy}^{n_{A}} .
\end{align}
Above, $q_{xy}$ (resp. $q_{z}$) is the probability that an allele inherited from a mother of type $xy$ (resp. a father of type $z$) is $A$, with mutation taken into account. Specifically,
\begin{subequations}
	\label{eq:haploqs}
	\begin{align}
		q_{aa} &= u\nu ; \\
		q_{Aa} &= \frac{1}{2}\left(1-u+2u\nu\right) ; \\
		q_{AA} &= 1-u\left(1-\nu\right) ; \\
		q_{a} &= u\nu ; \\
		q_{A} &= 1-u\left(1-\nu\right) .
	\end{align}
\end{subequations}

Note that the entries of $\mathbf{R}$ do not depend on the parameters quantifying selection ($h$, $m$, and $s$).

\subsubsection{Reproductive value and selection}
The marginal probability that $g$ transmits a copy of itself to $\ell$ is
\begin{equation}
	\label{eq:eglhaplo}
	e_{g\ell}\left(\mathbf{x}\right) = 
	\begin{cases}
		\displaystyle
		\frac{1}{2} \frac{F_{x_{g}x_{g'}}}{n_{aa}F_{aa}+ n_{Aa}F_{Aa}+ n_{AA}F_{AA}} 
		& g\in G_{\F},\ \ell\in G_{\M} \cup G_\F^\F, \\[5mm]
		\displaystyle
		\frac{F_{x_{g}}}{n_{a}F_{a}+ n_{A}F_{A}} & g\in G_{\M},\ \ell\in G_{\F}^\M , \\[5mm]
		\displaystyle
		0 & \text{otherwise.}
	\end{cases}
\end{equation}
Above, $g'$ denotes the other site in the same individual as site $g \in G_\F$. For neutral drift ($\delta=0$), Eq.~\eqref{eq:eglhaplo} reduces to
\begin{equation}
	\label{eq:eglhaploneut}
	e_{g\ell}^\circ = \begin{cases} 
		\frac{1}{2N_{\F}} & g\in G_{\F},\ \ell\in G_{\M} \cup G_\F^\F, \\
		\frac{1}{N_{\M}} & g \in G_{\M},\ \ell \in G_{\F}^\M , \\
		0 & \text{otherwise.}
	\end{cases}
\end{equation}

A straightforward calculation gives the mutant-appearance distribution,
\begin{align}
	\mu_{A}\left(\vx\right) &= 
	\begin{cases} 
		\frac{1}{n} & \text{if $\vx = \mathbf{1}_{\left\{g\right\}}$ for some $g \in G$} , \\
		0 & \text{otherwise} ,
	\end{cases}
\end{align}
and analogously for $\mu_{a}\left(\vx\right)$.

Substituting from Eq.~\eqref{eq:eglhaploneut} into the recurrence equation \eqref{eq:vrecur} for reproductive values,  using the fact that there are $2N_{\F}$ female genetic sites and $N_{\M}$ male sites, yields
\begin{align}
	v_{\F} &= \frac{1}{2} v_{\F} + \frac{N_{\M}}{2N_{\F}} v_{\M} ,
\end{align}
so $v_{\F}=\frac{N_{\M}}{N_{\F}}v_{\M}$. Since $\sum_{g\in G}v_{g}=n=2N_{\F}+N_{\M}$, we obtain
\begin{subequations}
	\begin{align}
		v_{\F} &= \frac{2N_{\F}+N_{\M}}{3N_{\F}} ; \\
		v_{\M} &= \frac{2N_{\F}+N_{\M}}{3N_{\M}} .
	\end{align}
\end{subequations}
Note that for an even sex ratio ($N_{\F}=N_{\M}$), all sites have equal reproductive value.

Applying Eq.~\eqref{eq:fitdef}, the fitness of site $g$ in state $\vx$ is given by
\begin{equation}
	w_{g}\left(\mathbf{x}\right) =
	\begin{cases} \displaystyle
		\left(\frac{2N_{\F}+N_{\M}}{3}\right) \frac{F_{x_gx_{g'}}}{n_{aa}F_{aa}+ n_{Aa}F_{Aa}+ n_{AA}F_{AA}}
		& g \in G_{\F} , \\[5mm]
		\displaystyle \left(\frac{2N_{\F}+N_{\M}}{3}\right) \frac{F_{x_g}}{n_{a}F_a+ n_{A}F_A}
		& g \in G_{\M} .
	\end{cases}
\end{equation}
On the individual level, the fitness of a female with genotype $xy$ is 
\begin{equation}
	W_{xy}(\vx) = 2 \left( \frac{2N_{\F}+N_{\M}}{3}\right) \frac{F_{xy}}{n_{aa}F_{aa}+ n_{Aa}F_{Aa}+ n_{AA}F_{AA}},
\end{equation}
and the fitness of a male with genotype $z$ is 
\begin{equation}
	W_{z}(\vx) = \left(\frac{2N_{\F}+N_{\M}}{3}\right) \frac{F_{z}}{n_{a}F_a+ n_{A}F_A}.
\end{equation}

The RV-weighted change due to selection is
\begin{multline}
	\delhatsel(\vx) = \left(\frac{2N_{\F}+N_{\M}}{3}\right) \times\\
	\left( \frac{ n_{Aa} F_{Aa} + 2 n_{AA}F_{AA}}{n_{aa}F_{aa}+ n_{Aa}F_{Aa}+ n_{AA}F_{AA}}  + \frac{ n_{A}F_A}{n_{a}F_a+n_{A}F_A} - \frac{n_{Aa}+2n_{AA}}{N_\F} - \frac{n_A}{N_\M} \right) .
\end{multline}
Substituting from Eqs.~\eqref{eq:femF} and \eqref{eq:maleF} and applying a weak-selection expansion, we obtain 
\begin{multline}
	\label{eq:delhatselhaplo}
	\delhatsel '\left(\mathbf{x}\right) = \left(\frac{2N_{\F}+N_{\M}}{3}\right) \times \\
	\left( \frac{\left(hn_{Aa}+2n_{AA}\right) s}{N_{\F}} -  \frac{\left(hn_{Aa}+n_{AA}\right)\left(n_{Aa}+2n_{AA}\right) s}{N_{\F}^{2}} +  \frac{mn_{a}n_{A}}{N_{\M}^{2}} \right).
\end{multline}

\subsubsection{Neutral stationary distributions}
The two-step nature of transitions allows us to define the \emph{parental Markov chain}, which describes the transition from one generation's parents to the next. The parental Markov chain has state space $S_\mathrm{par}$ and transition matrix $\mathbf{P}=\mathbf{RS}$. More explicitly, we can write
\begin{align}
	\label{eq:PRS}
	P_{(xy,z) \rightarrow (x'y',z')} &= \sum_{\vx \in \{0,1\}^N} P_{\left(xy,z\right) \to \vx'} \, P_{\vx \to \left(x'y',z'\right)} .
\end{align}

For neutral drift, probabilities for selection (entries of $\mathbf{S}$) are given by
\begin{equation}
	\label{eq:pxyzneut}
	P_{\vx\rightarrow\left(xy,z\right)}^{\circ} = \frac{n_{xy}}{N_{\F}}\frac{n_{z}}{N_{\M}}.
\end{equation}
Substituting from Eqs.~\eqref{eq:parentsToOffspring}, \eqref{eq:haploqs}, and \eqref{eq:pxyzneut} into Eq.~\eqref{eq:PRS} and applying multinomial expectations, we obtain the following neutral transition probabilities on the parental chain:
\begin{subequations}
	\begin{align}
		P_{\left(xy,z\right)\rightarrow\left(aa,a\right)}^{\circ} &= \left(1-q_{xy}\right)^{2}\left(1-q_{z}\right) ; \\
		P_{\left(xy,z\right)\rightarrow\left(Aa,a\right)}^{\circ} &= \Big(q_{xy}+q_{z}-2q_{xy}q_{z}\Big)\left(1-q_{xy}\right) ; \\
		P_{\left(xy,z\right)\rightarrow\left(AA,a\right)}^{\circ} &= q_{xy}q_{z}\left(1-q_{xy}\right) ; \\
		P_{\left(xy,z\right)\rightarrow\left(aa,A\right)}^{\circ} &= \left(1-q_{xy}\right)\left(1-q_{z}\right) q_{xy} ; \\
		P_{\left(xy,z\right)\rightarrow\left(Aa,A\right)}^{\circ} &= \Big(q_{xy}+q_{z}-2q_{xy}q_{z}\Big) q_{xy} ; \\
		P_{\left(xy,z\right)\rightarrow\left(AA,A\right)}^{\circ} &= q_{xy}^{2}q_{z} .
	\end{align}
\end{subequations}

Since the parental chain has only six states, one can directly solve for its stationary distribution, $\Pi_{\MSS}^{\circ}$:
\begin{subequations}
	\begin{align}
		\Pi_{\MSS}^{\circ}\left(aa,a\right) &= \frac{\left(1 - \nu\right)\left(\substack{1- 2u^{7}\nu^{2} + 3u^{7}\nu - u^{7} + 14u^{6}\nu^{2} - 21u^{6}\nu + 7u^{6} - 48u^{5}\nu^{2} \\ + 70u^{5}\nu - 23u^{5} + 96u^{4}\nu^{2} - 134u^{4}\nu + 43u^{4} - 110u^{3}\nu^{2} \\ + 143u^{3}\nu - 43u^{3} + 58u^{2}\nu^{2} - 61u^{2}\nu + 13u^{2} - 16u\nu + 11u}\right)}{1+11u+13u^{2}-43u^{3}+43u^{4}-23u^{5}+7u^{6}-u^{7}} ; \\
		\Pi_{\MSS}^{\circ}\left(Aa,a\right) &= \frac{4u\nu\left(1 - \nu\right)\left(\substack{3 + 20u - 29u\nu + 55u^{2}\nu - 48u^{3}\nu + 24u^{4}\nu \\ - 7u^{5}\nu + u^{6}\nu - 45u^{2} + 43u^{3} - 23u^{4} + 7u^{5} - u^{6}}\right)}{1+11u+13u^{2}-43u^{3}+43u^{4}-23u^{5}+7u^{6}-u^{7}} ; \\
		\Pi_{\MSS}^{\circ}\left(AA,a\right) &= \frac{u\nu\left(1 - \nu\right)\left(\substack{4+ 58u\nu - 19u - 110u^{2}\nu + 96u^{3}\nu - 48u^{4}\nu \\ + 14u^{5}\nu - 2u^{6}\nu + 37u^{2} - 38u^{3} + 22u^{4} - 7u^{5} + u^{6}}\right)}{1+11u+13u^{2}-43u^{3}+43u^{4}-23u^{5}+7u^{6}-u^{7}} ; \\
		\Pi_{\MSS}^{\circ}\left(aa,A\right) &= \frac{u\nu\left(1-\nu\right)\left(\substack{4 + 39u - 58u\nu + 110u^{2}\nu - 96u^{3}\nu + 48u^{4}\nu \\ - 14u^{5}\nu + 2u^{6}\nu - 73u^{2} + 58u^{3} - 26u^{4} + 7u^{5} - u^{6}}\right)}{1+11u+13u^{2}-43u^{3}+43u^{4}-23u^{5}+7u^{6}-u^{7}} ; \\
		\Pi_{\MSS}^{\circ}\left(Aa,A\right) &= \frac{4u\nu\left(1-\nu\right)\left(\substack{3 + 29u\nu - 9u - 55u^{2}\nu + 48u^{3}\nu \\ - 24u^{4}\nu + 7u^{5}\nu - u^{6}\nu + 10u^{2} - 5u^{3} + u^{4}}\right)}{1+11u+13u^{2}-43u^{3}+43u^{4}-23u^{5}+7u^{6}-u^{7}} ; \\
		\Pi_{\MSS}^{\circ}\left(AA,A\right) &= \frac{\nu\left(\substack{1 - 2u^{7}\nu^{2} + u^{7}\nu + 14u^{6}\nu^{2} - 7u^{6}\nu - 48u^{5}\nu^{2} + 26u^{5}\nu \\ - u^{5} + 96u^{4}\nu^{2} - 58u^{4}\nu + 5u^{4} - 110u^{3}\nu^{2} + 77u^{3}\nu \\ - 10u^{3} + 58u^{2}\nu^{2} - 55u^{2}\nu + 10u^{2} + 16u\nu - 5u}\right)}{1+11u+13u^{2}-43u^{3}+43u^{4}-23u^{5}+7u^{6}-u^{7}} .
	\end{align}
\end{subequations}
Notably, this distribution is independent of both $N_{\F}$ and $N_{\M}$.

The stationary distribution for the full evolutionary Markov chain can be obtained from $\Pi_{\MSS}^\circ$ using the fact that, if $\Pi$ is stationary for $\mathbf{P}=\mathbf{RS}$, then, defining $\pi \coloneqq\mathbf{R}^\intercal \Pi$ we have $\pi^{\intercal}\mathbf{SR}= \Pi^{\intercal}\mathbf{RSR}=\Pi^{\intercal}\mathbf{R}=\pi^{\intercal}$. Applying this idea, we obtain the neutral RMC distribution for the full chain:
\begin{align}
	2\left(16+n-3\cdot 2^{-N+2}\right)\pi_{\RMC}^{\circ}\left(\mathbf{x}\right) &= 4\delta_{n_{aa},0}\,\delta_{n_{AA},0}\left(\delta_{n_{a},0} + \delta_{n_{A},0}\right) \nonumber \\
	&\quad +3\cdot 2^{-N+2}\left(\delta_{n_{aa},0} + \delta_{n_{AA},0}\right) \nonumber \\
	&\quad + 2^{n_{Aa}} \delta_{n_{AA},0}\left(\delta_{n_{Aa},0}\,\delta_{n_{A},1}+\delta_{n_{Aa},1}\,\delta_{n_{A},0}\right) \nonumber \\
	\label{eq:RMChaplo}
	&\quad + 2^{n_{Aa}} \delta_{n_{aa},0}\left(\delta_{n_{Aa},0}\,\delta_{n_{a},1}+\delta_{n_{Aa},1}\,\delta_{n_{a},0}\right).
\end{align}
Above, $\delta_{m,m'}$ is the Kronecker symbol, equal to 1 if $m=m'$ and 0 otherwise.

In contrast, the parental chain has a much simpler RMC distribution:
\begin{subequations}
	\label{eq:parentalRMC}
	\begin{align}
		\Pi_{\RMC}^{\circ}\left(Aa,a\right)  & = \Pi_{\RMC}^{\circ}\left(Aa,A\right) = \frac{3}{8} ; \\ 
		\Pi_{\RMC}^{\circ}\left(AA,a\right) & =  \Pi_{\RMC}^{\circ}\left(aa,A\right) = \frac{1}{8} .
	\end{align}
\end{subequations}

\subsubsection{Condition for success}
Combining Eqs.~\eqref{eq:delhatselhaplo} and \eqref{eq:RMChaplo}, it follows from Corollary \ref{cor:delhatselweak} that
\begin{align}
	\rho_{A} > \rho_{a} &\iff \mathbb{E}_{\RMC}^{\circ}\left[\delhatsel '\right] > 0 \nonumber \\
	&\iff 8\frac{N_{\F}}{N_{\F}-1} m + 5\frac{N_{\M}}{N_{\M}-1} s > 0 .
\end{align}
Interestingly, this condition is independent of the degree $h$ of genetic dominance in females. Moreover, for large $N_{\F}$ and $N_{\M}$, we have the approximate condition
\begin{align}
	\label{eq:haplolargeN}
	\rho_{A} > \rho_{a} &\iff 8m + 5s > 0 .
\end{align}
Condition \eqref{eq:haplolargeN} becomes exact under the limit ordering (i) $\delta \to 0$, (ii) $N_M \to \infty, N_F \to \infty$ (i.e.~the $wN$ limit \emph{sensu} \citealp{jeong2014optional,sample2017limits}).

It is tempting to think that one could more easily obtain Condition \eqref{eq:haplolargeN} by first defining an analogue of $\delhatsel$ on the parental chain, and then averaging this analogue over the RMC distribution on the parental chain, as given by Eq.~\eqref{eq:parentalRMC}. This scheme turns out not to work because it does not account for mutations correctly. Since different numbers of male and female offspring are produced each generation, and each offspring provides an independent opportunity for mutation, the mutant appearance distribution in parents differs from the expression given in Eqs.~\eqref{eq:mutA}--\eqref{eq:muta}; this effect is not captured by the above scheme.

\section{Discussion} \label{sec:discussion}

\subsection{Summary}
\subsubsection{Generality and abstraction}
The aim of this work is to provide a mathematical formalism for describing natural selection that is general enough to encompass a wide variety of biological scenarios and modeling approaches. Dealing in such generality requires ``abstracting away" the details of particular models while preserving what is ultimately relevant for natural selection. To this end, we use the replacement rule (introduced by \citealp{allen2014measures}) to represent birth, death, and inheritance, and we use the formalism of genetic sites to allow for different genetic systems (haploid, diploid, haplodiploid, or polyploid). This level of abstraction, although atypical for evolutionary theory, entails a number of advantages: (i) it provides a common language for natural selection with formal definitions of key concepts; (ii) it enables proofs of general theorems that eliminate the duplicate work of deriving analogous results one model at a time; (iii) it may help distinguish robust theoretical principles from artifacts of particular modeling assumptions.

\subsubsection{Main results}
Our main results, Theorem \ref{thm:rhodelsel} and Theorem \ref{thm:delhatselweak}, show that various criteria for success under natural selection become equivalent in the limit of low mutation. Theorem \ref{thm:rhodelsel}, which shows the equivalence of success criteria based on fixation probability, stationary frequency, and change due to selection, is quite general: it applies under arbitrary strength of selection and requires no assumptions beyond the basic setup of our formalism. However, these criteria may all be intractable for a model of reasonable complexity \citep{ibsen2015computational}.

Our second main result, Theorem \ref{thm:delhatselweak}, applies to weak selection. The advantage of Theorem \ref{thm:delhatselweak} is that two of these criteria involve expectations under neutral drift, for which the recurrence relations for the MSS and RMC distributions simplify considerably. Moreover, if the additional assumptions of Corollary \ref{cor:delhatselweak} hold, the direction of selection is completely determined by the sign of $\E^\circ_\RMC\left[\delhatsel'\right]$ or $\frac{d}{du}\E^\circ_\MSS\left[\delhatsel' \right]\big|_{u=0}$. Evaluating these criteria may not require knowing the full RMC or MSS distributions, but only certain statistics of them. For the example of games on graphs, one only needs to obtain the pairwise expectations, $\E_\RMC^\circ\left[\underline{x}_g\underline{x}_h\right]$, which obey their own recurrence relations \eqref{eq:RMCrecurBD} and \eqref{eq:RMCrecurDB}. In contrast, the other success criteria, 
$\rho_A > \rho_a$ and $\lim_{u\to 0} \E_\MSS\left[x\right] > \nu$, are of clear biological relevance but are difficult to analyze directly.

To be clear, neither of our main results comes as a surprise. Theorem \ref{thm:rhodelsel} generalizes the main result of \cite{allen2014measures}, and special cases of this result are also discussed by \cite{RoussetBilliard}, \cite{TaylorFixProb}, \cite{Eusociality} and \cite{tarnita2014measures}. Theorem \ref{thm:delhatselweak} generalizes and extends the main result of \cite{tarnita2014measures}. Aspects and instances of Corollary \ref{cor:delhatselweak}---which is a special case of Theorem \ref{thm:delhatselweak}---have been obtained in so many contexts that this result might be called a ``folk theorem" \citep{TaylorHow,RoussetBilliard,leturque2002dispersal,NowakFinite,Ohtsuki,LessardFixation,Taylor,AntalPhenotype,chen2013sharp,cox2013voter,wakano2013mathematical,debarre2014social,durrett2014spatial,van2015social,allen2017evolutionary}. The contribution of Theorems \ref{thm:rhodelsel} and \ref{thm:delhatselweak} is to formalize and prove these results in a general setting, and elucidate the assumptions on which upon which they depend. For example, the equivalence $\E^\circ_\RMC\left[\delhatsel '\right]>0 \Leftrightarrow \rho_A > \rho_a$, under weak selection, relies on the addition assumptions of Corollary \ref{cor:delhatselweak}, and is not valid in the more general setting of Theorem \ref{thm:delhatselweak} (see also \citealp{tarnita2014measures}).

\subsubsection{Examples}
Our two examples illustrate how our results (in particular, Corollary \ref{cor:delhatselweak}) can be applied  to obtain conditions for success in models with different spatial and genetic structures. The games on graphs example (Section \ref{sec:games}) recovers the main results of \cite{allen2017evolutionary}, i.e.~the conditions for success under weak selection in two-player games on arbitrary weighted graphs. Notably, while \cite{allen2017evolutionary} relied on sophisticated results regarding perturbations of voter models \citep{chen2013sharp} along with the Structure Coefficient Theorem \citep{Corina}, the present analysis uses only the results of this work. This fact underscores the fact that Corollary \ref{cor:delhatselweak} provides not only an equivalence result, but also a problem-solving methodology general enough to handle reasonably complicated models.

For the haplodiploid model, we have found that a mutation with selection coefficient $m$ in males and $s$ in females is favored under weak selection in large populations if $8m+5s>0$, regardless of the degree of dominance, $h$. This result suggests, intriguingly, that a mutation's effect in males is $5/8$ as important as its effect in females. The applicability of this particular result is limited by the rather artificial assumption that each generation is produced by a single mating pair. However, there is no \emph{a priori} reason the analysis could not be generalized to populations with larger numbers of reproducers.

\subsection{Conceptual issues}
\subsubsection{Gene's-eye view}
Identifying the salient units or levels at which selection operates is a longstanding conversation in evolutionary theory \citep{williams1966adaptation,lewontin1970units,hull1980individuality,gould1999individuality,okasha2006evolution,akccay2016there}. This work employs a ``gene's-eye view" \citep{williams1966adaptation}, in that the analysis is conducted almost exclusively at the level of genetic sites. In particular, the criteria for success in Theorems \ref{thm:rhodelsel} and \ref{thm:delhatselweak} are expressed in terms of gene-level quantities. Under Assumptions \ref{ass:ind} and \ref{ass:fairmeiosis}, these criteria may be rewritten in terms of individual-level quantities via Proposition \ref{prop:delselind}. However, with meiotic drive and/or frequent horizontal gene transfer, these assumptions are violated, and an analysis based solely on individual-level quantities would not accurately describe selection.

Our gene-centric perspective does not preclude the possibility that selection may also operate at other levels, including the individual and the group or colony. However, our analysis makes clear that all effects of individual-level or group-level selection can be analyzed using of gene-level quantities (e.g. $\delhatsel$). Analysis of individual-level or group-level quantities is not mathematically necessary, although it may be conceptually illuminating.

\subsubsection{Reproductive value}
Reproductive value has been a central concept in evolutionary theory since \cite{fisher1930genetical}. Recent work \citep{maciejewski2014reproductive} has extended this notion to populations with heterogeneous spatial structure. The key property of reproductive value---that RV-weighted frequency has zero expected change under neutral drift---is encapsulated in our Theorem \ref{thm:delhatselneut}. A closely related result is that the reproductive value of a site is proportional to the fixation probability of a neutral mutation arising at this site (Theorem \ref{thm:neutfixprob}; see also \citealp{maciejewski2014reproductive,allen2015molecular}).

Our work underscores the importance of reproductive value, but also reveals an important limitation of this concept. Reproductive value can only be defined with reference to some particular process of neutral drift. In Section \ref{sec:RV}, we defined reproductive value with respect to replacement probabilities in the monoallelic states $\va$ and $\vA$. However, without further assumptions, there is no guarantee that these two states will yield the same reproductive value for each genetic site. It may therefore be impossible to assign consistent reproductive values in a given model. This issue may be exacerbated if one considers models with more than two alleles (see Section \ref{sec:twoalleles}), in which each new selective sweep may alter the reproductive values of sites, thereby affecting selection pressure on subsequent mutations. Here, we obtained consistent reproductive values either by invoking Assumption \ref{ass:neutral}, or (in the context of weak selection) by using $\delta=0$ as a reference process of neutral drift. For an arbitrary mathematical model or biological population, there is no guarantee that well-defined reproductive values exist.

\subsubsection{Fitness}
Like reproductive value, the concept of fitness is both fundamental to evolutionary theory and difficult to formalize in a fully consistent way \citep{MetzFitness,akccay2016there,lehmann2016invasion,metz2016frequency,doebeli2017towards}. Here, we have defined fitness as the RV-weighted contribution of a genetic site to the next time-step, which can be decomposed as the sum of the RV-weighted survival probability, $v_g - \hat{d}_g\left(\vx\right) = v_g\left(1-d_g\left(\vx\right)\right)$, and the RV-weighted birth rate, $\hat{b}_g\left(\vx\right)$. This definition is consistent with other definitions of fitness in the literature (e.g. \citealp{tarnita2014measures}). However, we do not intend this definition to be universal or canonical, for two reasons. First, it depends on the existence of well-defined reproductive values, which is not guaranteed if Assumption \ref{ass:neutral} is violated. Second, it quantifies only the one-step contribution of a site in a given state, which may not accurately capture the long-term success of the progeny of this site. In general, there may not be any fully satisfactory definition of fitness of a single genetic site (let alone an individual). Indeed, \cite{akccay2016there} argue that fitness should instead be ascribed to a \emph{genetic lineage}---that is, the progeny of a given allele copy. An interesting direction for future work would be to describe the lineage-eye view of fitness in the context of our formalism.

\subsection{Limitations and possible extensions}\label{sec:extensions}
Although we have aimed for a relatively high level of generality, our formalism still entails a number of limiting assumptions. Here, we discuss these assumptions and the prospects for extending beyond them.

\subsubsection{Fixed population size}
Fluctuations in population size can affect the process of natural selection \citep{lambert2006probability,parsons2007fixation,parsons2007fixationII,wakano2009spatial,parsons2010some,schoener2011newest,uecker2011fixation,waxman2011unified}. These effects cannot be studied in the current framework, which assumes constant population size. In some cases, fluctuations in population size can be safely ignored---for example, if selection depends only on the \emph{relative} frequencies of competing types (as in the classical population genetics setting of \citealp{crow1970introduction,burger2000mathematical,ewens2004mathematical}), or if the population is assumed to remain close to its carrying capacity \citep{burger2005multilocus}. In other cases, however, population dynamics have important consequences for long-term evolution \citep{Dieckmann,Metz,pelletier2007evolutionary,DercoleBook,wakano2009spatial,schoener2011newest,korolev2013fate,constable2016demographic,chotibut2017population}, including the possibility of evolutionary branching \citep{Geritz,DieckmannSpeciation} or evolutionary suicide \citep{GyllenbergSuicide}. 

Technical complications arise when mathematically modeling populations of fluctuating size. For example, it may be possible for the entire population to become extinct, and as a result, there may be no non-trivial stationary distribution. An alternative is to consider the quasi-stationary distribution for the process \citep{haccou:CUP:2005,GyllenbergQuasi,CattiauxQuasi,faure:AAP:2014}, which is conditioned on non-extinction of the population. Quantifying the evolutionary success of an allele also becomes more nuanced when population size can vary \citep{lambert2006probability,parsons2010some,constable2016demographic,mcavoy:TPB:2018}. 

To extend the current framework to populations of changing size would require the set of sites, $G$, to itself be an aspect of the population state. The offspring-to-parent map, $\alpha$, would then map the set of sites at time $t+1$ to the set of sites at time $t$. Such complications are not necessarily insurmountable but would add considerable notational and technical burden to our formalism.

\subsubsection{Fixed spatial structure and environment, trivial demography}
In our framework, the probability of a replacement event depends only on the current population state, $\vx$. This precludes the possibility that the population structure could change over time \citep{PachecoCoevolution,PachecoLinking,AntalPhenotype,CorinaSet,WuDynamicalNetworks,cavaliere2012prosperity,wardil2014origin}, or that replacement events could depend on variable aspects of the environment \citep{cohen1966optimizing,philippi1989hedging,HaccouEstablishment,kussell2005phenotypic,starrfelt2012bet,cvijovic2015fate} or on the demographic stages of individuals (e.g. juvenile versus adult; \citealp{DiekmannLinear,DiekmannNonlinear,parvinen2016fitness,Lessard2018}). Eco-evolutionary feedbacks between a population and its demography and environment can be important drivers of evolutionary change \citep{Dieckmann,Metz,Geritz,DercoleBook,DurinxPhysiological,PercCoevolutionary}, but such phenomena lie outside the scope of the formalism presented here.

These limitations can be overcome by allowing the probabilities of replacement events to depend on additional variables beyond the population state $\vx$. These variables can represent the current population structure, environmental conditions, and/or the demographic stages of individuals. Additional rules must then be provided for updating these variables. In such an extension, natural selection is described by a Markov chain with state space $\left\{0,1\right\}^G \times E$, where $E$ is the set of possible joint values for these extra variables. This extension would bring phenomena such as bet-hedging \citep{cohen1966optimizing,philippi1989hedging,kussell2005phenotypic,starrfelt2012bet} and structure-strategy coevolution \citep{PachecoCoevolution,CorinaSet,PercCoevolutionary} into the purview of our formalism. However,  the incorporation of this additional information may introduce significant difficulties in establishing basic results such as the equivalence of success criteria.

\subsubsection{Two alleles} \label{sec:twoalleles}
We have focused here on the case of two competing alleles; however, generalizing to any number of alleles appears straightforward. In the limit of rare mutation, the outcome of natural selection depends only on the pairwise competitions between alleles. We therefore anticipate that the RMC distribution (generalized for more than two alleles) will be nonzero only for states containing exactly two alleles. Moreover, the results of \cite{fudenberg:JET:2006} imply that the equilibrium mutation-selection balance will be obtainable from the pairwise fixation probabilities of each allele into each other. On the other hand, away from the limit of rare mutation, determining which of multiple alleles is favored by selection is a more nuanced question \citep{antal2009mutation,traulsen2009exploration,CorinaMultiple,wu2012small}.

\subsubsection{One locus}

Although our formalism allows for arbitrary genetics in the sense of ploidy and mating structure, we have focused on selection at a single genetic locus. Extending to multilocus genetics would require considering additional genetic sites---one for each locus on each chromosome. It would then be natural to impose assumptions on the replacement rule so that replacement can only occur between sites at the same locus. The replacement rule could then encode an arbitrary pattern of linkage and recombination. One complication is that quantifying success under natural selection becomes subtler in the context of multiple linked loci \citep{hammerstein1996darwinian,eshel1998long,lehmann2009perturbation}.

\subsubsection{Independence and rarity of mutations}
Although we have generalized the framework of \cite{allen2014measures} to include arbitrary mutational bias, $\nu$, we still assume that mutations occur independently with a constant probability per offspring. This assumption may be violated in natural systems; for example, an adult may acquire a germline mutation that is passed to all offspring. Our formalism can be generalized to accommodate such cases by allowing for adult mutation, or by relaxing the assumption that offspring mutations are independent. One could also allow for the mutation parameters $u$ and $\nu$ to depend on the parental site and/or the population state. The effect of such amendments to our framework would be to alter the mutant appearance distribution (which in turn affects fixation probabilities) as well as the RMC distribution.

Additionally, most of our results assume that mutation is either absent or rare. With nonvanishing mutation, even defining which type is favored under natural selection is a nontrivial question, since distinct, intuitive criteria for success may disagree with each other \citep{allen2014measures}. Furthermore, mutation can alter the direction of selection \citep{traulsen2009exploration}; for example, high mutation rates can impede the evolution of cooperation in spatially-structured populations by diluting the clustering of cooperators \citep{AllenGraphMut,debarre2017fidelity}. General mathematical theorems on natural selection with nonvanishing mutation could shed new light on such results and perhaps uncover new ones.

\subsection{Connections to other approaches}
Beyond the above-mentioned possibilities for generalization, there is significant opportunity to connect our formalism to population genetics and other general approaches in evolutionary theory. For example, one might conceive of a coalescent process \citep{kingman1982coalescent,cox1989coalescing,WakeleyCoalescent} for our formalism, using the neutral replacement probabilities $p_{\left(R,\alpha\right)}^{\circ}$. Likewise, it appears possible to define a notion of identity-by-descent probability \citep{malecot1948IBD,RoussetBilliard,allen2014games} to characterize the extent to which a given pair (or set) of sites are genetically related. More speculatively, one could ask whether diffusion methods \citep{KimuraDiffusion,Chen2018diffusion} could be applied to our formalism in the large-population limit. Developing these connections may provide opportunities to generalize results from population genetics to a wider variety of spatial and genetic structures.

\section{Conclusion} \label{sec:conclusion}
Mathematical modeling of evolution is a robust and growing field, with a rapid pace of new discoveries and rich interplay between theoretical and empirical study. At such moments of expansion, unifying mathematical frameworks are particularly helpful---by clarifying concepts, providing common definitions, and establishing fundamental results. We hope that the formalism presented here will provide a strong foundation for the theory of natural selection in structured populations to build upon.

\begin{longtable}{cp{9cm}c}
	\caption{Glossary of Notation}
	\label{notationtable}\\
	\hline
	Symbol & Description\\
	\hline
	\endhead
	\hline \endfoot
	$\alpha$ & Parentage map in a replacement event & \ref{sec:replacement}\\
	$b\left(\vx\right)$ & Total (population-wide) expected offspring number in state $\vx$ & \ref{sec:demographic}\\
	$b_g\left(\vx\right)$ & Expected offspring number of site $g$ in state $\vx$ & \ref{sec:demographic}\\
	$d_g\left(\vx\right)$ & Death probability of site $g$ in state $\vx$& \ref{sec:demographic}\\
	$\delta$ & Selection strength & \ref{sec:weakseldef}\\
	$\delsel\left(\vx\right)$ & Expected change due to selection in the frequency of allele $A$ & \ref{sec:changesel}\\
	$e_{gh}\left(\vx\right)$ & Marginal probability that allele in site $h$ is replaced by copy of allele in site $g$ in state $\vx$ & \ref{sec:demographic}\\
	$F_{xy}, F_z$ & Fecundity of genotypes $xy$ and $z$ for haplodiploid population; $x,y,z \in \left\{a,A\right\}$ & \ref{sec:haplodiploid}\\
	$f_{xy}$ & Payoff to $x$ interacting with $y$ for games on graphs; $x,y \in \left\{a,A\right\}$ & \ref{sec:games}\\
	$G$ & Set of genetic sites & \ref{sec:sites}\\
	$G_i$ & Set of genetic sites in individual $i$ & \ref{sec:sites}\\
	$I$ & Set of individuals & \ref{sec:sites}\\
	$\cM$ & Evolutionary Markov chain & \ref{sec:Markov}\\
	$n$ & Number of genetic sites & \ref{sec:sites}\\
	$N$ & Number of individuals (population size) & \ref{sec:sites}\\
	$n_i$ & Ploidy (number of genetic sites) of individual $i$ & \ref{sec:sites}\\
	$\nu$ & Mutational bias & \ref{sec:mutation}\\
	$\omega_{gh}$ & Edge weight between vertices $g$ and $h$ for games on graphs & \ref{sec:games}\\
	$R$ & Set of replaced positions in a replacement event & \ref{sec:replacement}\\
	$p_{gh}$ & Probability of stepping from vertex $g$ to vertex $h$ for games on graphs & \ref{sec:games}\\
	$p_{gh}^{\left(m\right)}$ & Probability that $m$-step random walk from $g$ terminates at $h$ for games on graphs & \ref{sec:games}\\
	$\rho_A$, $\rho_a$ & Fixation probabilities of alleles $A$ and $a$, respectively & \ref{sec:fixprob}\\
	$u$ & Mutation probability per reproduction & \ref{sec:mutation}\\
	$v_g$ & Reproductive value of site $g$ & \ref{sec:RV}\\
	$w_g$ & Fitness of site $g$ & \ref{sec:fitness}\\
	$x$ & Frequency of allele $A$ & \ref{sec:demographic}\\
	$\vx$ & Vector of alleles occupying each site; state of $\cM$ & \ref{sec:states}\\
	$x_g$ & Allele occupying site $g \in G$ (0 for $a$, 1 for $A$) & \ref{sec:demographic}\\
	\hline
	$\MSS$ & Mutation-selection stationary distribution & \ref{sec:MSS}\\
	$\RMC$ & Rare-mutation conditional distribution & \ref{sec:RMC}\\
	\hline
	$\hat{}$ & Indicates weighting by reproductive value & \ref{sec:RV}\\
	${}^\circ$ & Indicates the absence of selection & \ref{sec:consistency}\\
	${}'$ & Indicates a quantity to first order in $\delta$ as $\delta \to 0$& \ref{sec:weakseldef}
\end{longtable}

\section*{Acknowledgements}
B.~A.~is supported by National Science Foundation award \#DMS-1715315. A.~M.~is supported by the Office of Naval Research, grant N00014-16-1-2914. We thank Martin A.~Nowak for helpful discussions.

\end{document}